\documentclass{CSML}

\usepackage{lastpage}
\def\dOi{13(2:14)2017}
\lmcsheading%
{\dOi}
{1--\pageref{LastPage}}
{}
{}
{Jan.~30, 2016}
{Jun.~30, 2017}
{}

\usepackage{amsthm}
\usepackage[utf8]{inputenc}
\usepackage{hyperref}\hypersetup{hidelinks}

\usepackage[all]{xy}
\usepackage{amsmath}
\usepackage{amssymb}
\usepackage{amsfonts}

\usepackage[all]{xy}
\usepackage{tikz}
\usetikzlibrary{fit,shapes}
\usetikzlibrary{arrows,positioning,shapes,fit,calc}

\newtheorem{claim}[thm]{Claim}

\newenvironment{tbs}{%
   \small\tt
   \begin{itemize}}{\end{itemize}}
\newcommand{\btbs}{\begin{tbs}}                                                                      
\newcommand{\etbs}{\end{tbs}}

\newcommand{\hide}[1]{}

\renewcommand{\phi}{\varphi}

\newcommand{\isdef}{\mathrel{:=}}

\newcommand{\bbA}{\mathbb{A}}
\newcommand{\bbB}{\mathbb{B}}
\newcommand{\bbS}{\mathbb{S}}
\newcommand{\bbT}{\mathbb{T}}

\newcommand{\fun}{\mathsf{T}}
\newcommand{\psf}{\mathcal{P}}
\newcommand{\cvp}{\mathcal{Q}}
\newcommand{\mon}{\mathcal{M}}

\newcommand{\mathlang}[1]{\mathtt{#1}}
   
\newcommand{\FOE}{\mathlang{FOE}}
\newcommand{\FO}{\mathlang{FO}}
\newcommand{\MSO}{\mathlang{MSO}}
   \newcommand{\MSOT}{\MSO_{\fun}}
   \newcommand{\MSOLa}{\MSO_{\Lambda}}
\newcommand{\muML}{\mathlang{\mu ML}}
   \newcommand{\muMLT}{\muML_{\fun}}
   \newcommand{\muMLLa}{\muML_{\Lambda}}

   \newcommand{\MLLa}{\mathlang{1ML}_{\Lambda}}

   \newcommand{\SOT}{\mathlang{1SO}_{\fun}}
   \newcommand{\SOLa}{\mathlang{1SO}_{\La}}

\newcommand{\Var}{\mathit{Var}}

\newcommand{\isbnf}{\mathrel{::=}}
\newcommand{\divbnf}{\mathrel{|}}

\newcommand{\sr}{\mathsf{sr}}
\newcommand{\Em}{\mathsf{Em}}

\newcommand{\Sing}{\mathsf{Sing}}

\newcommand{\Aut}{\mathit{Aut}}

\newcommand{\smod}{\mathbb{S}}

\newcommand{\si}{\sigma}

\newcommand{\De}{\Delta}
\newcommand{\La}{\Lambda}
\newcommand{\Om}{\Omega}
\newcommand{\gbox}{\underline{\Box}}
\newcommand{\gdi}{\underline{\Diamond}}
\newcommand{\funco}{\underline{\mathsf{T}}}
\renewcommand{\tilde}[1]{\underline{#1}}

\newcommand{\gsim}{\;\underline{\sim}\;}

\begin{document}

\title[Expressive completeness for coalgebraic $\mu$-calculi]{an expressive completeness theorem for coalgebraic modal $\mu$-calculi}

 \author[S.~Enqvist]{Sebastian Enqvist\rsuper a}
 \address{{\lsuper a}Institute for Logic, Language and Computation 
\\ University of Amsterdam
\\ PO Box 94242, 1090 GE Amsterdam,The Netherlands,
\\ and
\\ Department of Philosophy, Stockholm University, SE - 10691 Stockholm, Sweden}
\email{thesebastianenqvist@gmail.com}

 \author[F.~Seifan]{Fatemeh Seifan\rsuper b}
 \address{{\lsuper{b,c}}Institute for Logic, Language and Computation 
\\ University of Amsterdam
\\ PO Box 94242, 1090 GE Amsterdam,The Netherlands}
\email{fateme.sayfan@gmail.com, Y.Venema@uva.nl}
\thanks{Fatemeh Seifan was supported by \emph{Vrije Competitie} grant 
  612.001.115 of the Netherlands Organisation for Scientific Research (NWO)}

 \author[Y.~Venema]{Yde Venema\rsuper c}
 \address{\vspace{-18 pt}}

\maketitle

\begin{abstract}
Generalizing standard monadic second-order logic for Kripke models, we 
introduce monadic second-order logic interpreted
over coalgebras for an arbitrary set functor. 
We then consider invariance under behavioral equivalence of MSO-formulas.
More specifically, we investigate whether the coalgebraic mu-calculus is the 
bisimulation-invariant fragment of the monadic second-order language for a given functor. Using automata-theoretic techniques and building on recent results by the third author, we show that in order to provide
such a characterization result it suffices 
to find what we call an adequate uniform construction for the coalgebraic type functor. 
As direct applications of this result we obtain a partly new proof of the Janin-Walukiewicz 
Theorem for the modal mu-calculus, avoiding the use of syntactic normal forms, and bisimulation invariance results for the bag functor (graded modal
logic) and all exponential polynomial functors (including the ``game functor''). As a more involved application, involving additional non-trivial ideas, we also derive a characterization theorem for the monotone modal mu-calculus, with respect to a natural monadic second-order language for monotone neighborhood models.
\end{abstract}

\section{Introduction}

\subsection{Logic, automata and coalgebra}
The aim of this paper is to strengthen the link between the areas of logic, 
automata and coalgebra.
More precisely, we provide a coalgebraic generalization of the
automata-theoretic approach towards monadic second-order logic ($\MSO$), and
we address the question of whether the Janin-Walukiewicz Theorem can be 
generalized from Kripke structures 
to the setting of arbitrary coalgebras\footnote{The paper is a substantially extended version of a previous manuscript that was presented at the 2015 conference for Logic in Computer Science (LICS 2015) in Kyoto, Japan \cite{enqv:mona15}.}. 

The connection between \textit{monadic second-order logic} and \textit{automata}
is classic, going back to the seminal work of B\"uchi, Rabin, and others.
Rabin's decidability result for the monadic second-order theory
of binary trees, or $S2 S$, makes use of a translation of 
monadic second-order logic into a class of automata, thus reducing the 
satisfiability problem for $S 2 S$ to the non-emptiness problem for the
corresponding automata \cite{rabi:deci69}.
The link between $\MSO$ and automata over trees with arbitrary branching
was further explored by Walukiewicz~\cite{walu:mona96}.
Janin and Walukiewicz considered monadic second-order logic interpreted over 
Kripke structures, and used automata-theoretic techniques to obtain a van 
Benthem-like characterization theorem for monadic second-order logic, identifying
the modal $\mu$-calculus as the bisimulation invariant fragment of 
$\MSO$ \cite{jani:expr96}.
Given the fact that in many applications bisimilar models are considered
to represent the \emph{same} process, one has little interest in properties of 
models that are \emph{not} bisimulation invariant.
Thus the Janin-Walukiewicz Theorem can be seen as an expressive completeness 
result, stating that all \emph{relevant} properties in monadic second-order 
logic can be expressed in the modal $\mu$-calculus.

Coalgebra enters naturally into this picture. 
Recall that Universal Coalgebra~\cite{rutt:univ00} provides the notion of a 
\emph{coalgebra} as the natural mathematical generalization of state-based 
evolving systems such as streams, (infinite) trees, Kripke models, 
(probabilistic) transition systems, and many others.
This approach combines simplicity with generality and wide applicability: many
features, including input, output, nondeterminism, probability, and interaction, 
can easily be encoded in the coalgebra type $\fun$ (formally an endofunctor on 
the category $\mathbf{Set}$ of sets as objects with functions as arrows).
Starting with Moss' seminal paper~\cite{moss:coal99}, coalgebraic logics have 
been developed for the purpose of specifying and reasoning about \emph{behavior},
one of the most fundamental concepts that allows for a natural coalgebraic 
formalization.
And with Kripke structures constituting key examples of coalgebras, it should
come as no surprise that most coalgebraic logics are some kind of modification
or generalization of \emph{modal logic}~\cite{cirs:moda11}.

The coalgebraic modal logics that we consider here originate with
Pattinson~\cite{patt:coal03}; they are characterized by a completely standard
syntax, in which the semantics of each modality is determined by a so-called 
\emph{predicate lifting} (see Definition~\ref{d:pl} below).
Many well-known variations of modal logic in fact arise as the coalgebraic logic
 associated with a set $\La$ of such predicate liftings; examples
include both standard and (monotone) neighborhood modal logic, graded and 
probabilistic modal logic, coalition logic, and conditional logic.
Extensions of coalgebraic modal logics with fixpoint operators, needed 
for describing \emph{ongoing} behavior, were developed 
in~\cite{vene:auto06,cirs:expt09}.

The link between coalgebra and automata theory is by now well-established.
For instance, finite state automata operating on finite words have been 
recognized as key examples of coalgebra from the 
outset~\cite{rutt:univ00}.
More relevant for the purpose of this paper is the link with precisely the
kind of automata mentioned earlier, since the (potentially infinite) objects 
on which these devices operate, such as streams, trees and Kripke frames,
usually are coalgebras.
Thus, the automata-theoretic perspective on modal fixpoint logic may be
lifted to the abstraction level of coalgebra~\cite{vene:auto06,font:auto10}.
In fact, many key results in the theory of automata operating on infinite 
objects, such as Muller \& Schupp's Simulation Theorem~\cite{mull:simu95} can
in fact be seen as instances of more general theorems in Universal 
Coalgebra~\cite{kupk:coal08}.


\subsection{Coalgebraic monadic second-order logic}

Missing from this picture is, to start with, a coalgebraic version of 
\emph{(monadic) second-order logic}.
Filling this gap is the first aim of the current paper, which introduces a notion
of \emph{monadic second-order logic} $\MSOT$ for coalgebras of type $\fun$.
Our formalism combines two ideas from the literature.
First of all, we looked for inspiration to the coalgebraic versions of 
\emph{first-order logic} of Litak \& alii~\cite{lita:coal12}.
These authors introduced Coalgebraic Predicate Logic as a common generalisation 
of first-order logic and coalgebraic modal logic, combining first-order 
quantification with coalgebraic syntax based on predicate liftings.
Our formalism $\MSOT$ will combine a similar syntactic feature with second-order
quantification. In fact it can be shown to contain the Coalgebraic Predicate Logic of Litak \& alii as a fragment, though we will omit the verification of this fact. 

Second, following the tradition in automata-theoretic approaches towards monadic
second-order logic, our formalism will be \emph{one-sorted}.
That is, we \emph{only} allow second-order quantification in our language, 
relying on the fact that individual quantification, when called for, can be
encoded as second-order quantification relativized to singleton sets.
Since predicate liftings are defined as families of maps on powerset
algebras, these two ideas fit together very well, to the effect that our
second-order logic is in some sense simpler than the first-order 
formalism of~\cite{lita:coal12}.

In section~\ref{sec:mso} we will define, for any set $\Lambda$ of 
monotone\footnote{%
  In the most general case, restricting to monotone predicate liftings is not 
  needed, one could define $\MSOT$ as the logic obtained by taking for $\La$
  the set of \emph{all} predicate liftings. 
  However, in the context of this paper, where we take an automata-theoretic
  perspective on $\MSO$, this restriction makes sense.
  }
predicate liftings, a formalism $\MSOLa$, and we let $\MSOT$ denote the logic
obtained by taking for $\Lambda$ the set of \emph{all} monotone predicate 
liftings.
Clearly we will make sure that this definition generalizes the standard case,
in the sense that the standard version of $\MSO$ for Kripke structures 
instantiates the logic $\MSO_{\{\Diamond\}}$ and is equivalent to the coalgebraic
logic $\MSO_{\psf}$ (where $\psf$ denotes the power set functor). 

The introduction of a monadic second-order logic $\MSOT$ for $\fun$-coalgebras 
naturally raises the question, for which $\fun$ does the coalgebraic modal
$\mu$-calculus for $\fun$, denoted $\muMLT$, correspond to the bisimulation-invariant fragment of 
$\MSOT$.

\begin{qu}
\label{q:Q}
\text{Which functors $\fun$ satisfy $\muMLT \equiv \MSOT/{\sim}$?}
\end{qu}

\subsection{Automata for coalgebraic monadic second-order logic}

In order to address Question~\ref{q:Q}, we take an automata-theoretic 
perspective on the logics $\MSOT$ and $\muMLT$, and as the second contribution
of this paper we introduce a class of parity automata for $\MSOT$.

As usual, the operational semantics of our automata is given in terms
of a two-player acceptance game, which proceeds in \emph{rounds} moving from
one basic position to another, where a basic position is a pair consisting of
a state of the automaton and a point in the coalgebra structure under 
consideration.
In each round, the two players, $\exists$ and $\forall$, focus on a certain
local `window' on the coalgebra structure.
This `window' takes the shape of a \emph{one-step $\fun$-model}, that is, a 
triple $(X,\alpha,V)$ consisting of a set $X$, a \emph{chosen object} $\alpha
\in \fun X$ (representing partial information about the transition map of the coalgebra as seen from the perspective of some given state), and 
a valuation $V$ interpreting the states of the automaton as subsets of $X$.
More specifically, during each round of the game it is the task of $\exists$
to come up with a valuation $V$ that creates a one-step model in which a 
certain \emph{one-step formula} $\delta$ (determined by the current basic 
position in the game) is true.

Generally, our automata will have the shape $\bbA = (A,\De,\Om,a_{I})$
where $A$ is a finite carrier set with initial state $a_{I} \in A$, and $\Om$
and $\De$ are the parity and transition map of $\bbA$, respectively.
The flavour of such an automaton is largely determined by the co-domain of its
transition map $\De$, the so-called \emph{one-step language} which consists of 
the one-step formulas that feature in the acceptance game as 
described.

Each one-step language $\mathlang{L}$ induces its own class of automata 
$\Aut(\mathlang{L})$.
For instance, the class of automata corresponding to the coalgebraic fixpoint
logic $\muMLLa$ can be given as $\Aut(\mathtt{ML}_{\La})$, where $\mathtt{ML}_{\La}$ is the set
of positive modal formulas of depth one that use modalities from 
$\La$~\cite{font:auto10}.
Basically then, the problem of finding the right class of automata for the
coalgebraic monadic second-order logic $\MSOLa$ consists in the identification
of an appropriate one-step language.
Our proposal comprises a one-step \emph{second-order logic} $\mathtt{SO}_\La$ which uses predicate 
liftings to describe the chosen object of the one-step model.

Finally, note that similar to the case of standard $\MSO$, the equivalence 
between formulas in $\MSOT$ and automata in $\Aut(\mathtt{SO}_\La)$ is only guaranteed to 
hold for coalgebras that are `tree-like' in some sense (to be defined further 
on).

\begin{thm}[Automata for coalgebraic $\MSO$]
\label{t:automatachar}
For any set $\La$ of monotone predicate liftings for $\fun$ there is an
effective construction mapping any formula $\varphi\in \MSOLa$ into an
automaton $\mathbb{A}_\varphi \in \Aut(\mathtt{SO}_\La)$, 
which is
equivalent to $\varphi$ over $\fun$-tree models.
\end{thm}

The proof of Theorem~\ref{t:automatachar} proceeds by induction on the 
complexity of $\MSOT$-formulas, and thus involves various \emph{closure 
properties} of automata, such as closure under complementation, union and
projection.
In order to establish these results, it will be convenient to take an
\emph{abstract} perspective revealing how closure properties of a class of
automata are completely determined at the level of the one-step language, similar to the development in \cite{walu:mona96}. 

\subsection{Bisimulation Invariance}

With automata-theoretic characterizations in place for both coalgebraic $\MSO$ 
and the coalgebraic $\mu$-calculus $\muML$, we can address Question~\ref{q:Q}
by considering the following question:
\begin{qu}
\label{q:Qa}
\text{Which functors $\fun$ satisfy $\Aut(\mathtt{ML}_\fun)\equiv \Aut(\mathtt{SO}_\fun)/{\sim}$?}
\end{qu}
Continuing the program of the third author~\cite{vene:expr14}, we will approach
this question \emph{at the level of the one-step languages}, $\mathtt{SO}$ and $\mathtt{ML}$.
To start with, observe that any translation (from one-step formulas in) $\mathtt{SO}$ to
(one-step formulas in) $\mathtt{ML}$ naturally induces a translation from $\mathtt{SO}$-automata
to $\mathtt{ML}$-automata.
A new observation we make here is that any so-called \emph{uniform construction}
on the class of one-step models for the functor $\fun$ that satisfies certain
\emph{adequacy} conditions, provides
(1) a translation $(\cdot)^{*}: \mathtt{SO} \to \mathtt{ML}$, together with
(2) a construction $(\cdot)_{*}$ transforming a pointed $\fun$-model $(\bbS,s)$
into a tree model $(\bbS_{*},s_{*})$ which is a coalgebraic pre-image of 
$(\bbS,s)$ satisfying
\[
\bbA \text{ accepts } (\bbS_{*},s_{*}) \text{ iff }
\bbA^{*} \text{ accepts } (\bbS,s).
\]
where $\bbA^*$ is obtained by applying to $\bbA$ the lifting of the one-step translation $(\cdot)^{*}: \mathtt{SO} \to \mathtt{ML}$ to the level of automata.
From this it easily follows that an $\mathtt{SO}$-automaton $\bbA$ is bisimulation 
invariant iff it is equivalent to the $\mathtt{ML}$-automaton $\bbA^{*}$.

On the basis of these observations we can prove the following generalisation of 
the Janin-Walukiewicz Theorem.

\begin{thm}[Coalgebraic Bisimulation Invariance]
\label{mainone}
Let $\fun$ be any set functor. If $\fun$ admits an adequate uniform construction for every finite set of second-order one-step formulas $\Gamma$, then 
$$\mathtt{MSO}_\fun/{\sim} \equiv \mu \mathtt{ML}_\fun.$$
\end{thm}

In our view, the proof of this theorem separates the `clean', abstract part of 
bisimulation-invariance results from the more functor-specific parts.
As a consequence, Theorem~\ref{mainone} can be used to obtain immediate results
in particular cases.
Examples include the power set functor (standard Kripke structures),
where the adequate uniform construction roughly consists of taking $\omega$-fold 
products (see Example~\ref{ex:psf}), the bag functor (Example~\ref{ex:bag}),
and all exponential polynomial functors (Corollary~\ref{c:epf}).

However, it turns out that there are cases of functors that provably do \emph{not} admit an adequate uniform
construction. A concrete example is the \emph{monotone neighborhood functor} 
$\mon$ (see the next section for its definition), which provides the semantics for \textit{monotone modal logic} ~\cite{hans:coal04} and its fixpoint extension, the \textit{monotone $\mu$-calculus}. The monotone $\mu$-calculus is interesting for at least two reasons: first, it has a purely theoretical interest as a natural candidate for a \emph{minimal} modal $\mu$-calculus, since the only constraint on the box operator is monotonicity rather than the stronger constraint of complete multiplicativity of the box operator in normal modal logic. In other words, the base modal logic has precisely the logical strength it needs so that Knaster-Tarski-style extremal fixpoints may be smoothly introduced, and no more. Second, a more application driven reason to study the monotone $\mu$-calculus is that it is the natural fixpoint logic in which Parikh's dynamic logic of games \cite{pari:game85} sits as a fragment  in the same manner that PDL is a fragment of the modal $\mu$-calculus. (More precisely, game logic is a fragment of a multi-modal variant of the monotone modal $\mu$-calculus, but we will stick to the monomodal case for the sake of notational convenience.)   
As the final contribution of this paper we show how, with some extra work, a characterization result for the monotone $\mu$-calculus can be derived using our main result \ref{mainone}, with respect to a natural monadic second-order language for monotone neighborhood structures that we call ``monotone MSO''.  
\begin{thm}
\label{t:mon-char}
The monotone $\mu$-calculus is the fragment of monotone MSO that is invariant for neighborhood bisimulations.  
\end{thm}

\section{Some technical background}

In this section we provide some of the necessary technical background.
We shall assume familiarity with coalgebras, and at least some basic acquaintance with modal logic and the modal $\mu$-calculus. We will provide insofar as possible a self-contained presentation of the required material on automata and infinite games.

\subsection{Automata}
\label{s:automata}

Throughout most of its history, the mathematical theory surrounding the modal $\mu$-calculus has relied heavily on the theory of \emph{automata} operating on infinite objects and,  what is closely related, \emph{infinite games}. It turns out that these frameworks fit nicely into the coalgebraic world view, which is to be expected as infinite games and automata operating on infinite objects both have a coinductive flavour. This connection is in fact a big part of what makes the coalgebraic modal $\mu$-calculus tick, and will be central to this paper. In this section we give a brief overview of this field. The material that we cover is entirely standard, so readers who are already familiar with these topics will want to merely skim the section to fix notation and terminology.

We start with something that probably most readers have already come across at some point: the theory of word automata recognizing regular languages over some alphabet.  (For an introduction to the basic concepts of automata theory, see \cite{hopc:intro01}.)
\begin{defi}
 A (non-deterministic) \emph{word automaton} is a tuple:
$$(A,\Delta,a_I,F)$$
where $A$ is a finite set, $\Delta : A \times \Sigma \to \psf(A)$, $a_I \in A$ and $F \subseteq A$.
\end{defi}
The general mold into which many types of automata fit is of this shape: they consist of a set of \emph{states} ($A$), a \emph{transition structure } ($\Delta$) telling us how to move between states of the automaton, an \emph{initial state} ($a_I$) and, finally, some data structure over $A$ that lets us know whether or not a given computation carried out by an automaton accepts or rejects its input. In the case of word automata, the input $\vec{w}$ of an automaton is a finite word $c_0....c_n$ over the alphabet $\Sigma$ (i.e. an element of $\Sigma^n$ for some $n$). 
\begin{defi}
A \emph{run} of the automaton $\mathbb{A}$ on a word $\vec{w}$ is a tuple $a_0....a_{n+1}$ of states in $A$, such that $a_0 = a_I$ and $a_{i + 1} \in \Delta(a_i,c_i)$ for $i \leq n$. The run is said to be \emph{accepting} if $a_{n+1} \in F$, and we say that $\mathbb{A}$ \emph{accepts} the word $\vec{w}$ if there is some accepting run of $\mathbb{A}$ on the input $\vec{w}$. The \emph{language} $L(\mathbb{A})$ of an automaton $\mathbb{A}$ is simply the collection of all the words that it accepts. 
\end{defi}
A central technique in automata theory is to show that one can simplify the transition structure of some class of automata, without losing any expressive power. In the case of word automata, it is well known that the \emph{non-deterministic} automata that we have defined here have the same expressive power as \emph{deterministic} word automata, which are defined as ones that satisfy the condition that $\Delta(a,c)$ is a \emph{singleton} for each $a \in A$. This means that there is a \emph{unique} run of the automaton on every input. 
\begin{fact}[Simulation of word automata]
For every word automaton $\mathbb{A}$, there exists a deterministic word automaton $\mathbb{A}'$ with $L(\mathbb{A}) = L(\mathbb{A}')$.
\end{fact}
The proof consists of an elegant construction known as the \emph{powerset construction}: given an automaton $\mathbb{A} = (A,\Delta,a_I,F)$, we construct a deterministic automaton $\mathbb{A}' = (A',\Delta',a_I',F')$ recognizing the same language by setting $A' = \psf(A)$, $a_I' = \{a_I\}$, the accepting states are defined by
 $$F' = \{B \subseteq A \mid B \cap F \neq \emptyset\}$$ and the transition structure is defined
$$\Delta'(B,c) = \bigcup \{\Delta(b,c) \mid b \in B\}.$$
Intuitively, the states of $\mathbb{A}'$ consist of ``macro-states'' over $\mathbb{A}$, consisting of several states of $\mathbb{A}$ being visited at once, and a run of $\mathbb{A}'$ on a word can be viewed as a ``bundle'' of runs of $\mathbb{A}$ on the same word. 
These sorts of simulation theorems turn up in various 
places, for various types of automata. In particular, one of our main technical results of this paper is a simulation theorem, allowing us to simplify the transition structure of the automata that we introduce for monadic second-order logic interpreted on coalgebras.

In this paper we are mainly interested in automata that, unlike word automata, classify \emph{infinite} objects. The simplest non-trivial example of this is the theory of automata operating on infinite words, or \emph{streams} as we shall call them. For a given alphabet $\Sigma$, a $\Sigma$-\emph{stream} is an $\omega$-sequence of elements of $\Sigma$. Automata that are used to classify streams are called \emph{stream automata}. The simplest kind of stream automata, known as \emph{B\"{u}chi automata}, are simply the same as word automata: they consist of a set of states $A$, a transition relation $\Delta : A \times \Sigma \to \psf(A)$, a start state $a_I$ and a set $F$ of accepting states. The difference is in the definition of acceptance.
\begin{defi}
 Given a stream $c_0c_1c_2c_3...$ over $\Sigma$, a \emph{run} of the automaton $\mathbb{A}$ is a stream $a_0a_1a_2a_3...$ over $A$, such that $a_0 = a_I$ and $a_{i + 1} \in \Delta(a_i,c_i)$ for all $i < \omega$. The run is said to be \emph{accepting} if the following acceptance condition is met: for infinitely many $i \in I$, we have $a_i \in F$.
\end{defi}
In other words, a run is accepting if from each given point in the computation, the automaton will eventually end up in an accepting state. The language $L(\mathbb{A})$ consists of the set of $\Sigma$-streams accepted by $\mathbb{A}$. 

Again, one of the central results about stream automata is a simulation theorem. However, here the B\"{u}chi condition is not flexible enough - it is a standard fact that deterministic B\"{u}chi automata have strictly less expressive power than non-deterministic ones. To obtain a simulation theorem, we need to switch to a more sophisticated acceptance condition, and several are available (see \cite{vene:lect12} for an overview). For the purposes of the present paper, the most natural condition is the \emph{parity} condition.

Given a finite set $A$, a \emph{priority map} is simply a map $\Omega : A \to \omega$, assigning a natural number to each element of $A$. Such a map encodes, in a compact way, two pieces of information that we want to keep track of: first, it partitions the states in $A$ into two sets, the states with even priority and the ones with odd priority. Intuitively, think of the states with even priority as the accepting states. Second, the priority map encodes an \emph{ordering} of the states, where states of higher priority are regarded as more important when determining whether or not a run is accepting. A \emph{parity stream automaton} is then a tuple $(A,\Delta,a_I,\Omega)$, where $\Omega$ is a priority map. A run $a_0a_1a_2...$ on some stream is now said to be \emph{accepting} if the highest number $k$ such that $\Omega(a_i) = k$ for infinitely many $i$ is an even number. In other words: the run visits some accepting state $a$ infinitely many times, and any non-accepting state that is also visited infinitely many times has lower priority than $a$.

Switching from B\"{u}chi automata to parity automata does not give any more expressive power: the two automaton models are expressively equivalent. The point of the parity condition is that we now have a simulation theorem available:
\begin{fact}[Simulation of stream automata]
 For every parity stream automaton $\mathbb{A}$, there is a deterministic parity stream automaton $\mathbb{A}'$ (i.e. such that $\Delta'(a',c)$ is a singleton for all $a' \in A'$) with $L(\mathbb{A}) = L(\mathbb{A}')$.
\end{fact}
A standard proof of this simulation theorem for stream automata uses a construction due to Safra, which is quite involved compared to the powerset construction for word automata. We refer to \cite{vene:lect12} for the details.  

Simulation theorems are particularly useful for proving \emph{closure} properties of automata; for example, given a deterministic parity stream automaton $\mathbb{A}$, it is trivial to define its \emph{complement}, i.e. an automaton that accepts all and only those streams that are \emph{not} accepted by $\mathbb{A}$. The construction just adjusts the priority map: if $\Omega$ was the priority map of $\mathbb{A}$, then define the complement of $\mathbb{A}$
 to be just like $\mathbb{A}$ but with the priority map $\Omega'$ defined by $\Omega'(a) = \Omega(a) + 1$.

\subsection{Infinite games}

The automata we considered in the previous section are all of a fairly simple character: when processing a word or a stream,  these automata simply change their state depending on the input that they read, in a deterministic or non-deterministic manner, and the sequence of states visited in this manner is what we have called a run. For more complex automaton models, this notion of a ``run'' of the automaton on some input will no longer be adequate - more elaborate models of computation will be required. Here, the theory of infinite games has turned out be a powerful mathematical framework, providing simple and perspicious models of how automata process their input. The idea is to think of the computation carried out by an automaton as a game between two players with opposite goals: one player tries to defend the view that the input should be accepted, the other tries to show that it should be rejected. 
\begin{defi}
A \emph{board game} $\mathbb{G}$ is a tuple $(B_\exists,B_\forall,E,W)$ where $B_\exists$ and $B_\forall$ are disjoint sets, $E \subseteq B \times B$ where $B = B_\exists \cup B_\forall$ (this union is called the \emph{board}) and $W$ is a map from the set $B^\omega$ of $B$-streams into the set $\{\exists,\forall\}$. Elements of $B$ are called \emph{positions}.
\end{defi}
Intuitively, $\exists$ (or Eloise) and $\forall$ (or Abelard) are the \emph{players}, $B_\exists$ consists of the positions on the board at which it is $\exists$'s turn to move, and similarly for $B_\forall$. Finally, the relation $E$ tells us what the legal moves are at a given position of the game. It can happen that there are no legal moves at all at a position of the game.

There is nothing to prevent that a match of a board game could go on indefinitely. In fact this is crucial to the connection between infinite games and automata operating on non-wellfounded structures, in which the computation carried out by the automaton processing some input can be infinitely long. 
In these cases, the map $W$ will be used to decide a winner. Formally: 
\begin{defi}
Let $\mathbb{G} = (B_\exists,B_\forall,E,W)$ be any board game. A \emph{match}  in $\mathbb{G}$ is either a finite tuple of positions $\pi_0...\pi_n$ or an infinite sequence $\pi_0\pi_1\pi_2...$ of positions in $B$, such that $a_i E a_{i+ 1}$ for all $i \in \omega$ (or for all $i < n$, if the match is finite.) A match $\rho$ is called \emph{complete} if it is infinite, or if $E[\pi] = \{\pi' \mid \pi E \pi'\}= \emptyset$ where $\pi$ is the last position on $\rho$. If $\rho$ is not complete it is called a \emph{partial} match. The set of partial matches $\rho$ of $\mathbb{G}$ for which the last position $\pi$ of $\rho$ belongs to $\Pi$, and such that $E[\pi]\neq \emptyset$,  is denoted by $\mathsf{PM}(\mathbb{G},\Pi)$. The \emph{winner} of an infinite match $\rho$ is defined to be $W(\rho)$, and the winner of a complete but finite match is $\exists$ if the last position is in $B_\forall$, and $\forall$ if the last position of the match is in $B_\exists$. In these cases, we say that the losing player got stuck. 
\end{defi}

\begin{defi}
A \emph{strategy} for a player $\Pi$ in a board game $\mathbb{G}$ is a map $\chi : \mathsf{PM}(\mathbb{G},\Pi) \to B$. A match $\rho$  (complete or partial) is said to be $\chi$-\emph{guided} if, for every proper initial segment $\rho'$ of $\rho$ that belongs to $ \mathsf{PM}(\mathbb{G},\Pi)$, the tuple $\rho'  \chi(\rho')$ is also an initial segment of $\rho$. 

Finally, a strategy $\chi$ for $\Pi$ is said to be a \emph{winning strategy} at a position $\pi$ in $\mathbb{G}$ if every $\chi$-guided match beginning with the position $\pi$ is won by $\Pi$. If there is a winning strategy for $\Pi$ at the position $\pi$ then we say that $\pi$ is a \emph{winning position} for $\Pi$.
\end{defi}
\begin{rem}
The reader may be more familiar with a different terminology here, where the word ``play'' is used for what we call a ``match'', and a play is said to ``respect the strategy $\chi$'' if it is $\chi$-guided in our sense. 
\end{rem}

In practice, when we define a winning strategy $\chi$ for either player in a game $\mathbb{G}$ for some given position $\pi$, we will usually only take care to define the strategy on those partial matches that are actually reachable from the  position $\pi$ by a $\chi$-guided match.

A special role in connection with automata is played by \emph{parity games}, where the winning condition $W$ is defined via a map $\Omega : B \to \omega$, such that the range $\Omega[B]$ of this map is finite. We then set $W(\rho) = \exists$ if the highest natural number $k$ such that $\Omega(\pi_i) = k$ for infinitely many $i \in \omega$ is an even number, and we set $W(\rho) = \forall$ otherwise. It follows straightforwardly from a classic result in the theory of infinite games, the \emph{Borel determinacy theorem} \cite{mart:bore75}, that all parity games are \emph{determined}, meaning that the winning positions for $\exists$ and $\forall$ partition the board: every position is winning for one of the two players. However, a much stronger and quite useful property holds for parity games:
\begin{defi}
A strategy $\chi$ for $\Pi$ is said to be \emph{positional} if, for any two partial matches $\rho_1,\rho_2 \in \mathsf{PM}(\mathbb{G},\Pi)$ such that the last positions of $\rho_1$ and $\rho_2$ are identical, we have $\chi(\rho_1) = \chi(\rho_2)$.
\end{defi}   
A positional strategy for $\Pi$ can thus be represented equivalently as a map $\chi : B_\Pi^+ \to B$, where $$B_\Pi^+ = \{\pi \in B_\Pi \mid E[\pi] \neq \emptyset\}.$$
\begin{fact}[Positional Determinacy]
Let $\mathbb{G}$ be a parity game. Then for every position $\pi \in B$, there exists a unique player $\Pi \in \{\exists,\forall\}$ and a positional winning strategy for $\Pi$ at the position $\pi$.
\end{fact} 
Several proofs of the positional determinacy theorem are known, see \cite{emersonjutla,mostowski,ziel:infi98}. We will briefly consider games with non-parity winning condition later; however, all the games we consider will be determined (although not necessarily by positional strategies), as an easy consequence of the Borel determinacy theorem. 
\subsection{Logics for Kripke models and their automata}

With the required machinery in place, we are ready to explain the connection between infinite games, automata, monadic second-order logic and the modal $\mu$-calculus. 

Let us start by briefly reviewing monadic second-order logic and the modal $\mu$-calculus, viewed as languages for reasoning about Kripke models. For simplicity we restrict our attention to Kripke models with just one accessibility relation, i.e. models for modal logic in the basic similarity type with just one diamond and one box.
Let $\Var$ be  a fixed infinite supply of variables. 
\begin{defi}
A \textit{Kripke model} is a structure $\bbS = (S,R,V)$ where $S$ is a set,
$R \subseteq S \times S$ and $V : \Var \rightarrow \psf(S)$ is a 
$\Var$-valuation. A \textit{pointed} Kripke model is a structure $(\bbS,s)$ where $\bbS$ is a
Kripke model and $s$ is a point in $\bbS$. 
\end{defi}
Associated with a valuation $V$, we define the \emph{conjugate coloring} 
$V^{\dagger}: S \to \psf(\Var)$ by $V^{\dagger}(s) \isdef \{ p \in \Var \mid
s \in V(p)\}$.
Given a subset $Z \subseteq S$, the valuation $V[p \mapsto Z]$ is as $V$ except 
that it maps the variable $p$ to $Z$.

Turning to syntax, we define the formulas of monadic second-order logic $\MSO$ 
through the following grammar:
$$
\varphi \isbnf \sr(p) \divbnf p \subseteq q \divbnf R(p,q)   
  \divbnf \neg \varphi \divbnf\varphi \vee\varphi \divbnf \exists p. \varphi
$$
with $p,q \in \Var$. (Here, $\sr$ stands for ``source''.) We define $\top := \forall p. p \subseteq p $ and $\bot := \neg \top$. Formulas are evaluated over pointed Kripke models by the following induction:
\begin{itemize}
\item $(S,R,V,s) \Vdash  \sr(p) $ iff $V(p) = \{s\}$
\item $(S,R,V,s)\Vdash p \subseteq q $ iff $V(p) \subseteq V(q)$
\item $(S,R,V,s) \Vdash R(p,q)$ iff for all $v \in V(p)$ there is $w \in V(q)$ 
  with $v R w$ 
\item standard clauses for the boolean connectives
\item $(S,R,V,s)\Vdash \exists p. \varphi$ iff $(S,R,V[p \mapsto Z],s)\Vdash 
   \varphi$ for some $Z \subseteq S$.
\end{itemize}
We present the language of the \emph{modal $\mu$-calculus} $\muML$ in
negation normal form, by the following grammar:
$$ \varphi \isbnf p \divbnf \neg p \divbnf \phi \lor \phi \divbnf \phi \land \phi
   \divbnf \Box \varphi \divbnf \Diamond \varphi 
   \divbnf \eta p. \varphi 
$$
where $p \in \Var$, $\eta \in \{ \mu, \nu \}$,  and in the formula $\eta p. \varphi$
no free occurrence of the variable $p$ in $\varphi$ may be in the scope of a negation. We define $\top := \nu p. p$ and $\bot := \mu p. p$.

The satisfaction relation between pointed Kripke models and formulas in 
$\muML$ is defined by the usual induction, with, e.g.
\begin{itemize}
\item 
$(S,R,V,s)\Vdash \mu p. \varphi$ iff $s \in \bigcap \{Z \subseteq S 
\mid \varphi_{p}(Z) \subseteq Z\}$
where $\varphi_{p}(Z)$ denotes the truth set of the formula $\varphi$ in the 
model $(S,R,V[p \mapsto Z])$.
\end{itemize}

\noindent We assume familiarity with the notion of bisimilarity between two (pointed)
Kripke models, and say that a formula of $\MSO$ is \textit{bisimulation
invariant} if  it has the same truth value in any pair of bisimilar pointed
Kripke models. 

\begin{fact}[Janin-Walukiewicz Theorem \cite{jani:expr96}]
A formula $\varphi$ of $\MSO$ is equivalent to a formula of $\muML$ iff
$\varphi$ is invariant for bisimulations.
\end{fact}

Generally, given two languages $\mathtt{L}$ and $\mathtt{L}'$ we shall write $\mathtt{L} \equiv \mathtt{L}'$ to say that every formula of $\mathtt{L}$ is semantically equivalent to a formula of $\mathtt{L}'$, and vice versa. We will denote the bisimulation invariant fragment of $\mathtt{L}$ by $\mathtt{L}/{\sim}$. With this notation we can state the previous fact more succinctly as:
$$\muML \equiv \MSO/{\sim}$$ 
The proof of the Janin-Walukiewicz theorem proceeds by comparing suitable automaton models for $\MSO$ and the modal $\mu$-calculus. 
To get an understanding of how $\MSO$-automata on Kripke models work, the main conceptual step one needs to take from the perspective of word and stream automata is to take a \emph{logical} view on the transition structure of an automaton, i.e. the transition map of an automaton will be constructed from formulas in a logical system.
Let $A$ be a set (of variables). The language of \emph{first-order logic with equality} over $A$ is given by the grammar:
$$a(x) \mid x = y \mid \varphi \wedge \varphi \mid \neg \varphi \mid \exists x .\varphi$$
where $a \in A$ and $x,y$ belong to some fixed, countably infinite set of individual variables. The set of \emph{sentences} in this language is denoted $\FOE(A)$, and the set of sentences in which each $a \in A$ appears positively is denoted by $\FOE^+(A)$. The fragment of $\FOE(A)$ in which the equality symbol does not occur is denoted by $\FO(A)$, and we write $\FO^+(A)$ for $\FOE^+(A) \cap \FO(A)$.

 A model for this language, in the usual sense, is simply a set $X$ together with an assignment of a subset of $X$ to each $a \in A$, viewed as a predicate - in other words, a valuation $V : A \to \psf(X)$. 
\begin{defi}
Let $P$ be a set of propositional variables. A \emph{$P$-chromatic $\MSO$-automaton} $\mathbb{A}$ is a tuple $(A,\Delta,a_I,\Omega)$, where $A$ is a finite set, $a_I \in A$, $\Omega : A \to \omega$ is a priority map and 
$$\Delta : A \times \psf(P) \to \FOE^+(A)$$
is the transition map, assigning a sentence in $\FOE^+(A)$ to each pair $(a,Q)$ with $a \in A$ and $Q \subseteq P$.
\end{defi}
Note that, in this automaton model, the powerset of $P$ takes on the role of an alphabet of the automaton. In order to decide whether or not an automaton $\mathbb{A}$ accepts or rejects a given pointed Kripke model $\mathbb{S},s$, we define a parity game called the \emph{acceptance game} for $\mathbb{A}$ with respect to the model $\mathbb{S}$. The role of $\exists $ is to act as a ``defense attorney'', in favor of acceptance of the model, and the goal of $\forall$ is to act as a ``prosecutor'' trying to show that the model should be rejected. A ``round'' of this game will have the following shape: first, the automaton is in a state $a$ and is visiting some point $s$ in the model $\mathbb{S}$. At such a \emph{basic position} of the acceptance game, it will be $\exists$'s turn to move, and her moves are constrained by the transition map. Formally, she needs to come up with a valuation $U : A \to \psf(R[s])$ over the immediate successors of $s$ with respect to the accessibility relation, such that the one-step Kripke model $(R[s],U)$ satisfies the  $\FOE^+(A)$-sentence $\Delta(a,V^\dagger(s)\cap P)$. Here, we think of $V^\dagger(s) \cap P$ - the collection of propositional variables in $P$ true at the point $s$ - as the input from the alphabet currently being read by the automaton. In practice, we will assume that we are dealing with a model in which the valuation $V$ assigns an empty value to each variable not in $P$, so that we may just write $V^\dagger(s)$ rather than $V^\dagger(s) \cap P$.

For a more precise description of the game, we provide the positions, player assignments, admissible moves and priorities of the acceptance game in a table:
\begin{table}[h]
    \centering
\begin{tabular}{|l|c|l|l|}
\hline
Position  & Player  &  Admissible moves & Priority 
\\ \hline
     $(a,s)\in A\times S$  
   & $\exists$  
   & $\{U : A \to \psf(R[s]) \mid 
        (R[s],U) \Vdash \Delta(a, V^{\dagger}(s)) \}$ & $\Omega(a)$
\\
     $U:A\rightarrow{\mathcal{P}}S$ 
   & $\forall$ 
   & $\{(b,t)\mid t\in U(b)\}$      & $0$                                              \\
\hline
    \end{tabular}
\caption{\label{accgame} Acceptance game for $\MSO$-automata}
\end{table}

Given $s \in S$, we say that $\mathbb{A}$ accepts the pointed model $(\mathbb{S},s)$ if $(a,s)$ is a winning position for $\exists$ in the acceptance game. In practice, we will often identify a partial match of the form:
$$(a_0,s_0)U_0(a_1,s_1)U_1...U_{n-1}(a_n,s_n)$$
of the acceptance game belonging  with the sequence of basic positions   
$$(a_0,s_0)....(a_n,s_n)$$
appearing in the match. 

It may be helpful at this point to compare the acceptance game for these automata to the notion of acceptance we presented for stream automata in Section \ref{s:automata}. In particular, the reader can may check that an accepting run of a non-deterministic parity stream automaton can be viewed as a winning strategy for $\exists$ in a parity game, but one in which the role of the opposing player $\forall$ becomes trivialized: \emph{all} choices in this game will belong to $\exists$. 

$\MSO$-formulas are not equivalent to $\MSO$-automata in general. To obtain the desired connection between formulas of $\MSO$ and $\MSO$-automata, it is necessary to restrict attention to pointed Kripke models $(S,R,V,u)$ that are \emph{tree-like}, i.e. such that the underlying frame $(S,R)$ is a tree rooted at $u$. 
The following result is proved in \cite{walu:mona96}:
\begin{fact}
\label{f:msoautomata}
For every formula $\varphi$ of $\MSO$ with free second-order variables $P$, there exists a $P$-chromatic $\MSO$-automaton accepting all and only those tree-like pointed Kripke models that satisfy $\varphi$.
\end{fact}
A converse of this result also holds, so that any $\MSO$-automaton can be translated to a sentence of $\MSO$ that is equivalent to it over tree-like models. 
Furthermore, a key insight behind the Janin-Walukiewicz theorem is that  $\MSO$ automata are very closely related to an automaton model that is suitable to model $\mu$-calculus formulas, on arbitrary Kripke models.
\begin{fact}
A class of pointed Kripke models is definable by a formula in $\muML$ with free variables $P$ if, and only if, it is recognized by some $P$-chromatic $\MSO$-automaton such that, for all $Q \subseteq P$ and $a \in A$, we have $\Delta(a,Q) \in \FO^+(A)$.
\end{fact}
The proof of Fact \ref{f:msoautomata} involves establishing a number of closure properties of $\MSO$-automata, corresponding to the logical connectives. For example, negation corresponds to \emph{complementation} on the automata side, where the complement of an automaton $\bbA$ accepts precisely those inputs that are rejected by $\bbA$. Similarly, the \emph{union} of $\mathbb{A}$ and $\mathbb{B}$ accepts precisely those inputs accepted by either $\bbA$ or $\bbB$, and corresponds to disjunction on the logical side. The most difficult property to establish is projection, corresponding to second-order quantification. The solution lies in a simulation theorem for $\MSO$-automata, simplifying the formulas appearing in the range of the transition map.
\begin{defi}
A formula in $\FOE^+(A)$ is said to be a \emph{disjunctive formula} if it is a disjunction of formulas of the form:
$$\exists x_1...\exists x_n (\mathsf{diff}(x_1,...,x_n) \wedge a_1(x_1) \wedge...\wedge a_n(x_n) \wedge \forall y(\mathsf{diff}(x_1,...,x_n,y) \rightarrow b_1(y) \vee ... \vee b_k(y)))$$
where $a_1,...,a_n,b_1,...,b_k \in A$ and $\mathsf{diff}(z_1,...,z_m)$ is an abbreviation for the formula: $$\bigwedge_{1 \leq i < j \leq m} \neg(z_i = z_j)$$
A $P$-chromatic $\MSO$-automaton $\mathbb{A}$ is said to be \emph{non-deterministic} if, for all $a \in A$ and all $Q \subseteq P$, the formula $\Delta(a,Q)$ is disjunctive.
\end{defi}
\begin{fact}[Simulation of $\MSO$-automata]
For every $\MSO$-automaton $\mathbb{A}$, there exists a non-deterministic $\MSO$-automaton $\mathbb{A}'$ that accepts the same tree-like pointed Kripke models as $\mathbb{A}$.
\end{fact}
One of our contributions in this paper is a purely \emph{semantic} account of what makes disjunctive formulas useful, that is independent of their rather specific syntax and thus suitable for our more abstract setting. 

\section{Coalgebras and generic modal fixpoint logics}
\subsection{Coalgebras and models}
\label{coalgebras-and-models}
In the general case our basic semantic structures will consist of coalgebras together with valuations.
We only consider coalgebras over the base category $\mathbf{Set}$ with sets as
objects and functions as arrows.
The co- and contravariant power set functors will be denoted by $\psf : \mathbf{Set} \to \mathbf{Set}$ and 
$\cvp : \mathbf{Set} \to \mathbf{Set}^{op}$, respectively. We recall the actions of these functors on morphisms: given $f : X \to Y$, we have $\psf f (Z) = f[Z]$ for $Z \subseteq X$, and $\cvp f (Z) = f^{-1}(Z)$ for $Z \subseteq Y$.
Covariant endofunctors on $\mathbf{Set}$ will be called \textit{set 
functors}.
\begin{conv}
We will make an additional assumption on all set functors here: we require that they preserve all monics in $\mathbf{Set}$, i.e. that $\fun f$ is an injective map whenever $f$ is. This is a very mild constraint, since it very nearly holds for all set functors already, the only possible exception being maps with empty domain \cite{trnk:some69}. So throughout this paper, ``set functor'' will be taken to mean: ``set functor that preserves all monics''. All the functors we consider here will satisfy the constraint.
\end{conv}

\begin{defi}
Let $\fun$ be a set functor. 
A \textit{$\fun$-coalgebra} is a pair $(S,\si)$ consisting of a set $S$, 
together with a map $\sigma : S \to \fun S$.
A $\fun$-\textit{model} is a structure $\mathbb{S} = (S,\sigma,V)$ where 
$(S,\sigma)$ is a $\fun$-coalgebra and $V : \Var \rightarrow \psf S$.
A \textit{pointed} $\fun$-model is a structure $(\mathbb{S},s)$ where 
$\mathbb{S}$ is a $\fun$-model and $s \in S$.
\end{defi}
The usual notion of a $p$-morphism between Kripke models can be generalized as
follows:
\begin{defi}
Let $\mathbb{S}_1 = (S_1,\sigma_1,V_1)$ and $\mathbb{S}_2 = (S_2,\sigma_2,V_2)$
be two $\fun$-models and let $f : S_1 \rightarrow S_2$ be any map. 
Then $f$ is said to be a $\fun$-\textit{model morphism} if:
\begin{enumerate}
\item 
for each variable $p$ and each $u \in S_1$, we have $u \in V_1(p)$ iff 
$f(u) \in V_2(p)$;
\item 
the map $f$ is a \textit{coalgebra morphism}, i.e. we have
$$\sigma_2 \circ f = \fun f \circ \sigma_1.$$
\end{enumerate}
\end{defi}
\begin{defi}
Two pointed $\fun$-models $(\bbS,s)$ and $(\bbS',s')$ are said to be  \textit{behaviourally equivalent}, notation: $(\bbS,s) \sim (\bbS',s')$, if $s$ and $s'$ can be 
identified by morphisms $f: 
\bbS \to \bbT$ and $f': \bbS' \to \bbT$ such that $f(s) = f'(s')$ (for some $\fun$-model $\mathbb{T}$). 
\end{defi}
It is well known that coalgebras and coalgebra morphisms for a set functor form a co-complete category, since its forgetful functor creates colimits \cite{barr:term93}. In particular this means that coproducts of arbitrary families of $\fun$-models exist, and they correspond to disjoint unions of models in the case of Kripke semantics. Concretely, we define disjoint unions of $\fun$-models as follows:
\begin{defi}
Let $\{\mathbb{S}_i\}_{i \in I}$ be a family of $\fun$-models, where $\mathbb{S}_i = (S_i,\sigma_i,V_i)$. Then we define the \emph{disjoint union} $\coprod_{i \in I} \mathbb{S}_i = (S',\sigma',V')$ by first setting $S' = \coprod_{i\in I} S_i$ to be the disjoint union of the  sets $S_i$. Let $f_i$ denote the insertion map of $S_i $ into $S'$. Then we define $\sigma'$ to be the unique map with the property that  $\sigma' \circ f_i = \fun f_i \circ \sigma_i$ for all $i \in I$. For $p \in \Var $ we define $V'(p) = \bigcup_{i \in I}f_i[V_i(p)]$.
\end{defi}
\begin{fact}
For each $j \in I$, the insertion map $f_j : \mathbb{S}_j \to \coprod_{i \in I}\mathbb{S}_i$ is a $\fun$-model morphism.
\end{fact}

Kripke models are $\fun$-models for $\fun$ equal to the (covariant) power set functor $\psf$, and two pointed Kripke models are behaviourally equivalent iff they are bisimilar in the usual sense. Comprehensive lists of examples of coalgebraic modal logics for various type functors and their applications can be found in several research papers by various authors, see for example \cite{cirs:moda11}. We shall review some of these examples later, when we consider applications of our main characterization results.
One functor that will be of particular interest in this paper is the \textit{monotone 
neighborhood} functor $\mon$, usually defined as the subfunctor of $\cvp \circ
\cvp$ given by setting $\mon X \subseteq \cvp \cvp X$ to be:
$$
 \{N \in \cvp \cvp X \mid \forall Z,Z': Z \in N 
\mathrel{\&} 
   Z \subseteq Z' \Rightarrow Z' \in N\}.$$
Following the usual terminology we shall often refer to $\mon$-models as \emph{neighborhood models}. There is a standard definition of bisimilarity for neighborhood models:
\begin{defi}
A \emph{neighborhood bisimulation} between neighborhood models $\mathbb{S}_1$ and 
$\mathbb{S}_2$ is a relation $R \subseteq S_{1} \times S_{2}$ such that, if 
$s_1 R s_2$ 
then:
\begin{itemize}
\item $V_1^\dagger(s_1) = V_2^\dagger(s_2)$;
\item for all $Z_1$ in $\sigma_1(s_1)$ there is $Z_2 $ in $\sigma_2(s_2)$ such 
   that for all $t_2 \in Z_2$ there is $t_1 \in Z_1$ with $t_1 R t_2$;
\item for all $Z_2$ in $\sigma_2(s_2)$ there is $Z_1 $ in $\sigma_1(s_1)$ such
   that for all $t_1 \in Z_1$ there is $t_2 \in Z_2$ with $t_1 R t_2$.
\end{itemize}
\end{defi}
Neighborhood bisimilarity coincides with behavioural equivalence for $\mon$, and in the next section we shall see how neighborhood bisimulations arise from a relation lifting for $\mon$. 

In many cases, behavioural equivalence of coalgebras can be described more concretely by bisimulations defined via certain relation liftings. For some work specificially on this topic, see the following papers by Marti and Venema: \cite{Martithesis,MartiVenema12,MartiVenema15}.  As an example, bisimilarity of Kripke models can be described using what is known as the ``Egli-Milner lifting'' for the powerset functor. This lifting assigns to a relation $R \subseteq X \times Y$ the ``lifted'' relation $\overline{\psf}(R) \subseteq \psf X \times \psf Y$ defined by:
$$(\alpha, \beta) \in \overline{\psf}(R) \Leftrightarrow (\forall x \in \alpha \exists y \in \beta: x R y) \;\&\; (\forall y \in \beta \exists x \in \alpha: x Ry)$$
When the functor satisfies a certain condition, \emph{preservation of weak pullbacks}, the appropriate relation lifting can be defined directly in terms of the functor. In the special case of endofunctors on $\mathbf{Set}$, weak pullback preservation can be described rather concretely as the following condition \cite{gumm:tcoa05}: consider any pair of maps $f_1 : X_1 \to Y$ and $f_2 : X_2 \to Y$, and any $\alpha_1 \in \fun X_1$, $\alpha_2 \in \fun X_2$ with $\fun f_1 (\alpha_1) = \fun f_2(\alpha_2)$. Let $$R = \{(u,v) \in X_1 \times X_2 \mid f_1(u) = f_2(v)\}$$ That is, $R$ together with its projection maps $\pi_1 : R \to X_1 $ and $\pi_2 : R \to X_2$ is the pullback of the co-span $f_1,f_2$ in $\mathbf{Set}$. Then $\fun$ preserves weak pullbacks if, in this situation, we can always find some $\gamma \in \fun R$ with $\fun \pi_1 (\gamma) = \alpha_1$ and $\fun \pi_2 (\gamma) = \alpha_2$.

Provided that $\fun$ preserves weak pullbacks, we can define a relation lifting $\overline{\fun}$ known as the \emph{Barr extension} of $\fun$ \cite{CKW}. This is a functor on the category $\mathbf{Rel}$ of sets with binary relations as arrows, and formally its action on a relation $R \subseteq X_1 \times X_2$ is defined by setting:  
$$\overline{\fun} R := \{(\alpha_1,\alpha_2) \in \fun X_1 \times \fun X_2 \mid \exists \gamma \in \fun R: \; \fun \pi_1(\gamma) = \alpha \; \& \; \fun \pi_2 (\gamma) = \alpha_2\}$$
The Egli-Milner lifting is the Barr extension of the covariant powerset functor, and bisimulations that are sound and complete for behavioural equivalence can be defined from the Barr extension of a functor in precisely the same manner as standard bisimulations are defined from the Egli-Milner lifting. This will not play a big role in the present paper, but the Barr extension will make a minor appearance when we discuss some basic facts about predicate liftings. We remark that the monotone neighborhood functor is known \emph{not} to preserve weak pullbacks, so neighborhood bisimulations can not be recovered via the Barr extension. This was in fact part of the motivation behind the work on relation liftings by Marti and Venema. 

\subsection{The coalgebraic $\mu$-calculus}

The modal $\mu$-calculus is just one in a family of logical systems that may 
collectively be referred to as the \textit{coalgebraic 
$\mu$-calculus}~\cite{cirs:expt09}.  Generally, the key notion that connects coalgebra with logic is that of a \textit{predicate lifting}. Here, we recall that $\cvp$ denotes the contravariant powerset functor, as defined in Section \ref{coalgebras-and-models}. 
\begin{defi}
\label{d:pl}
Given a set functor $\fun$, an \textit{$n$-place predicate lifting} for $\fun$ 
is a natural transformation
$$\lambda : \cvp^n \rightarrow \cvp \circ \fun,
$$
where $\cvp^n$ denotes the $n$-fold product of $\cvp$ with itself.
A predicate lifting $\lambda$ is said to be \textit{monotone} if
$$
\lambda_X(Y_1,...,Y_n) \subseteq \lambda_X(Z_1,...,Z_n),
$$
whenever $Y_i \subseteq Z_i$ for each $i$. 
The \textit{Boolean dual}  $\lambda^d$ of $\lambda$ is defined by
$$(Z_1,...,Z_n) \mapsto \fun X \setminus 
   (\lambda_X(X\setminus Z_1,...,X \setminus Z_n)).
$$
\end{defi}

We shall often think of an arbitrary finite set $A$ of size $n$ as being identified with the set $\{1,...,n\}$ via some fixed one-to-one correspondence, so that we may view any natural transformation $\lambda: \cvp^A \to \cvp \circ \fun$ as an $n$-place predicate lifting. In these cases we sometimes speak of ``predicate liftings over $A$'' rather than ``$n$-place predicate liftings'', but note that this is merely a notational variation rather than a substantial generalization of the notion of predicate liftings. A predicate lifting $\lambda$ over $A$ can be represented equivalently as a subset of $\fun \psf  A$, via an application of the Yoneda lemma (this observation is due to Schr\"{o}der, \cite{schr:expr08}). Roughly speaking, predicate liftings are used to express ``local'' properties about the next unfolding of the transition map of a coalgebra, viewed from a given state of the coalgebra and using properties of states in the coalgebra as parameters.

Given a set functor $\fun$ and a set of monotone predicate liftings $\Lambda$ for $\fun$, the language $\mu \mathtt{ML}_\Lambda$ of the coalgebraic 
$\mu$-calculus based on $\Lambda$ is defined as follows. Again we present the language in negation normal form; note that we make sure that the modal operators of the language are closed under Boolean duals.
$$
\varphi \isbnf p  \divbnf \neg p  
   \divbnf \lambda(\varphi_1,...,\varphi_n) 
   \divbnf \lambda^d(\varphi_1,...,\varphi_n) \divbnf \varphi \vee \varphi \divbnf  \varphi \wedge \varphi 
   \divbnf \eta p. \varphi
$$
where $p \in \Var$, $\lambda \in \Lambda$, $\eta \in \{\mu,\nu\}$, and, in $\eta p. \varphi$, no free occurrence of 
the variable $p$ in $\varphi$ is in the scope of a negation. If $\Lambda$ consists of \textit{all} monotone predicate liftings for $\fun$, then we denote the language $\mu \mathtt{ML}_\Lambda$ also by $\mu \mathtt{ML}_\fun$.

The semantics of formulas in a pointed $\fun$-model $(\mathbb{S},s)$, where $\mathbb{S} = (S,\sigma,V)$, is defined as follows:
\begin{itemize}
\item 
$(\mathbb{S},s)\Vdash p $ iff $s \in V(p)$ and $(\mathbb{S},s)\Vdash \neg p$ iff 
   $s \notin V(p)$.
\item If $\lambda \in \Lambda$ or is the dual of some lifting in $\Lambda$, we set $(\mathbb{S},s)\Vdash \lambda(\varphi_1,...,\varphi_n)$ iff $\sigma(s)\in 
   \lambda_S(\Vert \varphi_1 \Vert,...,\Vert \varphi_n\Vert)$, where
  $\Vert \varphi_i \Vert = \{t \in S \mid (\mathbb{S},t)\Vdash \varphi_i\}$.
\item Standard clauses for the boolean connectives
\item $(\mathbb{S},s)\Vdash \mu p. \varphi$ iff
  $s \in \bigcap \{X \subseteq S \mid \varphi_{p}(X) \subseteq X\}$,
  where $\varphi_{p}(Z)$ denotes the truth set of the formula $\varphi$ in the 
  $\fun$-model $(S,\si,V[p \mapsto Z])$.
\end{itemize}
It is routine to prove that all formulas in $\muMLT$ are bisimulation 
invariant, i.e. they have the same truth value in every pair of behaviourally equivalent pointed $\fun$-models. (We shall continue to use the term ``bisimulation invariant'' rather than ``invariant for behavioural equivalence'', although strictly speaking bisimulation and behavioural equivalence are distinct concepts.)

\begin{exa}
The powerset functor $\psf$ has a unary predicate lifting $\Diamond$ defined by setting, for $\alpha \in \psf X$:
$$\alpha \in \Diamond_X (Z) \Leftrightarrow \alpha \cap Z \neq \emptyset$$
The language $\mu \mathtt{ML}_{\{\Diamond\}}$ is precisely the standard modal $\mu$-calculus $\mu \mathtt{ML}$. The dual of $\Diamond$ is denoted $\Box$, as usual. We recall its standard definition in modal logic as $\Box \psi := \neg \Diamond \neg \psi$, which is in accordance with our definition of the dual of a predicate lifting.
\end{exa}
\begin{exa}
The monotone neighborhood functor comes equipped with a unary predicate lifting, which we shall also denote by the symbol $\Box$. It is defined by setting, for $\alpha \in \mon X$:
$$ \alpha \in \Box_X(Z) \Leftrightarrow Z \in \alpha$$
In fact, we can embed $\psf$ as a subfunctor of $\mathcal{M}$ by mapping $\alpha \in \psf X$ to the neighborhood structure $\{Z \subseteq X \mid \alpha \subseteq Z\}$.  Clearly the two predicate liftings denoted by $\Box$ both agree on this subfunctor, so it makes sense to use the same symbol to denote both. This lifting for $\mon$ corresponds to the usual box modality for monotone modal logic, and so the language $\mu \mathtt{ML}_{\{\Box\}}$ corresponds to a fixpoint extension of monotone modal logic that we will here refer to as the \textit{monotone $\mu$-calculus}, denoted $\mu \mathtt{MML}$.  The logic $\mu \mathtt{MML}$ is a very natural system to study: in a sense, it is the \emph{minimal} modal $\mu$-calculus since the \textit{only} constraint we put on the box modality is monotonicity, which is precisely what is needed for the semantics of the fixpoint operators to be well defined. It is a well-behaved system: it was recently shown to have uniform interpolation \cite{MartiSeifanVenema15}, and we shall show here that it enjoys a Janin-Walukiewicz style characterization theorem.  
\end{exa}

\subsection{Coalgebraic automata}
\label{coal-aut}

Turning to the parity automata corresponding to the language $\muMLLa$, we
first define the \textit{modal one-step language} $\MLLa$.
\begin{defi} 
\label{onestepdef}
The set $\MLLa(A)$ of \emph{modal one-step formulas} over a set $A$ of
variables is given by the following grammar:
$$ 
\varphi \isbnf
   \bot \divbnf \top \divbnf \lambda(\psi_1,...,\psi_n) \divbnf \lambda^d(\psi_1,...,\psi_n) 
   \divbnf \varphi \vee \varphi \divbnf \varphi \wedge \varphi
$$
where $\psi_1,...,\psi_n$ are formulas built up from variables in $A$ using 
disjunctions and conjunctions. More formally, we require that $\psi_1,...,\psi_n \in \mathtt{Latt}(A)$ where $\mathtt{Latt}(A)$ denotes the set of \emph{lattice formulas} over $A$ given by the grammar:
$$ \bot \divbnf \top \divbnf a \divbnf  \varphi \vee \varphi \divbnf \varphi \wedge \varphi$$
where $a$ ranges over $A$. 
\end{defi}
\begin{defi}
Given a functor $\fun$ and a set of variables $A$, a \textit{one-step model} 
over $A$ is a triple $(X,\alpha,V)$ where $X$ is any set, $\alpha \in \fun X$
and $V : A \rightarrow \cvp(X)$ is a valuation.
\end{defi}

Note that we have written the valuation $V$ as having the type $A \to \cvp (X)$ rather than $A \to \psf (X)$ in this definition. This notation is equally correct since the covariant and contravariant powerset functors differ only in their action on morphisms, and it is sometimes a more convenient notation since the naturality condition of predicate liftings is formulated in terms of $\cvp$ and not $\psf$.

The semantics of formulas in the modal one-step language in a one-step model is
given as follows:
\begin{itemize}
\item standard clauses for the boolean connectives,
\item $(X,\alpha,V) \Vdash_1 \lambda(\psi_1,...,\psi_n)$ iff 
   $\alpha \in \lambda_X(\Vert \psi_1\Vert_V,...,\Vert \psi_n \Vert_V)$
\end{itemize}
where
   $\Vert \psi_i \Vert_V \subseteq X$ is the (classical) truth set of the formula 
   $\psi_i$ under the valuation $V$. Formally, this is defined by $\Vert \bot \Vert_V = \emptyset$, $\Vert \top \Vert_V = X$, $\Vert a \Vert_V = V(a)$ for $a \in A$, $\Vert \varphi \wedge \psi \Vert_V = \Vert \varphi \Vert_V \cap \Vert \psi \Vert_V$ and $\Vert \varphi \vee \psi \Vert_V = \Vert \varphi \Vert_V \cup \Vert \psi \Vert_V$. For $u \in X$ and a lattice formula $\psi$ we also write $u \Vdash^0_V \psi$ to say that $u \in \Vert \psi \Vert_V$. 

Recalling that any predicate lifting over a set of $n$ variables $A$ corresponds to an $n$-place predicate lifting, we can view any predicate lifting $\lambda$ over $A$ as a one-step formula in $\mathtt{1ML}_{\{\lambda\}}(A)$\footnote{Or rather, if we want to be completely precise, $\lambda$ corresponds to a one-step formula in $\mathtt{1ML}_{\{\lambda'\}}(A)$ where the lifting $\lambda' : \cvp^n \to \cvp \circ \mathsf{T}$ is obtained by composing $\lambda$ with the natural isomorphism between functors $\cvp^n$ and $\cvp^A$ induced by some fixed bijection $i : n \to A$. If we write $A = \{a_1,...,a_n\}$ with $a_k = f(k -1)$, this means that we have $\lambda_X(V) = \lambda'_X(V(a_1),...,V(a_n))$. We shall permit some abuse of notation and simply identify $\lambda$ with the associated lifting $\lambda'$.
}. With this mind we can write $(X,\alpha,V) \Vdash_1 \lambda$ instead of $\alpha \in \lambda_X(V)$. In this notation the naturality constraint for a predicate lifting $\lambda$ over $A$ becomes, for every  $\alpha \in X$, every map $f : X \to Y$ and every $V : A\to \cvp (Y)$:
$$(X,\alpha,\cvp f \circ V) \Vdash_1 \lambda \; \Leftrightarrow \; (Y,\fun f(\alpha),V) \Vdash_1 \lambda. $$
 We can now define the class of automata used to characterize the coalgebraic 
$\mu$-calculus.

\begin{defi}
Let $P$ be a finite set of variables and $\Lambda$ a set of predicate liftings. 
Then a \textit{($P$-chromatic) modal $\Lambda$-automaton} is a tuple 
$(A,\Delta,\Omega,a_I)$ where $A$ is a finite set of states with $a_I \in A$,
\[
\Delta : A \times \psf(P) \rightarrow \MLLa(A)
\]
is the transition map of the automaton, and $\Omega : A \rightarrow \omega$ is 
the parity map. 
The class of these automata is denoted as $\Aut(\MLLa)$.
\end{defi}
Note that there are two distinct sets of ``variables'' involved in the automaton $\mathbb{A}$, and it is important to keep these apart since they have different roles: the variables $P$ are used to provide the alphabet of the automaton (and correspond to free variables of corresponding fixpoint formulas), while the variables $A$ are the states of the automaton (and correspond to bound variables of a corresponding fixpoint formula.)

The acceptance game for an automaton $\mathbb{A} = (A,\Delta,\Omega,a_I)$ and 
a $\fun$-model $(S,\sigma,V)$ is the parity game given by the following table:

\begin{table}[h]
    \centering
\begin{tabular}{|l|c|l|l|}
\hline
Position  & Player  &  Admissible moves & Priority
\\ \hline
     $(a,s)\in A\times S$  
   & $\exists$  
   & $\{U \in ({\mathcal{P}}S)^A \mid 
        (S,\sigma(s),U) \Vdash_{1}\Delta(a, V^{\dagger}(s)) \}$ & $\Omega(a)$
\\
     $U:A\rightarrow{\mathcal{P}}S$ 
   & $\forall$ 
   & $\{(b,t)\mid t\in U(b)\}$   & $0$                                                 \\
\hline
    \end{tabular}
\caption{\label{accgame} Acceptance game for modal automata}
\end{table}
The loser of a finite match is the player who got stuck, and the 
winner of an infinite match is $\exists$ if the greatest parity that 
appears infinitely often in the match is even, and the winner is 
$\forall$ if this parity is odd.  Note that the valuations $U$ and $V$ in Table \ref{accgame} play fundamentally different roles: $V$ is a fixed valuation given by the model on which the automaton is run, assigning values to the open variables of the automaton, while $U$ is a ``local'' valuation assigning values to the states of the automaton, which correspond roughly to bound variables of a fixpoint formula.

Given a strategy $\chi$ for either player, a match is said to be $\chi$-guided if it is consistent 
with every choice made by that player according to $\chi$. Winning strategies are defined as usual.
\begin{defi}
The automaton $\mathbb{A}$ \emph{accepts} the pointed model $(\smod,s)$, written $(\smod,s)\Vdash \mathbb{A}$,
if $\exists$ has a winning strategy in the acceptance game from the starting
position $(a_I,s)$. 
We say that an automaton $\mathbb{A}$ is \emph{equivalent} to a formula 
$\varphi \in \muMLLa$ if, for every pointed $\fun$-model 
$(\mathbb{S},s)$, we have that $\mathbb{A}$ accepts $(\mathbb{S},s)$ iff 
$(\mathbb{S},s)\Vdash \varphi$.
\end{defi}


\begin{fact}\cite{font:auto10}
\label{coalgebraicmu}
Let $\fun$ be a set functor, and $\Lambda$ a set of monotone predicate 
liftings for $\fun$.
Then
\[
\muMLLa \equiv \Aut(\MLLa).
\]
That is, there are effective transformations of formulas in $\muMLLa$ into
equivalent automata in $\Aut(\MLLa)$, and vice versa.
\end{fact}

\subsection{One-step expressive completeness}

Our main result in this paper is formulated in terms of the set of \emph{all} monotone predicate liftings for a given functor, and we would argue that this is a rather natural choice. However, in some cases it is possible to choose some smaller set of liftings, in order to have a more concrete and manageable presentation of the syntax of the corresponding $\mu$-calculus. It will then be important to choose an \emph{expressively complete} set of predicate liftings, meaning that \emph{any} monotone lifting for the functor can be defined by a formula in the modal one-step language. In particular, this will be required to recover the Janin-Walukiewicz theorem in its original form as a special case of our main results. 

\begin{defi}
Let $\Lambda$ be a set of monotone predicate liftings and $\lambda$ any predicate lifting over the finite set $A$. We say that $\lambda$ is $\Lambda$-\emph{definable} if there is a formula $\varphi \in \mathtt{1ML}_\Lambda(A)$ such that, for any one-step model $(X,\alpha,V)$ we have:
$$\alpha \in \lambda_X(V)  \text{ iff }  (X,\alpha,V) \Vdash_1 \varphi. $$  
We say that a set of monotone predicate liftings $\Lambda$ is \emph{expressively complete} if every monotone predicate lifting over any given finite set $A$ is $\Lambda$-definable.
\end{defi}
Our main observation in this section is that any weak pullback preserving functor that also preserves finite sets has a natural choice of expressively complete predicate liftings. 
\begin{defi}
Given a functor $\fun$ that preserves weak pullbacks, the \emph{$n$-ary Moss lifting} corresponding to $\alpha \in \fun \{1,...,n\}$ is the predicate lifting $\langle \alpha \rangle : \cvp^n \to \cvp \circ \fun$ such that for all sets $X$ and $Z_1,...,Z_n \subseteq X$, we have: 
$$\langle \alpha \rangle_X(Z_1,...,Z_n) = \{\beta\in \fun X \mid (\beta,\alpha) \in \overline{\fun}R\}$$
where $R\subseteq X \times \{1,...,n\}$ is defined by $u R k$ iff $u \in Z_k$. 
\end{defi}
In particular, for any given lattice formulas $\psi_1,...,\psi_n \in A$, any $\beta \in \fun \{1,...,n\}$ and any one-step model $(X,\alpha,V)$ over $A$, we have:
$$X,\alpha,V \Vdash_1 \langle \beta \rangle (\psi_1,...,\psi_n) \text{ iff } (\alpha,\fun f(\beta)) \in \overline{\fun} (\Vdash^0_V)  $$
where $f : \{1,...,n\} \to \{\psi_1,...,\psi_n\}$ is the map defined by $k \mapsto \psi_k$. 

Note that the Moss lifting $\langle \alpha \rangle $ is always a natural transformation, as required: it is obtained by composing the natural transformation $h^\alpha : \cvp^n \to \fun \circ \cvp$ defined by $$ h^\alpha_X(f) := \fun f(\alpha)\text{, for } f : n \to \cvp X\text{,}$$ with the distributive law $\delta : \fun \circ \cvp \to \cvp \circ \fun$ defined for $\beta \in \fun \cvp X $ by:
$$\delta_X(\beta) := \{\gamma \in \fun X \mid (\gamma,\beta) \in \overline{\fun}{\in_X}\}$$
where $\in_X$ is the membership relation from $X$ to $\cvp X$. That this distributive law is a natural transformation is a standard fact, see for example \cite{leal:pred08} where this view of Moss formulas as special predicate liftings is investigated thoroughly.
\begin{prop}
\label{p:expr-comp-nablas}
If $\fun$ preserves weak pullbacks and finite sets, then the set of all Moss liftings for $\fun$ (of all arities $n \in \omega$) is expressively complete. 
\end{prop}

\begin{proof}
Fix a monotone predicate lifting $\lambda$ over $A$.  
Let $(X,\alpha,V)$ be any one-step model over $A$. Let $n$ be the number of lattice formulas over $A$, and let $f :  \mathtt{Latt}(A) \to \{1,...,n\}$ be some fixed bijective enumeration of $\mathtt{Latt}(A)$, so that we can write 
$$\mathtt{Latt}(A) = f^{-1}(\{1,...,n\}) = \{\psi_1,...,\psi_n\}$$ 
Let $\theta : X \to \mathtt{Latt}(A)$ be the map sending each $u \in X$ to $\bigwedge \{a \in A \mid u \in V(a)\}$. The \emph{characteristic formula} of $(X,\alpha,V)$, denoted $\chi(X,\alpha,V)$, is defined as:
$$\langle \fun (f \circ \theta)(\alpha) \rangle (\psi_1,...,\psi_n)$$
It is an exercise in coalgebra to show that, for any one-step model $(X',\alpha',V')$, we have 
$$X',\alpha',V' \Vdash_1 \chi(X,\alpha,V) \text{ iff } (\alpha,\alpha') \in \overline{\fun}(\preceq)$$
where we write $u \preceq u'$, for $u \in X$ and $u' \in X'$, if $u \in V(a)$ implies $u' \in V'(a)$ for all $a\in A$. Furthemore, there are only finitely many characteristic formulas relative to any given set of variables $A$, since $\fun$ preserves finite sets (and hence $\fun \{1,...,n\} $ is finite). We claim that the formula:
$$\varphi := \bigvee \{\chi(X,\alpha,V) \mid \alpha \in \lambda_X(V) \}$$
defines the monotone lifting $\lambda$. One direction is clear: if $\alpha \in \lambda_X(V)$  then $X,\alpha,V \Vdash_1 \varphi$ since $(X,\alpha,V) \Vdash_1 \chi(X,\alpha,V)$. Conversely, suppose that $X,\alpha,V \Vdash_1 \varphi$. Then there exists some one-step model $(X',\alpha',V')$ such that $\alpha' \in \lambda_{X'}(V')$ and $(\alpha',\alpha) \in \overline{\fun}(\preceq)$. Write $R = \{(u,v) \mid u \preceq v\}$, let $\pi' : R \to X'$ and $\pi : R \to X$ be the projection maps, and let $\gamma \in \fun R$ be such that $\fun \pi'(\gamma) = \alpha'$ and $\fun \pi(\gamma) = \alpha$. By naturality we get $\gamma \in \lambda_{R}(\cvp \pi' \circ V')$, by monotonicity we get $\gamma \in \lambda_{R}(\cvp \pi \circ V)$ (since $\pi'(r) \in V'(a)$ implies $\pi(r) \in V(a)$ for all $a \in A$), and finally by naturality again we get $\gamma \in \lambda_X(V)$ as required. 
\end{proof}
In the special case of the powerset functor,  given a set of formulas $\{\psi_1,...,\psi_n\}$ (viewed as a member of $\psf(\mathtt{Latt}(A))$), the corresponding substitution instance of the $n$-ary Moss lifting corresponds to the usual ``cover'' formula $\nabla\{\psi_1,...,\psi_n\}$ which is defined in terms of $\Diamond$ and its dual by:
$$\nabla\{\psi_1,...,\psi_n\} := \Diamond \psi_1 \wedge ... \wedge \Diamond \psi_n \wedge \Box (\psi_1 \vee ... \vee \psi_n)$$
Hence, it follows from Proposition \ref{p:expr-comp-nablas} that the single lifting $\Diamond$ for $\psf$ is expressively complete.

\section{Coalgebraic $\MSO$}
\label{sec:mso}
\subsection{Introducing coalgebraic MSO}
We now introduce coalgebraic monadic second-order logic for a set functor $\fun$
and a set of liftings $\La$, and show how $\MSO$ can be recovered as a special 
case. Given a set $\Lambda$ of monotone predicate liftings,
we define the syntax of the monadic second-order logic $\MSO_\Lambda$ by the 
following grammar:
\[
\varphi \isbnf
  \sr(p) \divbnf p \subseteq q  
  \divbnf \lambda(p,q_1,..,q_n) \divbnf
  \varphi \vee \varphi \divbnf \neg \varphi \divbnf \exists p. \varphi
\]
where $\lambda$ is any $n$-place monotone predicate lifting in $\Lambda$ and $p,q,q_1,...,q_n \in Var$. 
For $\Lambda$ equal to the set of all monotone predicate liftings for $\fun$, we write $\mathtt{MSO}_\Lambda = \mathtt{MSO}_\fun$.

For the semantics, let $(\mathbb{S},s)$ be a pointed  $\fun$-model. 
We define the satisfaction relation ${\Vdash} \subseteq S \times 
\MSO_\fun$ as follows:
\begin{itemize}
\item $(\mathbb{S},s) \Vdash \sr(p)$ iff $V(p) = \{s\}$,
\item $(\mathbb{S},s)\Vdash p \subseteq q$ iff  $V(p) \subseteq V(q)$,
\item $(\mathbb{S},s) \Vdash \lambda(p, q_1,...,q_n)$ iff $\sigma (v)\in 
   \lambda_S(V(q_1),..,V(q_n))$ for all $ v \in V(p)$,
\item standard clauses for the Boolean connectives,
\item $(\mathbb{S},s) \Vdash \exists p .\varphi$ iff
   $(S,\sigma,V[p \mapsto Z],s)\Vdash  \varphi$, some $Z \subseteq S$.
\end{itemize}

\noindent
We introduce the following abbreviations:
\begin{itemize}
\item $p = q$ for $p \subseteq q \wedge q \subseteq p$,
\item $\Em(p)$ for $\forall q. (p \subseteq q)$,
\item $\Sing(p)$ for $\neg \Em (p) \wedge 
   \forall q (q \subseteq p \rightarrow (\Em(q) \vee q = p)),$
\end{itemize}
expressing, respectively, that $p$ and $q$ are equal, that $p$ denotes the
empty set, and that $p$ denotes a singleton.

Clearly, standard $\MSO$ is a notational variant of the logic $\MSO_{\{\Diamond\}}$, since the formula $R(p,q)$ has precisely the same satisfaction clause as $\Diamond(p,q)$ - to see this, one just has to unfold the definitions. We can state the Janin-Walukiewicz theorem in a formula as:
$$ \mu \mathtt{ML}_{\{\Diamond\}} \equiv \mathtt{MSO}_{\{\Diamond\}}/ {\sim} $$
The following useful fact is fairly easy to check, so we leave its proof to the reader:
\begin{fact}
If $\Lambda$ is a one-step expressively complete set of predicate liftings, then $\mathtt{MSO}_\Lambda \equiv \mathtt{MSO}_\fun$ and $\mu \mathtt{ML}_\Lambda \equiv \mu \mathtt{ML}_\fun$. 
\end{fact}
Since we know that $\Diamond$ is expressively complete we can now state the Janin-Walukiewicz theorem in yet another form, where we recall that $\psf$ is the covariant powerset functor:
$$\mu \mathtt{ML}_\psf \equiv \mathtt{MSO}_\psf / {\sim}$$
Generally, we let $\mathtt{MSO}_\fun / {\sim}$ denote the fragment of $\mathtt{MSO}_\fun$ that is invariant for behavioural equivalence, and we refer to it as the \emph{bisimulation invariant fragment } of $\mathtt{MSO}_\fun$.

As mentioned in the introduction, the key question in this paper will 
be to compare the expressive power of coalgebraic monadic second-order logic
to that of the coalgebraic $\mu$-calculus.
The following observation, of which the (routine) proof is omitted, provides
the easy part of the link.

\begin{prop}
\label{p:mu-to-mso}
Let $\La$ be a set of monotone predicate liftings for the set functor $\fun$.
There is an inductively defined translation $(\cdot)^{\diamond}$ mapping any
formula $\phi \in \muMLLa$ to a semantically equivalent formula $\phi^{\diamond}\in \MSOLa$.
\end{prop}

Considering the lifting $\Box$ for $\mon$, we introduce the name \emph{monotone monadic second-order logic} for the language $\mathtt{MSO}_{\{\Box\}}$, which we will henceforth denote by $\mathtt{MMSO}$. Note that the atomic formula $\Box(p,q)$ encodes a pattern of quantifier alternation of the form $\forall \exists \forall$: it says that \emph{every} state that satisfies $p$ has \emph{some} neighborhood $Z$ such that \emph{every} state in $Z$ satisfies $q$.

\subsection{Automata for coalgebraic $\MSO$}
\label{sec:aut}

In this section we introduce automata for coalgebraic monadic second-order 
logic, and translate formulas of $\mathtt{MSO}_\Lambda$ into automata operating on tree-like coalgebras.

Since the translation of standard $\mathtt{MSO}$ into parity automata is valid only over trees and not over Kripke frames in general, we should expect that coalgebraic monadic second-order logic similarly translates into parity automata over a restricted class of $\fun$-models that are in some sense ``tree-like''. So we need to start by spelling out exactly what this means.

\begin{defi}
Given a set $S$ and $\alpha \in \fun S$, a subset $X \subseteq S$ is said to
be a \textit{support} for $\alpha$ if there is some $\beta \in \fun X$ with 
$\fun \iota_{X,S}(\beta) = \alpha$. A \textit{supporting Kripke frame} for a 
$\fun$-coalgebra $(S,\sigma)$ is a binary relation $R \subseteq S \times S$ 
such that, for all $u \in S$, $R(u) = \{v \mid u R v\}$ is a support for 
$\sigma(u)$. 
\end{defi}

Given that $X \subseteq S$ is a support for $\alpha \in \fun S$, our convention that $\fun$ preserves all injectives guarantees that there is a \emph{unique} $\alpha' \in \fun X$ such that $\fun \iota_{X,S}(\alpha') = \alpha$. We shall  denote this $\alpha'$ by $\alpha\vert_X$. Intuitively, $X$ is a support for $\alpha$ if we can know all there is to know about $\alpha$ just by looking at its restriction $\alpha\vert_X$ to $X$.

\begin{defi}

A $\fun$-\textit{tree model} is a structure $(\mathbb{S},R,u)$ where 
$\mathbb{S} = (S,\sigma,V)$ is a $\fun$-model and $u \in S$, such that $R$ is 
a supporting Kripke frame for the coalgebra $(S,\sigma)$, and furthermore 
$(S,R)$ is a tree rooted at $u$ (i.e. there is a unique $R$-path from $u$ 
to $w$ for each $w \in S$).
\end{defi}

Our goal is to translate formulas in $\MSOT$ to equivalent automata over 
$\fun$-tree models.
We start by introducing a very general type of automaton, originating 
with~\cite{vene:expr14}. The motivation for taking this general perspective is to emphasize that many automata theoretic concepts and basic results apply already in this setting, and we believe the general automaton concept we introduce here   has some independent theoretical interest.

\begin{defi}
Given a finite set $A$, a \textit{generalized predicate lifting} over $A$ 
comprises an assignment of a map
$$\varphi_X : (\cvp X)^A  \rightarrow \cvp \fun X.
$$
to every set $X$.
Concepts like \emph{boolean dual} and \emph{monotonicity} apply to these 
liftings in the obvious way.
\end{defi}

The difference with respect to standard predicate liftings is that the 
components of a generalized predicate lifting do not need to form a natural 
transformation. Throughout the rest of the paper, the term ``predicate lifting'' will be reserved for the subclass of general predicate liftings that satisfy the naturality constraint, i.e. predicate liftings in the usual sense. 
\footnote{
   In the style of abstract logic, it would make sense to require a
   generalized predicate lifting to be natural with respect to certain 
   maps, the natural choice in this case being bijections. 
   For the purpose of this paper such a restriction is not needed, 
   however.
   }

\begin{defi}
\label{oslanguage}
A \textit{one-step language} $\mathlang{L}$ for a functor $\fun$ consists of a collection 
$\mathlang{L}(A)$ of generalized predicate liftings for $\fun$, for every finite set $A$. 
\end{defi} 
As for ordinary predicate liftings, given a generalized predicate lifting $\varphi$ and a one-step model 
$(X,\alpha,V)$ we sometimes write
$
(X,\alpha,V)\Vdash_1 \varphi $ rather than $\alpha \in \varphi_X(V)
$.

Our automata will be indexed by a (finite) set of variables involved, 
corresponding to the set of free variables of the $\MSOT$-formula.

\begin{defi}
Let $P \subseteq \Var$ be a finite set of variables and let $\mathlang{L}$ be a 
one-step language for the functor $\fun$. 
A \textit{($P$-chromatic) $\mathlang{L}$-automaton} is a structure 
$(A,\Delta,\Omega,a_I)$ where
\begin{itemize}
\item $A$ is a finite set, with $a_I \in A$,
\item $\Omega : A \rightarrow \omega$ is a parity map, and
\item $\Delta : A \times \psf(P) \rightarrow \mathlang{L}(A) $ is the 
  transition map of $\bbA$.
\end{itemize}
\end{defi}

There are two main differences between these automata and the modal automata we have considered in Section \ref{coal-aut}; first, we have dropped the naturality constraint on the one-step language. Second, these automata will run on \emph{$\fun$-tree models} rather than $\fun$-models. 
The \textit{acceptance game} of $\mathbb{A}$ with respect to a $\fun$-tree model 
$(T,R,\sigma,V,u)$ is given by Table~\ref{table:accgame}.
We say that the automaton $\mathbb{A}$ accepts the model $(T,R,\sigma,V,u)$ 
if $\exists$ has a winning strategy in this game (initialized at position
$(a_{I},u)$).

\begin{table}[h]
{\normalsize
\centering
\begin{tabular}{|l|c|l|l|}
\hline
Position & Player & Admissible moves & Priority   \\
\hline
    $(a,s) \in A \times T$
  & $\exists$
  & $\{U : A \to \psf(R(s)) \mid (R(s),\sigma(s)\vert_{R(s)}, U) 
               \Vdash_{1} \Delta(a,V^{\dag}(s))  \}$ & $\Omega(a)$

\\
    $U : A \rightarrow \psf(T)$
  & $\forall$
  & $\{(b,t) \mid t \in U(b) \}$ & $0$
 
\\ \hline
   \end{tabular}
   \caption{\small Acceptance game for $\mathtt{L}$-automata.}
\label{table:accgame}
}
\end{table}

\subsection{Closure properties}

The abstract perspective of generalized predicate liftings is useful for establishing some simple closure properties 
of automata, based on properties of the one-step language. 
The first, easy, results establish sufficient conditions for closure under 
union and complementation.

\begin{prop}
\label{closureunion}
If the one-step language $\mathlang{L}$ is closed under disjunction,
then the class of $\mathlang{L}$-automata is closed under union.
\end{prop}

\begin{defi}
An $\mathtt{L}$-automaton $\mathbb{A}$ is said to be \textit{monotone} if, for all $a \in A$ and all colors $c \in \psf(P)$, the generalized lifting $\Delta(a,c)$ is monotone in each variable in $A$.
 \end{defi}

\begin{prop}
\label{closurecomplementation}
If the monotone fragment of the one-step language $\mathlang{L}$ is closed under
boolean duals, then the class of monotone  $\mathlang{L}$-automata is closed under 
complementation.
\end{prop}

The most interesting property concerns closure under existential projection. 
\begin{defi}
Let $\mathbb{A}$ be a $P$-chromatic $\mathtt{L}$-automaton, and let $p \in P$. Then a $P\setminus\{p\}$-chromatic $\mathtt{L}$-automaton $\mathbb{B}$ is said to be an \emph{existential projection} of $\bbA$ if it accepts a $\fun$-tree model $(T,R,\sigma, V,u)$ precisely when there is some $Z \subseteq T$ such that $\bbA$ accepts $(T,R,\sigma, V[p \mapsto Z],u)$.
\end{defi}

The following notion originates with~\cite{jani:expr96}, but instead
of relying on a particular syntactic shape of one-step formulas as in ~\cite{jani:expr96}, we define
the concepts in purely semantic terms.

\begin{defi}
A generalized predicate lifting $\varphi$ over $A$ is said to be \textit{disjunctive} if,
for every one-step model $(X,\alpha,V)$ such that 
$$(X,\alpha,V)\Vdash_1 \varphi$$
there is a valuation $V^* : A \rightarrow \cvp(X)$ such that
\begin{itemize}
\item $V^\ast(a) \subseteq V(a)$ for each $a \in A$,
\item $V^\ast(a) \cap V^\ast(b) = \emptyset$ whenever $a \neq b$, and
\item $(X,\alpha,V^*) \Vdash_1 \varphi$. 
\end{itemize}
The second clause can also be stated as: $(V^*)^\dagger(u)$ is either empty or a singleton, for all $u \in X$.
Call an $\mathlang{L}$-automaton \textit{non-deterministic} if every generalized lifting 
$\Delta(a,c)$ is disjunctive.
\end{defi}

It is easy to see that if the language $\mathlang{L}$ is closed under 
disjunctions, then so is its disjunctive fragment.
From this we obtain the following.

\begin{prop}
\label{existentialclosure}
If the one-step language $\mathlang{L}$ is closed under disjunction, then the
class of non-deterministic $\mathlang{L}$-automata is closed under existential
projection over $\fun$-tree models. 
\end{prop}

\begin{proof}
Suppose $\mathbb{A} = (A,\Delta,a_I,\Omega)$ is a non-deterministic
$\mathlang{L}$-automaton for the variable set $P$. 
Define the $P\setminus\{p\}$-chromatic automaton 
$\exists p.\mathbb{A} = (A,\Delta^*,a_I,\Omega)$ by setting
$$
\Delta^*(a,c) = \Delta(a,c) \vee \Delta(a,c\cup\{p\}).
$$
Every $\fun$-tree model accepted by $\mathbb{A}$ is 
also accepted by $\exists p. \mathbb{A}$, since any admissible move for $\exists$ at a position $(a,u)$ in the acceptance game for $\mathbb{A}$ is an admissible move at the same position in the acceptance game for $\exists p. \bbA$ too, by definition of the transition map $\Delta^*$.

Conversely, suppose $\exists p.\mathbb{A}$ accepts some $\fun$-tree 
model $(S,R,\sigma,V,s_I)$. For each winning position $(a,s)$ in the 
acceptance game, let $V_{(a,s)}$ be the valuation chosen by $\exists$ 
according to some given winning strategy $\chi$. Note that we can assume that $\chi$ is a \textit{positional} winning strategy, since $\exists p. \mathbb{A}$ is a parity automaton.
Since disjunctive formulas are closed under disjunctions, the 
automaton $\exists p .\mathbb{A}$ is a non-deterministic automaton, 
and so for each winning position $(a,s)$ there is a valuation $V_{(a,s)}^* : 
A \rightarrow \psf(R(s))$, which is an admissible move for $\exists$, such that 
$V_{(a,s)}^*(b) \subseteq V_{(a,s)}(b)$ and such that for all $b_1 \neq b_2 
\in A$ we have $V_{(a,s)}^*(b_1) \cap V_{(a,s)}^*(b_2) = \emptyset$.
Define the strategy $\chi^*$ by letting $\exists$ choose the valuation 
$V_{(a,s)}^*$ at each winning position $(a,s)$  - this is still a winning 
strategy, since the valuations chosen by $\exists$ are smaller and so no new 
choices for $\forall$ are introduced. Furthermore, $\chi^*$  is clearly still a positional winning strategy.

From these facts it follows by a simple induction on the depth of the nodes 
in the supporting tree that the strategy $\chi^*$ is \textit{separating}, i.e. 
that for every $s \in S$ there is at most one automaton state $a$ such that 
$(a,s)$ appears in a $\chi^*$-guided match of the acceptance game. 
So we can define a valuation $V^\prime$ like $V$ except we evaluate $p$ to be
true at all and only the states $s$ such that 
$$
(R(s),\sigma(s), V^*_{(a_s,s)})\Vdash_1 \Delta(a_s,c \cup \{p\}),
$$
where $a_s$ is a necessarily \emph{unique} automaton state such that $(a,s)$
appears in some $\chi^*$-guided match, and $c$ is the color consisting of the
variables true under $V$ at $s$.
Then $\mathbb{A}$ accepts $(S,R,\sigma,V',s_I)$, by the positional strategy mapping each position $(a,s)$ to the valuation $\chi^*(a,s)$. That is, $\exists$ can play the same strategy $\chi^*(a,s)$, now viewed as a strategy in the acceptance game for $\bbA$.
\end{proof}

\subsection{Second-order automata}

We now introduce a concrete one-step language for a given set functor 
$\fun$ and a given set of monotone predicate liftings $\Lambda$, and show that 
$\MSOLa$ can be translated into the corresponding class of automata. By the closure properties established in the previous section, the language  needs to be closed under disjunction and negation, and furthermore we need to establish that every automaton for this language is equivalent to a non-deterministic automaton.
\begin{defi}
Let $\Lambda$ be a set of monotone predicate liftings for $\fun$.
The set  of \textit{second-order one-step formulas} over any
set of variables $A$ and relative to the set of liftings $\Lambda$
is defined by the grammar:
$$\varphi \isbnf a \subseteq b \divbnf \lambda(a_1,...,a_n) \divbnf \neg \varphi 
   \divbnf \varphi \lor \varphi 
   \divbnf \exists a.\varphi,
$$
where $a,b,a_1,...,a_n \in A$ and $\lambda$ is any predicate lifting in 
$\Lambda$. 
\end{defi}
The semantics of a  second-order one-step formula in a one-step model 
$(X,\alpha,V)$ (with $V : A \to \psf(X)$) is defined by the following clauses:
\begin{itemize}
\item $(X,\alpha,V)\Vdash_1 a \subseteq b$ iff $V(a) \subseteq V(b)$,
\item $(X,\alpha,V)\Vdash_1 \lambda (a_1,...,a_n)$ iff $\alpha \in 
    \lambda_X(V(a_1),...,V(a_n))$,
\item standard clauses for the Boolean connectives,
\item $(X,\alpha,V)\Vdash_1 \exists a. \varphi$ iff 
    $(X,\alpha,V[a\mapsto S])\Vdash_1 \varphi$ for some $S \subseteq X$.
\end{itemize}

Fixing an infinite set of ``one-step variables'' $\mathsf{Var}_1$, and given a finite set
$A$, the set of \textit{second-order one-step sentences} over $A$,
denoted $\mathtt{1SO}_\La(A)$, is the set of one-step formulas over $A \cup \mathsf{Var}_1$, 
with all free variables belonging to $A$. 
\begin{rem}
The difference between the one-step language $\mathtt{1SO}_\La(A)$ and the full language $\MSO_\Lambda$ may seem rather subtle at first sight. It is important to note that an $n$-place predicate lifting now corresponds to an $n$-place predicate in the language, \emph{not} an $n+1$-place predicate as in $\MSO_\Lambda$. While the formula $\lambda(q_1,...,q_n)$ of $\mathtt{1SO}_\La(A)$ expresses a proposition about the specific, fixed object $\alpha $ in a given one-step model, the  formula $\lambda(p,q_1,...,q_n)$ in $\MSO_\Lambda$ rather describes the transition map of a coalgebra as a whole: it says that the condition $\lambda(q_1,...,q_n)$ holds for the unfolding of each state $s$ that satisfies $p$.
\end{rem}
Any second-order one-step  $A$-sentence $\phi$ can be regarded as a generalized 
predicate lifting over $A$, with
$$
\varphi_X(V) = \{\alpha \in \fun X \mid (X,\alpha,V) \Vdash_1 \varphi\}.
$$
In particular, general concepts like monotonicity and closure under boolean duals apply to second-order one-step sentences. Also, note that second-order sentences are invariant under a natural notion of isomorphism:

 \begin{defi}
 An \textit{isomorphism} between two one-step models $(X_1,\alpha_1,V_1)$ and 
 $(X_2,\alpha_2,V_2)$ is a bijection $i : X_1 \to X_2$ such that 
  $\fun i(\alpha_1) = \alpha_2$ and $V_1^\dagger(u) = V_2^\dagger(i(u))$ for each 
  $u \in X_1$.
  \end{defi}
  
 \begin{prop}
  Given any set of predicate liftings $\Lambda$ and a set of variables $A$,
  any two isomorphic one-step models satisfy the same  formulas in the one-step 
  language $\mathtt{1SO}_\Lambda(A)$.
 \end{prop}

We can now introduce second-order automata; since every sentence in $\mathtt{1SO}_\La(A)$ corresponds to a generalized predicate lifting, we may view $\mathtt{1SO}_\La$ (the assignment of the one-step second-order $A$-sentences 
$\mathtt{1SO}_\La(A)$ to every set of variables $A$) as a one-step language in the sense of Definition \ref{oslanguage}. Hence the following definition is sound:

\begin{defi}
Let $\Lambda$ be a set of monotone predicate liftings for $\fun$.
A \textit{second-order $\Lambda$-automaton} is an $\mathtt{1SO}_\La$-automaton.
We write $\Aut(\mathtt{1SO}_\La)$ to denote this class, and $\Aut(\SOT)$ in case $\La$ is
the set of \emph{all} monotone predicate liftings for $\fun$.
\end{defi}

Formulas of $\SOT(A)$ are not in general monotone in the variables $A$ since negation is present in the language (unlike the one-step language $\mathtt{1ML}_\Lambda$). However, a useful trick due to Walukiewicz allows us at any time to safely replace all one-step formulas in an automaton by monotone ones:

\begin{prop}
\label{p:mon-aut}
Let $\La$ be any set of monotone predicate liftings.
Then every second-order automaton $\bbA \in \Aut(\SOLa)$ is equivalent to a monotone second-order
automaton $\bbA \in \Aut(\SOLa)$.
\end{prop}
\begin{proof}
Enumerate $A$ as $\{a_1,...,a_k\}$, and just replace each formula $\Delta(a,c)$ 
by
$$
\exists z_1 ... \exists z_k. z_1 \subseteq a_1 \wedge ... \wedge 
  z_k \subseteq a_k \wedge \Delta(a,c)[(z_i / a_i)_{i \in \{1,...,k\}}]
$$
where $\Delta(a,c)[(z_i/a_i)_{i \in \{1,...,k\}}]$ is the result of substituting the variable
$z_i$ for each open variable $a_i$ in $\Delta(a,c)$. 
This new formula is monotone in the variables $A$ and the resulting automaton
is equivalent to $\mathbb{A}$.
\end{proof}

Our aim in this section is to show that every formula of $\mathtt{MSO}_\Lambda$ is equivalent to a monotone second-order $\Lambda$-automaton, and we shall prove this by induction on the complexity of a formula. By the previous proposition, it suffices at each step to prove the existence of a second-order $\Lambda$-automaton equivalent to the given formula, since this automaton is then guaranteed to be equivalent to a monotone one. So the induction proceeds as follows: first, we produce second-order automata equivalent to all atomic formulas. By Proposition \ref{p:mon-aut} this gives us monotone second-order automata for all atomic formulas. Then, for the inductive steps, we assume that we are given monotone second-order automata equivalent to $\varphi$ and $\psi$ respectively. Our aim is to produce second-order automata that are equivalent to $\varphi \vee \psi$, $\neg \varphi$ and $\exists p. \varphi$, and once this is established we can apply Proposition \ref{p:mon-aut} again to finish the induction.

For the atomic formulas, we present only the case of formulas of the form $\lambda(p,q_1,...,q_n)$. First, note that by naturality of $\lambda$, we have for any given $\fun$-tree model $(\mathbb{S},R,s) $ that  $(\mathbb{S},R,s) \Vdash \lambda(p,q_1,...,q_n)$ if, and only if, for each $u \in V(p)$:
$$\sigma(s)\vert_{R(s)} \in \lambda_{R(s)}(V(q_1)\cap R(s),...,V(q_n)\cap R(s))$$
With this in mind, it is not hard to check that the following automaton $\mathbb{A} = (A,\Delta,a_I,\Omega)$ is equivalent to $\lambda(p,q_1,...,q_n)$ over $\fun$-tree models, where:
\begin{itemize}
\item $A = \{a_I,b_1,...,b_n\}$
\item $\Omega(a_I) = \Omega(b_1) = ... = \Omega(b_n) = 0$
\item $\Delta(a_I,c) =  \left \{ \begin{array}{ll} \forall z. z \subseteq a_I & \text{ if } p \notin c \\
 \lambda(b_1,...,b_n) \wedge \forall z. z \subseteq a_I & \text{ if }p \in c 
\end{array}\right \} $
\item  $\Delta(b_i,c) = \left \{ \begin{array}{ll} 
\bot & \text{ if } q_i \notin c \\
 \top & \text{ if } q_i \in  c 
\end{array}\right \}$
\end{itemize}
Intuitively, this automaton goes on indefinitely searching the underlying tree of the model for states that satisfy $p$, and whenever it finds such a state $u$ it checks whether $u \in \lambda_S(V(q_1),...,V(q_n))$.

For the induction, we treat the Boolean cases first: the one-step language is clearly closed under disjunction, so second-order $\Lambda$-automata are closed under union by Proposition \ref{closureunion}. For negation, suppose we have a monotone second-order automaton $\bbA$ equivalent to $\varphi$. It is not hard to show that for any monotone one-step formula we can define its Boolean dual in the one-step language, and furthermore the dual is again monotone. Hence, since $\bbA$ was monotone, we find an automaton $\mathbb{B}$ equivalent to $\neg \varphi$ by Proposition \ref{closurecomplementation}. To complete the translation of $\mathtt{MSO}_\Lambda$ into second-order $\Lambda$-automata we now only need to prove closure under existential projection.

\subsection{A coalgebraic simulation theorem}

In this section we prove that (monotone) second-order $\Lambda$-automata are closed under existential projection. The key to this result is to prove a simulation theorem, showing that every second-order automaton is equivalent to a non-deterministic one. The result then follows from Proposition \ref{existentialclosure}.

The intuition behind the simulation theorem is the same as that behind the 
standard ``powerset construction'' for word automata: the states of the new 
non-deterministic automaton $\mathbb{A}_n$ are ``macro-states'' representing
several possible states of $\mathbb{A}$ at once. 
Formally, the states of $\mathbb{A}_n$ will be binary relations over $A$, and 
given a macro-state $R$, its range gives an exact description of the states in
$\mathbb{A}$ that are currently being visited simultaneously. It is safe to 
think of the macro-states as subsets of $A$, however: the only reason that we
have binary relations over $A$ as states rather than just subsets is to have 
a memory device so that we can keep track of traces in infinite matches. 
For each macro-state $R$ and each colour $c$ we want to be able to say that the 
one-step formulas corresponding to each state in the range of $R$ hold, so we
want to translate the one-step formulas over $A$ into one-step formulas over 
the set of macro-states. In order to translate a formula $\Delta(a,c)$ to a 
new one-step formula with macro-states as variables, we have to replace the 
variable $b$ in $\Delta(a,c)$ with a new variable that acts as a stand-in for
$b$. For this purpose we introduce a new, existentally quantified variable 
$z_b$, together with a formula stating explicitly that $z_b$ is to represent
a ``disjunction'' of all those macro states that contain $b$. 
Furthermore we want all the one-step formulas to be disjunctive, and for
this purpose we simply add a conjunct ``$\mathsf{Disj}$'' to each one-step
formula, stating that the values of any pair of distinct variables appearing
in the formula are to be disjoint.  Finally, in order to turn $\mathbb{A}_n$
into a parity automaton, we use a stream automaton to detect bad traces (see
for instance \cite{vene:lect12} for the details in a more specific case).

\begin{thm}[Simulation]
\label{simulationtheorem}
Let $\La$ be a set of monotone predicate liftings for $\fun$.
For any monotone automaton $\mathbb{A} \in \Aut({\SOLa})$ there exists an equivalent
non-deterministic $\mathbb{A}'\in \Aut({\SOLa})$.
\end{thm}

Fix an automaton $\bbA = (A,\Delta,a_I,\Omega)$. We consider the set $\psf(A \times A)$ of binary relations over $A$ as a set of variables.
Let
$$\mathsf{Disj} : = \bigwedge_{R \neq R' \subseteq A\times A} \forall x . (x \subseteq R \wedge x\subseteq R') \rightarrow \mathsf{Em}(x)$$
This formula says that any two distinct relations, viewed as variables, are given disjoint values. Note that the formula does \emph{not} express the obviously contradictory condition that any two relations $R \neq R' \subseteq A \times A$  themselves are disjoint! Relations over $A$ are here being treated as formal variables, and the formula is true in any one-step model $(X,\alpha,V)$ such that $V(R) \cap V(R') = \emptyset$ whenever $R \neq R'$.

Pick a fresh variable $z_a$ for each $a \in A$. Given a 1-step formula $\varphi$, let
$$\varphi[(z_a/ a)_{a \in A}]$$ be the result of substituting $z_a$ for each free variable $a \in A$ in $\varphi$.
If we enumerate the elements of $A$ as $a_1,...,a_k$, we now define the formula $\mathsf{Sim}(\varphi,b)$ for $b \in A$ to be
\begin{displaymath}
 \exists z_{a_1}...\exists z_{a_k}. \bigwedge_{1 \leq i \leq k}   \mathsf{Rep}(z_{a_i},\{R^\prime \mid (b,a_i) \in R^\prime \}) \wedge 
 \varphi[(z_a/ a)_{a \in A}] )
\end{displaymath}
where $\mathsf{Rep}(z_{a_i},\{R^\prime \mid (b,a_i) \in R^\prime \})$ is the formula:
$$ \forall x. \mathsf{Sing}(x) \rightarrow (x \subseteq z_{a_i} \leftrightarrow \bigvee \{x \subseteq R' \mid (b,a_i) \in R' \})$$
In informal terms, the formula  $\mathsf{Rep}(z_{a_i},\{R^\prime \mid (b,a_i) \in R^\prime \})$ says that the variable $z_{a_i}$ represents a disjunction of all the macro-states that contain $a_i$. The formula $\mathsf{Sim}(\varphi,b)$ can thus be thought of as reformulating the formula $\varphi$ in terms of macro-states.
 
Let $\mathbb{A} = (A,\Delta,a_I,\Omega)$ be any monadic $\Lambda$-automaton. We can assume w.l.o.g. that $\mathbb{A}$ is monotone. We first construct the automaton $\mathbb{A}_n = (A_n,\Delta_n,a_I^*,F)$ with a non-parity acceptance condition $F \subseteq (A_n)^\omega$ as follows:
\begin{itemize}
\item $A_n = \psf(A \times A)$
\item $\Delta_n(R,c) = \mathsf{Disj} \wedge  \bigwedge_{b \in \pi_2[R]}\mathsf{Sim}(\Delta(b,c),b)$
\item $a_I^\ast = \{(a_I,a_I)\}$
\item $F$ is the set of streams over $\psf(A \times A)$ with no bad traces.
\end{itemize}
Here, $\pi_2$ is the second projection of a relation $R$ so that
$\pi_2[R]$ denotes the range of $R$. A \textit{trace} in a stream
$(R_1,R_2,R_3,...)$ over $\psf(A \times A)$ is a stream
$(a_1,a_2,a_3,...)$ over $A$ with $a_1 \in \pi_2[R_1]$ and
$(a_{j-1},a_j) \in R_j$ for $j > 1$. A trace is \textit{bad} if the
greatest number $n$ with $\Omega(a_i) = n$ for infinitely many $a_i$
is odd.\newpage

\begin{prop}
The automaton $\mathbb{A}_n$ is  disjunctive.
\end{prop}
\begin{proof}
 Since every formula $\Delta(R,c)$ contains the formula $\mathsf{Disj}$ as a conjunct, any one-step model $(X,\alpha,V)$ that satisfies $\Delta(R,c)$ must be such that $V(R) \cap V(R') = \emptyset$ whenever $R \neq R'$. Clearlty, this means that each formula $\Delta(R,c)$ is disjunctive.
\end{proof}
The following lemma is routine to check, so we leave it to the reader:
\begin{lem}
$\mathbb{A}_n$ accepts precisely the same $\fun$-tree models as $\mathbb{A}$.
\end{lem}

\begin{proof}[Proof of Theorem \ref{simulationtheorem}]
The only thing left to do at this point is to transform the automaton $\mathbb{A}_n$ into an equivalent one that has its acceptance condition given by a parity map.  We make use of the following proposition (see \cite{vene:lect12}):
\begin{claim}
\label{p:omregular}
The set $F$ of streams over $\psf(A \times A)$ that contain no bad traces w.r.t. the parity map $\Omega$ is an $\omega$-regular stream language over the alphabet $\psf(A \times A)$.
\end{claim}
By Claim \ref{p:omregular}, there exists a deterministic parity stream automaton
$$\mathbb{Z} = (Z,\delta,z_I,\Omega_z)$$
that recognizes the language $F$, with $\delta : Z \times \psf(A \times A) \rightarrow Z$. We now construct the automaton
$$\mathbb{A}_n \odot \mathbb{Z} = (A_n',\Delta_n',a_I',\Omega_n')$$
as follows:
\begin{itemize}
\item $A_n' = A_n \times Z$
\item $a_I' = (a_I^*,z_I)$
\item $\Omega_n ' (R,z) = \Omega(z)$
\item $\Delta_n'((R,z),c) = \Delta_n(R,s)[ ((R',\delta(R,z))/ R')_{R' \in \psf(A \times A)}]$
\end{itemize}
It is not difficult to check that
$\mathbb{A}_n \odot \mathbb{Z}$ is equivalent to $\mathbb{A}_n$.
Since $\mathbb{A}_n\odot \mathbb{Z}$ is clearly still a non-deterministic automaton, this ends the proof of the simulation theorem.
\end{proof}

Combining Proposition \ref{existentialclosure} with 
Theorem~\ref{simulationtheorem},
we  obtain closure under existential projection.

\begin{prop}
\label{p:clos-exist}
Let $\La$ be a set of monotone predicate liftings for a set functor $\fun$.
Over $\fun$-tree models, the class of second-order $\Lambda$-automata is closed
under existential projection.
\end{prop}

We can now combine the closure properties we have established for second-order
automata to give the desired translation of $\MSOT$ into second-order automata:

\begin{thm}
\label{p:automatachar}
For every formula $\varphi\in \MSOLa$ with free variables in $P$, there exists 
a monotone $P$-chromatic second-order $\Lambda$-automaton $\mathbb{A}_\varphi \in \Aut(\SOLa)$ which is 
equivalent to $\varphi$ over $\fun$-tree models.
\end{thm}
\begin{proof}
It suffices to show that every formula $\varphi$ is equivalent to a $P$-chromatic $\Lambda$-automaton, since the monotonicity constraint can then be taken care of using Proposition \ref{p:mon-aut}. 
We prove the theorem by induction on the complexity of $\MSOLa$-formulas. We already showed that equivalent automata are available for the atomic formulas of $\MSOLa$. The induction step for disjunctions follows immediately from Proposition \ref{closureunion}, since the one-step language $\SOLa$ is clearly closed under disjunctions. 
The induction step for negations similarly follows from Proposition \ref{closurecomplementation} once we check that the monotone fragment of the one-step language $\SOLa$ is closed under boolean duals (here, we are again appealing to Proposition \ref{p:mon-aut}, to ensure that it suffices to show that monotone second-order automata are closed under complementation). To see this, suppose that $\psi \in \SOLa(A) $ is a monotone one-step sentence, and suppose the free variables of $\psi$ are $a_1,...,a_n \in A$. For each $a_i$, add a fresh variable $p_i$, and let $\mathsf{Comp}(p_i,a_i)$ be the formula:
$$\forall x. (x \subseteq p_i  \leftrightarrow \forall y.((y \subseteq x \wedge y \subseteq a_i) \rightarrow \mathsf{Em}(y))$$ 
Informally, this formula says that $p_i$ is the complement of $a_i$. Now, let $\psi^\partial$ be defined as:
$$ \exists p_1....\exists p_n.\mathsf{Comp}(p_1,a_1) \wedge ... \wedge \mathsf{Comp}(p_n,a_n) \wedge \neg \psi[p_1/a_1,...,p_n/a_n]$$
Then $\psi^\partial$ is a monotone formula with free variables $a_1,...,a_n$, and its semantics in a one-step model is indeed that of the boolean dual of $\psi$.

Finally, the induction step for the existential second-order quantifier is handled as follows: suppose that $\mathbb{A}$ is a $P$-chromatic second-order automaton equivalent over $\fun$-tree models to the formula $\varphi$ and let $p \in P$. By Theorem \ref{simulationtheorem}, we may assume without loss of generality that $\mathbb{A}$ is disjunctive.  By Proposition \ref{existentialclosure}, there is a $P\setminus \{p\}$-chromatic second-order automaton $\exists p. \mathbb{A}$ that defines the projection of $\mathbb{A}$ with respect to the variable $p$. Hence, this automaton is equivalent to $\exists p. \varphi$ as required.
\end{proof}

\begin{cor}
\label{c:automatachar}
Suppose $\Lambda$ is an expressively complete set of monotone predicate liftings for $\fun$.
Then for every formula of $\MSOT$, there exists an equivalent 
second-order $\Lambda$-automaton over $\fun$-tree models. 
\end{cor}

\section{Bisimulation invariance}
\label{sec:bisinv}

\subsection{Uniform translations}
This section continues the program of~\cite{vene:expr14}, making use of the
automata-theoretic translation of $\MSOT$ we have just established. 
The gist of our approach is that, in order to characterize a coalgebraic 
fixpoint logic $\muMLT$ as the bisimulation-invariant fragment of $\MSOT$, it 
suffices to establish a certain type of translation between the corresponding 
one-step languages. 
First we need some definitions.

\begin{defi}
Given sets $X,Y$, a mapping $h : X \rightarrow Y$ and a valuation $V : A 
\rightarrow \cvp (Y)$, we define the valuation $V_{[h]} : A \rightarrow \cvp (X)$ by 
setting $V_{[h]}(b) = h^{-1}(V(b))$ for each $b \in A$. In other words: $$V_{[h]} := \cvp h \circ V$$
\end{defi}

Note that for a pair of composable  maps $f,g$, we have $V_{[f \circ g]} = (V_{[g]})_{[f]}$.
The most important concept that we take from \cite{vene:expr14} is that
of a \textit{uniform translation} (called \textit{uniform correspondence} 
in \cite{vene:expr14}). 

\begin{defi}
A \textit{one-step frame} is a pair $(X,\alpha)$ with $\alpha \in \fun X$. A \textit{morphism} of one-step frames $h : (X',\alpha') \to (X,\alpha)$ is a map $h : X' \to X$ with $\fun h(\alpha') = \alpha$. A one-step frame $(X',\alpha')$ together with a morphism $h : (X',\alpha') \to (X,\alpha)$ is called a \textit{cover} of $(X,\alpha)$.
\end{defi}
We can now define the notions of uniform translations and uniform constructions:
\begin{defi}
Given a functor $\fun$, a \textit{uniform construction} for 
$\fun$ is an assignment  of a cover $ h_\alpha :  (X_*,\alpha_*) \to (X,\alpha)$ to every one-step frame $(X,\alpha)$. 
\end{defi}
We will usually denote the uniform construction consisting of an assignment of covers $ h_\alpha :  (X_*,\alpha_*) \to (X,\alpha)$ simply by $(-)_*$.
\begin{rem}
A good way to think of what uniform constructions will achieve, when we use them to obtain a characterization theorem later, is as methods to neutralize all the powers of distinction that the second-order one-step language has over the modal one-step language $\mathtt{1ML}_\Lambda$. Given two one-step models that cannot be distinguished by the modal one-step language, we want to apply suitable uniform constructions to re-shape these models so that they cannot be distinguished by the second-order one-step language either. For example,  the second-order one-step language for the powerset functor is equivalent to first-order monadic logic with equality \cite{walu:mona96}, which essentially adds to the modal one-step language the capability of \emph{counting} elements of one-step models. It is precisely this capacity for counting that is the reason why formulas in this one-step language do not generally correspond to predicate liftings for the powerset functor, since conditions like ``having at least two elements'' is clearly not preserved by morphisms between one-step models. So we are looking for a uniform construction that prevents making any distinctions between one-step models by counting, and the solution is to create infinitely many copies of each element. This corresponds to the $\omega$-unravelling construction used by \cite{jani:expr96}.
\end{rem}

\begin{defi}
We say that a one-step language $\mathtt{L}$ 
\textit{admits uniform translations} if, for any finite set $A$ and any finite set $\Gamma$ of formulas in $\mathtt{L}(A)$, there 
exists a uniform construction $(-)_*$ and a translation $(-)^* : \Gamma  \to \mathtt{1ML}_\Lambda(A)$ such that for any 
one-step model $(X,\alpha,V)$ and any $\varphi \in \Gamma$, we have 
$$(X,\alpha,V) \Vdash_1 \varphi^* \text{ iff } 
  (X_*,\alpha_*,V_{[h_\alpha]})\Vdash_1 \varphi
$$
where $h_\alpha$ is the cover assigned to $(X,\alpha)$ by the construction.
The pair consisting of the translation $(-)^*$ and the uniform construction $(-)_*$ will be referred to as a \emph{uniform translation} for the set of formulas $\Gamma$.
\end{defi}

As an example, consider the disjunctive formulas introduced by Walukiewicz,

\begin{defi}
Given a $P$-chromatic second-order automaton $\bbA = (A,\Delta,a_I,\Omega)$  and a uniform translation for the finite set $\Delta[A \times \psf(P)]$, 
we get a corresponding modal automaton
$\bbA^{*} = (A,\Delta^{*},a_I,\Omega)$, with $\Delta^{*}$ given by 
$\Delta^{*}(a,c) \isdef (\Delta(a,c))^{*}$. We overload the notation $(\cdot)^*$ to refer both to a uniform one-step translation and the induced translation at the level of automata.
\end{defi}

The proof of the following result closely follows that of the main result 
in \cite{vene:expr14}. The only difference with \cite{vene:expr14} is that here we need an
``unravelling''-like component, in order to turn an arbitrary model into a tree-like one.

\begin{prop}
\label{p:unr}
Assume that $\SOLa$ admits uniform translations, and let 
$\bbA$ be a second-order $\Lambda$-automaton.
Then for each pointed $\fun$-model $(\bbS,s)$ there is a $\fun$-tree model 
$(\bbT,R,t)$, with a $\fun$-model morphism $f$ from $\bbT$ to
$\bbS$, mapping $t$ to $s$, and such that
\[
\bbA \text{ accepts } (\bbT,R,t) \text{ iff }
\bbA^{*} \text{ accepts } (\bbS,s).
\]
\end{prop}

\begin{proof}
Consider any given pointed $\fun$-model $(\mathbb{S}_1,s_1)$ where $\mathbb{S}_1 = (S_1,\sigma_1,V_1)$. We are going to construct a $\fun$-tree model $(\mathbb{S}_2,R,s_2)$, $\mathbb{S}_2 = (S_2,\sigma_2,V_2)$, together with a model morphism from the underlying pointed $\fun$-model $\mathbb{S}_2$  to $\mathbb{S}_1$ mapping $s_2$ to $s_1$, and such that $\mathbb{A}$ accepts the $\fun$-tree model $(\mathbb{S}_2,s_2)$ if and only if $\mathbb{A}^*$ accepts the pointed $\fun$-model $(\mathbb{S}_1,s_1)$. 

We construct this $\fun$-tree model as follows: for each $u \in S_1$, we define an associated pair $(X_u,\alpha_u)$ by setting $(X_u,\alpha_u) = ((S_1)_*,\sigma_1(u)_*)$.
Observe that, by the construction of these one-step models, for each $u \in S_1$, there is a mapping
$\xi_u : X_u \rightarrow S_1$
such that:
\begin{enumerate}
\item $\fun(\xi_u)(\alpha_u) =  \sigma_1(u)$
\item For each valuation $U : A \rightarrow \cvp (S_1)$, every $u \in S_1$ and every one-step formula $\Delta(a,c)$ appearing in $\mathbb{A}$, we have
$$(S_1,\sigma_1(u),U)\Vdash_1 \Delta^*(a,c) \text{ iff } (X_u,\alpha_u,U_{[\xi_u]})\Vdash_1 \Delta(a,c)$$
\end{enumerate}
The map $\xi_u$ is given by $h_{\sigma_1(u)}$. We now construct the $\fun$-tree model $(S_2,R,\sigma_2,s_2,V_2)$ as follows: first, consider the set of all non-empty finite (non-empty) tuples $(v_1,...,v_n)$ of elements in 
$$\{s_1\}\cup\bigcup_{u \in S_1} X_u$$
such that $v_1 = s_1$. We define, by induction, for each natural number $n > 0$ a subset $M_n$ of this set, and a mapping $\gamma_n : M_n \rightarrow S_1$, as follows:
\begin{itemize}
\item Set $M_1 = \{(s_1)\}$, and define $\gamma_1(s_1) = s_1$.
\item Set $M_{n + 1} = \{\vec{v}  w \mid \vec{v} \in M_n, w \in X_{\gamma_n(\vec{v})}\}$. Define $\gamma_{n + 1} (\vec{v}  w) = \xi_{\gamma_n(\vec{v})}(w)$.
\end{itemize}
Here, we write $\vec{v} w$ to denote the tuple $(v_1,...,v_n,w)$ if $\vec{v} = (v_1,...,v_n)$. 
Set $S_2 = \bigcup_{n > 0} M_n$, and define $\gamma = \bigcup_{n> 0} \gamma_n$. Define the relation $R \subseteq S_2 \times S_2$ to be
$$\{(\vec{v},\vec{v}  w) \mid \vec{v} \in S_2, w \in X_{\gamma(\vec{v}})\}$$
Note that there is, for every $\vec{v} \in S_2$, a bijection $i_{\vec{v}} : X_{\gamma(\vec{v})} \to R(\vec{v})$ given by $w \mapsto \vec{v}  w$. Note also that, for each $\vec{v} \in S_2$, we have
$$ \gamma \circ \iota_{R(\vec{v}),S_2} \circ i_{\vec{v}} = \xi_{\gamma(\vec{v})}$$

 With this in mind, we define the coalgebra structure $\sigma_2$ by setting
$$\sigma_2(\vec{v}) = \fun (\iota_{R(\vec{v}),S_2} \circ i_{\vec{v}})(\alpha_{\gamma(\vec{v})})$$
Note that $\sigma_2(\vec{v}) \vert_{R(\vec{v})} = \fun i_{\vec{v}}(\alpha_{\gamma}(\vec{v}))$.
Finally, set $s_2$ to be the unique singleton tuple with sole element $s_1$, and define the valuation $V_2$ by setting $V_2^\dagger(\vec{v}) = V_1^\dagger(\gamma(\vec{v}))$.

Clearly, $(S_2,R,\sigma_2,s_2,V_2)$ is a $\fun$-tree model. Denote the underlying $\fun$-model by $\mathbb{S}_2$.
We can then prove the following two claims:
\begin{claim}
\label{firstclaimunitrans}The map $\gamma$ is a $\fun$-model morphism from $\mathbb{S}_2$ to $\mathbb{S}_1$.
\end{claim}

\begin{proof}[Proof of Claim \ref{firstclaimunitrans}]
The map $\gamma$ clearly respects the truth values of all propositional atoms, and $\gamma(s_2) = s_1$. It suffices to show that $\gamma$ is a coalgebra morphism, i.e. that $\fun \gamma  (\sigma_2(\vec{v})) = \sigma_1 (\gamma(\vec{v}))$ for all $\vec{v}$. Pick any $\vec{v} \in S_2$. We have:
\begin{eqnarray*}
\fun\gamma(\sigma_2(\vec{v})) & = & \fun\gamma \circ \fun \iota_{R(\vec{v}),S_2} \circ \fun(i_{\vec{v}})(\alpha_{\gamma(\vec{v})}) \\
& = & \fun(\gamma \circ \iota_{R(\vec{v}),S_2} \circ i_{\vec{v}})(\alpha_{\gamma(\vec{v})}) \\
& = & \fun(\xi_{\gamma(\vec{v})})(\alpha_{\gamma(\vec{v})}) \\
& = & \sigma_1(\gamma(\vec{v}))
\end{eqnarray*}
as required.
\end{proof}

\begin{claim}
\label{secondclaimunitrans} $\mathbb{A}$ accepts $(S_2,R,\sigma_2,s_2,V_2)$ iff $\mathbb{A}^*$ accepts $\mathbb{S}_1$.
\end{claim}
\begin{proof}[Proof of Claim \ref{secondclaimunitrans}:]
We have two parts that need to be proved here:
\subsubsection*{First part: $\mathbb{A}$ accepts $(S_2,R,\sigma_2,s_{2},V_2)$ implies $\mathbb{A}^*$ accepts $(\mathbb{S}_1, s_{1})$}
Suppose that $\chi$ is a (positional) winning strategy in the acceptance game for $\mathbb{A}$ and $(S_2,R,\sigma_2,s_2,V_2)$. We are going to define a strategy $\chi^{*}$ for $\exists$ in the acceptance game for $\mathbb{A}^*$ and $\mathbb{S}_1$, with the property that for any $\chi^{*}$-guided partial match
 $$\rho = (a_1,s^{1}_{1}),...,(a_n,s^{1}_{n})$$ of length $n$ with $s^{i}_{i}\in S_1$ and $s_1^1 = s_1$, there exists a $\chi$-guided partial match
$$\rho_* =(a_1,s^{2}_{1}),...,(a_n,s^{2}_{n})$$
with $s^{2}_{i}\in S_{2}$, $s_1^2 = s_2$ and $\gamma(s^{2}_{i}) = s^{1}_{i}$ for each index $i$, and chosen in such a way that if a $\chi^*$-guided match $\rho'$ is an extension of $\rho$, then the $\chi$-guided match $\rho'_*$ is an extension of the $\chi$-guided match  $\rho_*$. 

We define the strategy $\chi^*$ by induction on the length of a partial match. For the partial match $\rho$ consisting of the single position $(a_I,s^{1}_{1})$ we let $\rho_*$ consist of the single position $(a_I,s_2)$, and we define the valuation $\chi^{*}(\rho) : A \rightarrow \cvp (S_1)$ by setting, for each $b \in A$,
$$\chi^{*}(\rho)(b) = \gamma[\chi(a_I,s_2)(b)]$$
Similarly, suppose that $\chi^*$ has been defined for all matches of length less than $n$, and let $\rho$ be any match of length $n > 1$. If $\rho$ is not $\chi^*$-guided, then we can define $\chi^{*}(\rho)(b) = \emptyset$ for all $b\in A$. If $\rho$ is $\chi^*$-guided, then write
$$\rho = (a_1,s^{1}_{1}),....,(a_n,s^{1}_{n})$$
Since this match is $\chi^*$-guided, by the induction hypothesis there is a $\chi$-guided partial match
$$ \rho_* = (a_1,s^{2}_{1}),...,(a_n,s^{2}_{n})$$ with $\gamma(s^{2}_{i}) = s^{1}_{i}$ for all $i$. We set, for each $b \in A$, 
$$\chi^*(\rho)(b) = \gamma[\chi(a_n,s^{2}_{n})(b)]$$
Clearly, with this definition, the induction hypothesis continues to  hold for all $\chi^*$-guided matches of length $n + 1$. Furthermore, we note that, whenever $\rho$ is a $\chi^*$-guided match of length $n$, the move $\chi^{*}(\rho)$ is legal: since the move $\chi(a_n,s^{2}_{n})$ must be legal, we have
$$(R(s^{2}_{n}),\sigma_2(s^{2}_{n})\vert_{R(s^2_n)},\chi(a_n,s^{2}_{n}))\Vdash_1 \Delta(a_n,V_2^\dagger(s^{2}_{n}))$$
But  since 
$$\sigma_2(s^{2}_{n})\vert_{R(s^2_n)} = \fun i_{s^2_n}(\alpha_{\gamma(s^2_n)}) = \fun i_{s^2_n}(\alpha_{s^1_n})$$ we see that $i_{s^{2}_{n}}$ is an isomorphism between the one-step models $(R[s^{2}_{n}],\sigma_2(s^{2}_{n})\vert_{R(s^2_n)},\chi(a_n,s^{2}_{n}))$ and $(X_{s^{1}_{n}},\alpha_{s^{1}_{n}},\chi(a_n,s^{2}_{n})_{[i_{s^{2}_{n}}]})$, so we have
 $$ (X_{s^{1}_{n}},\alpha_{s^{1}_{n}},\chi(a_n,s^{2}_{n})_{[i_{s^{2}_{n}}]}) \Vdash_1 \Delta(a_n,V_2^\dagger(s^{2}_{n})) $$
 For all $b \in A$, we have  
\[
\begin{array}{lcl}
\chi(a_n,s^{2}_{n})_{[i_{s^{2}_{n}}]}(b) & \subseteq & \xi_{s^{2}_{n}}^{-1}(\xi_{s^{2}_{n}}[\chi(a_n,s^{2}_{n})_{[i_{s^{2}_{n}}]}(b)])\\
& = & \xi_{s^{2}_{n}}^{-1}(\gamma[\chi(a_n,s^{2}_{n})(b)])\\
& = & \xi_{s^{2}_{n}}^{-1}(\chi^{*}(\rho)(b)) \\
& = & \chi^*(\rho)_{[\xi_{(s^{2}_{n})}]}(b) 
\end{array}
\]
So by monotonicity we get
 $$ (X_{s^{1}_{n}},\alpha_{s^{2}_{n}},\chi^{*}(\rho)_{[\xi_{s^{2}_{n}}]}) \Vdash_1 \Delta(a_n,V_2^\dagger(s^{2}_{n})) $$
Hence, since $\Delta^*(a_n,V_2^\dagger(s^{2}_{n}))$ is the uniform translation of $\Delta(a_n,V_2^\dagger(s^{2}_{n}))$, we get
$$(S_1,\sigma_1(s^{2}_{n}),\chi^{*}(\rho)) \Vdash_1 \Delta^*(a_n,V_2^\dagger(s^{2}_{n})) $$
as required. It is now easy to show that $\exists$ never gets stuck in a $\chi^*$-guided partial match, and we also see that every infinite $\chi^*$-guided  match
$$(a_I,s^{1}_{1}),(a_2,s^{1}_{2}),(a_3,s^{1}_{3}),...$$
corresponds to an infinite $\chi$-guided match
$$(a_I,s^{2}_{1}),(a_2,s^{2}_{2}),(a_3,s^{2}_{3}),...$$
and so since $\exists$ wins every infinite $\chi$-guided match, she wins every infinite $\chi^*$-guided match as well.

\subsubsection*{Second part: $\mathbb{A}^*$ accepts $(\mathbb{S}_1, s_{1})$ implies $\mathbb{A}$ accepts $(S_2,R,\sigma_2,s_{2},V_2)$ }
Let $\chi^*$ be a winning strategy for $\exists$ in the acceptance game for $\mathbb{A}^*$ and $\mathbb{S}_1$. We define the strategy $\chi$ as follows: given a partial match
$$\rho = (a_1,s^{2}_{1}),...,(a_n,s^{2}_{n})$$
set $\chi(\rho) = \chi^{*}(\gamma(\rho))_{[\gamma]}$, where $\gamma(\rho)$ is the match
$$(a_1,\gamma(s^{2}_{1})),...,(a_n,\gamma(s^{2}_{n}))$$ 
It is straightforward to check that, whenever $\rho$ is a $\chi$-guided match, $\gamma(\rho)$ is a $\chi^*$-guided match, and that the move $\chi^{*}(\gamma(\rho))$ is legal if and only if the move $\chi(\rho)$ is legal. It follows that $\exists$ never gets stuck in a $\chi$-guided partial match, and that she wins every infinite $\chi$-guided match, since an infinite $\chi$-guided match
$$(a_I,s^{2}_{1}),(a_2,s^{2}_{2}),(a_3,s^{2}_{3}),...$$
corresponds to an infinite $\chi^{*}$-guided match
$$(a_I,\gamma(s^{2}_{1})),(a_2,\gamma(s^{2}_{2})),(a_3,\gamma(s^{2}_{3})),...$$
This concludes the proof.
\end{proof}
The lemma now follows by combining the two claims \ref{firstclaimunitrans} and \ref{secondclaimunitrans}.
 \end{proof}

From this, we get the following result. Here, we use the fact that every $\mathtt{MSO}_\Lambda$-formula is equivalent to a \textit{monotone} second-order automaton over $\fun$-tree models, so that it suffices to find translations of monotone one-step formulas.

\begin{thm}[Auxiliary Characterization Theorem]
\label{generalcharacterization}
Let $\Lambda$ be an expressively complete set of monotone predicate liftings 
for a set functor $\fun$, and assume that the monotone fragment of the one-step language $\mathtt{1SO}_{\Lambda}$ admits uniform translations. 
Then:
$$\mathtt{MSO}_{\Lambda}{/}{\sim} \equiv \mu \mathtt{ML}_{\Lambda}$$
\end{thm}

\begin{proof}
Given a bisimulation invariant formula $\varphi$ of $\mathtt{MSO}_\Lambda$, let $\mathbb{A}$ be an equivalent monotone second-order automaton, and suppose there is a uniform translation $(-)^*$ of the monotone fragment of  the one-step language $\mathtt{1SO}_{\Lambda}$. Let $\psi$ be a formula in $\mu \mathtt{ML}_\Lambda$ equivalent to the automaton $\mathbb{A}^*$. Then for any pointed $\fun$-model $\mathbb{S},s$, let $(\mathbb{T},R,t)$ be the  $\fun$-tree model provided by Proposition \ref{p:unr}. Then we get:
\begin{eqnarray*}
\mathbb{S},s \Vdash \varphi & \Leftrightarrow &  \mathbb{T},t \Vdash \varphi \\
& \Leftrightarrow & \mathbb{T},R,t \Vdash \mathbb{A} \\
& \Leftrightarrow & \mathbb{S},s \Vdash \mathbb{A}^* \\
& \Leftrightarrow & \mathbb{S},s \Vdash \psi
\end{eqnarray*}
and the proof is done.
\end{proof}

\subsection{Adequate uniform constructions}

The existence of uniform translations for the one-step language 
\cite{vene:expr14} involves two components: a translation 
on the syntactic side and a uniform construction on the semantic side. 
One of our main new observations in this paper is that we can actually forget about the syntactic translation and focus entirely on the model-theoretic side of the construction. If we can find a suitable 
uniform construction for the one-step models, satisfying a certain model-theoretic condition with respect to the second-order one-step language, the syntactic translation will 
come for free.

\begin{defi}
Let $\varphi$ be any one-step formula in  $\SOLa(A)$. Then a uniform construction $(-)_*$ for $\fun$ 
 is called \textit{adequate} for $\varphi$ if, for any pair of one-step frames $(X,\alpha)$ and $(Y,\beta)$, any one-step frame morphism
 $f : (X,\alpha) \rightarrow (Y,\beta)$ and any valuation $V : A \rightarrow \cvp(Y)$, we 
have the following condition $(\star)$:
$$
(X_*,\alpha_*,V_{[f \circ h_\alpha]})\Vdash_1 \varphi 
  \text{ iff } (Y_*, \beta_*, V_{[h_\beta]}) \Vdash_1 \varphi $$
where $h_\alpha : (X_*,\alpha_*) \to (X,\alpha)$ and $h_\beta : (Y_*,\beta_*) \to (Y,\beta)$ are the covering maps given by the uniform construction.
We say that the construction is adequate for a set of formulas $\Gamma$ if it is adequate for every member of $\Gamma$.
\end{defi} 

The following diagram clarifies condition $(\star)$: 
\begin{equation*}
\begin{xy}
(0,15)*+{ (X_*,\alpha_*,V_{[f \circ h_\alpha]})}="a";
(0,0)*+{(X,\alpha, V_{[f]})}="d"; 
(46,15)*+{(Y_*, \beta_*, V_{[h_\beta]})}="b";(46,0)*+{(Y,\beta,V)}="c";
(25,15)*+{\stackrel{\varphi}{\Longleftrightarrow}}="e"; {\ar_{f}"d";"c"};%
{\ar_{h_\alpha}"a";"d"};%
{\ar^{h_\beta}"b";"c"}
\end{xy}
\end{equation*}

We are now ready to prove our first main theorem, stating that if $\fun$ admits an adequate construction for any finite set of one-step formulas, then $\MSO_\fun/{\sim} \equiv \mu\mathtt{ML}_\fun$.

\begin{proof}[Proof of Theorem \ref{mainone}]
By Theorem \ref{generalcharacterization} it suffices to show that the monotone fragment of $\mathtt{1SO}_\fun$ admits uniform translations. Let $\Gamma$ be a finite set of monotone formulas of $\mathtt{1SO}_\fun(A)$, for some finite set $A$, and suppose the uniform construction $(-)_*$ is adequate for $\Gamma$. Given $\varphi \in \Gamma$ we define a monotone predicate lifting $\lambda$ over $A$  by setting
$$\alpha \in \lambda_X(V)  \text{ iff } (X_*,\alpha_*,V_{[h_\alpha]}) \Vdash_1 \varphi.$$
It is easy to check that this lifting is monotone, and naturality of $\lambda$ is a direct consequence of the condition $(\star)$. If $A$ has $n$ elements and we list these as $(a_1,...,a_n)$, then we can view $\lambda$ as  a formula in $\mathtt{1ML}_\Lambda(A)$ of the form $\lambda'(a_1,...,a_n)$, where $\lambda'$ is the $n$-place predicate lifting corresponding to $\lambda$ with this given ordering of $A$. We now get a uniform translation by mapping $\varphi$ to the formula $\lambda' (a_1,...,a_n)$.
\end{proof}

Note that, if $\Lambda$ is an expressively complete set of liftings for $\fun$, then it suffices to find an adequate uniform construction for every finite set of formulas in $\mathtt{1SO}_\Lambda(A)$, since every formula in $\mathtt{1SO}_\fun(A)$ is semantically equivalent to one in $\mathtt{1SO}_\Lambda(A)$. Hence, by Theorem \ref{generalcharacterization} we get:

\begin{cor}
\label{exp-complete-cor}
Let $\fun$ be any set functor, and let $\Lambda$ be an expressively complete set of predicate liftings for $\fun$. If $\fun$ admits an adequate uniform construction for every finite set of formulas of $\mathtt{1SO}_\Lambda$, then 
$$\mathtt{MSO}_\Lambda/{\sim} \equiv \mu \mathtt{ML}_\Lambda.$$
\end{cor}

 If $\fun$ preserves weak pullbacks then we can reformulate  $(\star)$ as the following condition: 
\begin{defi}
Suppose $\fun$ preserves weak pullbacks.  
A \textit{one-step $\overline{\fun}$-bisimulation} between one-step models $(X,\alpha,V)$ and $(Y,\beta,U)$ is a relation $R \subseteq X \times Y$ such that:
\begin{itemize}
\item If $x R y$ then $V^\dagger(x) = U^\dagger (y)$.
\item $(\alpha,\beta) \in \overline{\fun}R$.
\end{itemize}
\end{defi}

\begin{prop}
Suppose that $\fun$ preserves weak pullbacks, and let $(-)_*$ be a uniform construction for $\fun$. Then $(-)_*$ is adequate for $\varphi$ if, and only if, for every pair of one-step $\overline{\fun}$-bisimilar one-step models $(X,\alpha,V)$ and $(Y,\beta, U)$ we have:
$$(X_*,\alpha_*, V_{[h_{\alpha}]}) \Vdash_1 \varphi  \Leftrightarrow (Y_*, \beta_*,U_{[h_{\beta}]}) \Vdash_1 \varphi.$$
\end{prop}

\subsection{Some easy applications}
We now start with some straightforward applications of the main characterization theorem.

\begin{exa}
\label{ex:psf}
As a first application, the standard Janin-Walukiewicz characterization of the 
modal $\mu$-calculus can be seen as an instance of the first main characterization result by taking
$\La = \{ \Diamond \}$ and $\fun = \psf$, recalling that $\MSO = \MSO_{\{\Diamond\}}$. 
The  uniform construction for $\psf$, which is adequate with respect to any set of one-step formulas, is given as follows: consider a 
pair $(X,\alpha)$ with $\alpha \in \psf(X)$. 
We take this to $X_* =  \alpha_* = \alpha \times \omega$, and we let $h_\alpha: 
\alpha \times \omega \to X$ be the projection map. Proving adequacy of this construction is an exercise in model theory, and left to the reader. (The proof is a much simpler special case of the argument we will use in Section \ref{adguidedonc}.)
\end{exa}

\begin{exa}
\label{ex:bag}
Consider the finitary multiset (``bags'') functor 
$\mathcal{B}$, which sends a set $X$ to the set of mappings $f : X \rightarrow
\omega$ such that the set $\{u \in X \mid f(u) = 0\}$ is cofinite.
The action on morphisms is given by letting, for $f \in \mathcal{B}X$ and 
$h : X \rightarrow Y$, the multiset $\mathcal{B}h(f) : Y \rightarrow \omega$ 
be defined by $w \mapsto \sum_{h(v) = w} f(v)$. 
Given a pair $X,\alpha$ where $\alpha : X \rightarrow \omega$ has finite support,
we define 
$$X_{*} = \bigcup \{\{u\} \times \alpha(u) \mid u \in X\}.
$$
Here, we follow the standard procedure and identify each finite ordinal $n \in \omega$ with the set of smaller ordinals, so that $0$ is identified with $\emptyset$, $1$ is $\{\emptyset\}$, $2$ is $\{\emptyset, \{\emptyset\}\}$ and so on.  
The mapping $\alpha_{*} : X_{*} \rightarrow \omega$ is defined by setting 
$\alpha_{*}(w) = 1$ for all $w \in X_{*}$. 
The map $h_\alpha :  X_* \rightarrow X$ is defined by $(u,i) \mapsto u$. 
It is easy to check that whenever two one-step frames $(X,\alpha)$ and $(Y,\beta)$ are related by some morphism $f : (X,\alpha) \to (Y,\beta)$, the models $(X_*,\alpha_*,V_{[f \circ h_\alpha]})$ and $(Y_*,\beta_*,V_{[h_\beta]})$ are isomorphic, for any valuation $V : A \to \cvp (Y)$. It follows that the construction is adequate for any set of one-step formulas, hence we get
$\mu \mathtt{ML}_{\mathcal{B}}\equiv\mathtt{MSO}_{\mathcal{B}} {/} {\sim}$.
\end{exa}

\begin{exa}
Consider the set of all \textit{exponential polynomial 
functors}~\cite{jaco:intr12} defined by the ``grammar'':
$$\mathsf{T} \isbnf \mathsf{C} 
   \divbnf \mathsf{Id} \divbnf \mathsf{T} \times \mathsf{T} 
   \divbnf \coprod_{i \in I} \mathsf{T}_i \divbnf \mathsf{T}^\mathsf{C} 
$$
where $\mathsf{C}$ is any constant functor for some set $C$, and $\mathsf{Id}$ 
is the identity functor on $\mathbf{Set}$. 
These functors cover many important applications: streams, binary trees,
deterministic finite automata and deterministic labelled transition systems
are all examples of coalgebras for exponential polynomial functors, as is 
the so called \textit{game functor} whose coalgebras provide semantics for 
Coalition Logic \cite{cirs:moda11}. For this last instance, the ``game functor'' $\mathcal{G}$ for $n$ agents can be written in the form of an exponential polynomial functor as follows:
$$  \coprod_{\langle S_0,...,S_{n-1} \rangle \in (\psf(\omega)\setminus \{\emptyset\})^n} \{\langle S_0,...,S_{n - 1}\rangle \} \times \mathsf{Id}^{(S_0 \times ... \times S_{n-1})} $$
Then, for a given set $X$, an element of $\mathcal{G}X$ will be a pair consisting of a vector $\langle S_0,...,S_{n - 1}\rangle$ of available strategies for each player, together with an ``outcome map'' $f$ assigning an element of $X$ to each strategy profile in $S_0 \times ... \times S_{n-1}$.

Every exponential polynomial functor admits adequate uniform 
constructions for all sets of one-step formulas. The proof proceeds by a straighforward induction, and similarly to the case for the bags functor it provides each exponential polynomial functor $\fun$ with a uniform construction such that for any one-step frame morphism $f : (X,\alpha) \to (Y,\beta)$ and any $V : A \to \cvp (Y)$, the one-step models $(X_*,\alpha_*,V_{[f \circ h_\alpha]})$ and $(Y_*,\beta_*,V_{[h_\beta]})$ are isomorphic. Hence, we get:

\begin{thm}
\label{c:epf}
For every exponential polynomial functor $\mathsf{T}$, we have 
$\mu \mathtt{ML}_{\fun}\equiv\mathtt{MSO}_{\fun} {/} {\sim}$.
\end{thm}
\end{exa}
\begin{proof}
We provide a sketch of the inductive construction of an adequate uniform construction. 
\subsubsection*{Constant functor}
For the constant functor $\mathsf{C}$, a one-step frame is a pair $(X,c)$ with $c\in C$. We set $X_* = \emptyset$, $c_* = c$ and $h_c : \emptyset \to X$ to be the unique inclusion of the empty set.
\subsubsection*{Identity functor}
Given a one-step frame $(X,u)$ for the identity functor, which consists of a set $X$ and $u \in X$, we set $X_* = \{u\}$, $u_* = u$ and we set $h_u : \{u\} \to X$ to be the inclusion map sending $u$ to itself.
\subsubsection*{Product}
Suppose that $\fun_1$ and $\fun_2$ have associated adequate uniform constructions with the required property. Consider a one-step $\fun_1 \times \fun_2$-frame $(X,(\alpha,\beta))$ with $\alpha \in \fun_1 X$ and $\beta \in \fun_2 X$. Let $h_1 : (X_1,\alpha_1) \to (X,\alpha)$ be the cover assigned by the uniform construction for $\fun_1$ and let $h_2 : (X_2,\beta_2) \to (X,\beta)$ be the cover assigned by the uniform construction for $\fun_2$. Then we take $X_*$ to be the disjoint union $X_1 + X_2$, and set 
$$(\alpha,\beta)_* = (\fun_1 i_1(\alpha_1),\fun_2 i_2(\beta_2))$$
where $i_1 : X_1 \to X_1 + X_2$ and $i_2 :  X_2 \to X_1 + X_2 $ are the insertion maps for the co-product. Finally, we define the covering map $h_{(\alpha,\beta)} : X_1 + X_2 \to X$ be obtained by simply co-tupling the maps $h_1,h_2$, i.e. $h_{(\alpha,\beta)}$ is the map given by the universal property of the co-product applied to the diagram $X_1 \stackrel{h_1}{\longrightarrow} X \stackrel{h_2}{\longleftarrow} X_2$. We get:
\begin{eqnarray*}
(\fun_1 \times \fun_2)h_{(\alpha,\beta)}((\alpha,\beta)_*) & = & (\fun_1 h_{(\alpha,\beta)}(\fun_1 i_1(\alpha_1)), \fun_2 h_{(\alpha,\beta)}(\fun_2 i_2(\beta_2))) \\
& = & (\fun_1 (h_{(\alpha,\beta)} \circ i_1 )(\alpha_1), \fun_2 (h_{(\alpha,\beta)} \circ i_2) (\beta_2)) \\
& = & (\fun_1  h_1 (\alpha_1), \fun_2 h_2 (\beta_2)) \\
& = & (\alpha,\beta)
\end{eqnarray*}
so $h_{(\alpha,\beta)}$ is indeed a covering map as required.
\subsubsection*{Exponentiation}
The case of a functor $\fun^\mathsf{C}$ for some constant $\mathsf{C}$ is handled analogously with the case of binary products, so we leave it to the reader.
\subsubsection*{Co-product}
This step of the construction is actually the easiest one. Suppose that each functor $\fun_i$ for $i \in I$ is equipped with an adequation uniform construction, in the strong sense demanded by our induction hypothesis. Without loss of generality let's assume that $\fun_i X \cap \fun_j X = \emptyset$ for $i \neq j$. Let $(X,\alpha)$ be a one-step frame for the co-product $\coprod_{i \in I} \mathsf{T}_i$. Then there is a unique $i \in I$ with $\alpha \in \fun_i X$, and so we define the cover $h_\alpha : (X_*,\alpha_*) \to (X,\alpha)$ merely by applying the uniform construction for $\fun_i$. 
\end{proof}
\begin{rem}
These uniform constructions were all designed in a case-by-case fashion, and at the present time we do not know whether there is any general recipe for producing an adequate uniform construction when it exists. What the constructions mentioned so far seem to have in common is that we want to produce as many equivalent (in some sense) copies of each state in a one-step model, but this is not always sufficient. In the next section we will see a somewhat more involved construction for the monotone mu-calculus, which aims to create sufficiently many copies of each state but also, crucially, sufficiently many \emph{pairwise disjoint} copies of all the neighborhoods. In this case we are trying to neutralize not only the  capability of the second-order one-step language to count states in one-step models, but also its capability to express how certain neighborhoods are related to each other, in particular, whether they overlap or not. For example, the second-order one-step language can express the property that any two neighborhoods intersect, or there is a smallest neighborhood contained in all others etc., and the uniform construction we provide needs to trivialize all such statements. In general, applying our main result as it stands may require a bit of creativity, and we regard it as an interesting (possibly quite hard) task for future research to come up with a result that makes the task entirely mechanical. We mention some related questions in our concluding remarks section. 
\end{rem}
\section{Characterizing the monotone $\mu$-calculus}
\label{sec:mon}

This final section of the paper concerns the monotone neighborhood functor 
$\mon$. There are two main difficulties here. First, the uniform construction will be less straightforward than in the previous applications, for the reasons that we explained above. Second, we need to do a bit of ``pre-processing'' of the monotone neighborhood functor to make the technique of uniform constructions applicable. In fact, the characterization theorem cannot possibly apply directly to the monotone $\mu$-calculus: it turns out that there is \textit{no} adequate uniform construction for $\mon$. For these reasons, we view the results of this section as an independent contribution of the paper, rather than just another application to illustrate the main theorem.

\begin{prop}
The functor $\mon$ does not admit an adequate uniform construction for the formula $\varphi$ defined as: $\neg \mathsf{Em}(a)$.
\end{prop}
\begin{proof} The formula $\varphi$ just says that the value of $a$ is non-empty. Suppose there existed an adequate uniform construction for this formula. Consider the situation depicted in the diagram below, which shows three one-step frames together with two one-step frame morphisms, one for each of the two  bottom one-step frames. The top frame consists of two points $\{w_1,w_2\}$ and has neighborhoods $\{\{w_1,w_2\},\{w_2\}\}$, the bottom left frame has points $\{u_1,u_2,u_3\}$ and neighborhoods $\{\{u_1,u_2\},\{u_2,u_3\},\{u_1,u_2,u_3\}\}$. Finally the bottom right frame has a single point $\{v\}$ and this singleton as its only neighborhood. 
In other words, the picture shows a co-span in the category of one-step frames and one-step frame morphisms. Furthermore, consider the valuation on the topmost one-step frame which makes $a$ true at exactly the one state $w_1$, i.e. the one not belonging to the singleton neighborhood. This valuation is depicted in the diagram by representing the state where $a$ is true by a blank circle, and the state where it is not true by a filled circle. The induced valuations on the bottom one-step frames via the frame morphisms are depicted in the same manner. With respect to  these valuations, $p$ will be true at $u_1$, but false at $u_2,u_3$
 and $v$.
\begin{center}
\setlength{\unitlength}{0.8cm}
\begin{picture}(16,12)(0,0)
\put(-0.5,-7.7){\includegraphics[width = 13cm]{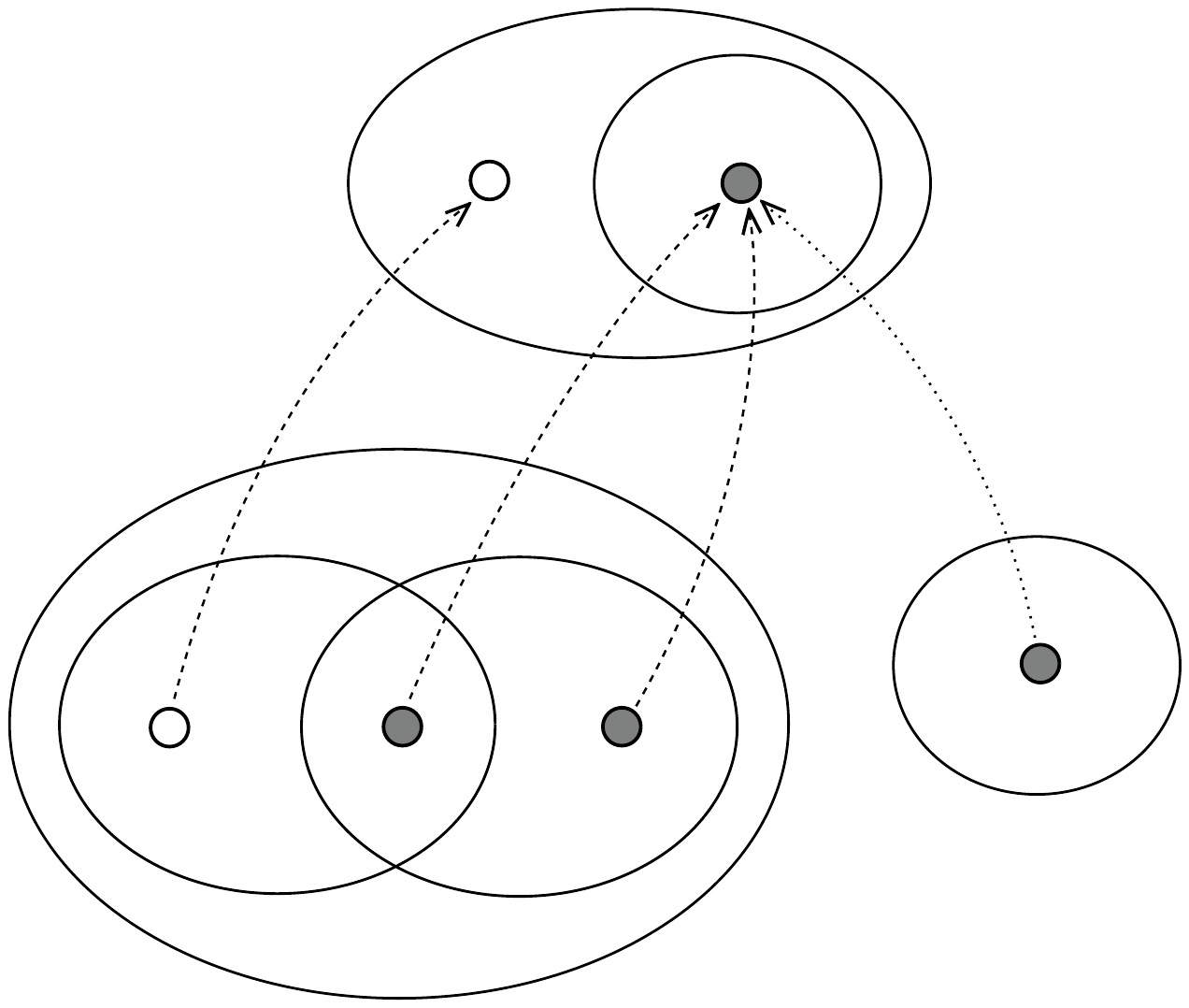}}
\put(2.7,1.4){Every cover satisfies $\varphi$}
\put(10,3.2){No cover satisfies $\varphi$}
\end{picture}
\end{center}
Let us denote the top frame by $(Y,\beta)$, its valuation by $V$, the bottom-left frame as $(X,\alpha)$, the bottom-right one as $(X',\alpha')$ and the corresponding frame morphisms as $f$ and $f'$ respectively. 
Now, the supposed adequate construction will assign a cover to each of the three frames, and we get valuations for each cover from the valuations depicted in the diagram. It follows from the defining condition $(\star)$ for adequacy that $\varphi$ must have the same truth value in each of these covers. But this leads to a contradiction: it is not hard to check that the cover $h_\alpha : (X_*,\alpha_*) \to (X,\alpha)$ must be such that $(X_*,\alpha_*,V_{[f \circ h_\alpha]}) \Vdash \varphi$, because the image of the map $h_\alpha$ must support $\alpha$ and therefore include the one element of $X$ coloured red. On the other hand, we clearly must have $(X'_*,\alpha'_*,V_{[f' \circ h_{\alpha'}]}) \nvDash \varphi$, which contradicts the condition $(\star)$.
\end{proof}
\subsection{The supported companion of a functor}

Our solution to the problem we are facing with the monotone neighborhood functor is to perform a gentle repair, changing the functor into a closely related one that does admit an adequate uniform construction. The construction has some independent interest from a general coalgebraic perspective, so we define it in general terms:
\begin{defi} The \textit{supported  companion} $\funco$ of a set functor $\fun$ is the sub-functor of $\psf \times \fun$ defined by:
$$\tilde{\fun}(X) = \{(Z,\alpha) \mid \alpha \in \fun X \; \& \; Z \text{ supports }  \alpha \}$$
\end{defi}
It is easy to check that this is indeed a well-defined subfunctor of $\psf \times \fun$, since for any map $f : X \to Y$, any $\alpha \in \fun X$ and any set $Z \subseteq X$,  the image $f[Z] = \psf f (Z)$ is a support for $\fun f (\alpha)$ whenever $Z$ is a support for $\alpha$. The reader may also note that what we have called ``$\fun$-tree models'' are actually special instances of $\tilde{\fun}$-models, so in a sense the supported companion functor has already played a role in the proof of the main theorem. We will show that this construction repairs the monotone neighborhood functor, so that the supported companion $\underline{\mon}$ of $\mon$ admits an adequate uniform construction. Interestingly enough, the same construction happens to repair weak pullback preservation:
\begin{prop}
The functor $\underline{\mon}$ preserves weak pullbacks. 
\end{prop}
We leave the verification of this to the reader; the argument is similar to the reasoning  in \cite{MartiVenema12} used to establish the existence of a well-behaved relation lifting for $\mon$. 

Note that we have an obvious natural transformation $\pi : \funco \to \psf$ with components given by the projection maps sending $(Z,\alpha) \in \funco X$ to $Z \in \psf (X)$, and we can compose this with the predicate lifting $\Box : \cvp \to \cvp \circ \psf$ to obtain a predicate lifting $\gbox$ for $\funco$ defined by:
$$\gbox := \cvp \pi \circ \Box$$
We call this lifting the \emph{support modality} for $\funco$, and write its dual as $\gdi$.
Furthermore, any lifting $\lambda$ for $\fun$ corresponds to a lifting for $\funco$ by composing $\lambda$ with the projection from $\funco$ to $\fun$ in the same manner, and we will not take care to distinguish $\lambda$ as a lifting for $\fun$ from the corresponding lifting for $\funco$.
 A little issue that we need to  address, before we can proceed to characterize the monotone $\mu$-calculus, is just how the language $\muML_\Lambda$ is related to $\muML_{\Lambda \cup \{\gbox\}}$ for a given set of liftings $\Lambda$ for $\fun$. The rest of this subsection will provide the answer, and give a characterization theorem for $\muML_\Lambda$ as the $\fun$-bisimulation invariant fragment of $\muML_{\Lambda \cup \{\gbox\}}$, where a formula is said to be $\fun$-bisimulation invariant if it has the same truth value in any two pointed $\funco$-models for which the corresponding $\fun$-models are behaviourally equivalent.  Formally we shall write $(\mathbb{S},s) \sim (\mathbb{S}',s')$, for $\funco$-models $\mathbb{S}$ and $\mathbb{S}'$, to say that the respective underlying pointed $\fun$-models are behaviourally equivalent. To distinguish this from actual behavioural equivalence in the sense of the companion functor $\underline{\fun}$, we write $(\mathbb{S},s) \;\underline{\sim} \; (\mathbb{S}',s')$ to say that these pointed models are behaviourally equivalent as $\underline{\fun}$-models.

We borrow a result from \cite{font:auto10}:
\begin{fact}[\cite{font:auto10}]
\label{fmp}
For any set of liftings $\Lambda$ for any set functor $\fun$, the logic $\mu \mathtt{ML}_\Lambda$ has the finite model property. 
\end{fact}
We assume that we have at our disposal a fixed $\funco$-model $ \mathbb{U} = (U,\gamma,V)$ which is a disjoint union of one isomorphic copy for every \emph{finite} pointed $\funco$-model $(\mathbb{S},s)$.  It is not hard to see that such a model does exist: just take a disjoint union containing each $\funco$-model defined on some finite subset $n = \{0,...,n -1\}$ of $\omega$. Since the class of all these models clearly forms a \emph{set}, their disjoint union is a well defined $\funco$-model.

\begin{defi}
Let $\mathbb{S}$ be any finite $\funco$-model.
We define the model $\mathbb{S} \oplus \mathbb{U} = (S + U,\sigma_+,V_+)$ such that, for each $s \in S$ we have $$\sigma_+(s) = (S + U,\fun(\iota_{S,S+U})(\alpha))$$ where $\sigma(s) = (X,\alpha)$, and for $u \in U$ we have $\sigma_+(u) = \funco \iota_{U,S+U}(\gamma(u))$\footnote{In this definition we have used ``$+$'' as the symbol for binary coproducts of sets.}.
\end{defi}
Note that the underlying $\fun$-model of $\mathbb{S} \oplus \mathbb{U}$ is just the co-product of the underlying $\fun$-models of $\mathbb{S}$ and $\mathbb{U}$ respectively.

The following lemmas are immediate from the definition.
\begin{lem}
\label{usim}
Let $\mathbb{S} $ be any $\funco$-model.  Then for all $u \in U$, we have:
$$  (\mathbb{U},u) \gsim ( \mathbb{S} \oplus \mathbb{U} ,u)  $$
\end{lem}


\begin{lem}
\label{applemma}
For any $\fun$-model $\mathbb{S}$ and $s\in S$:
$(\mathbb{S},s) \sim (\mathbb{S} \oplus \mathbb{U},s)$.
\end{lem}
Note that Lemma \ref{applemma} is \emph{not} guaranteed to hold if we replace $\sim$ by the finer equivalence relation $\underline{\sim}$.

It will be convenient in this section to work with a second version of the acceptance game for a modal $\Lambda$-automaton or a modal $\Lambda \cup \{\gbox\}$-automaton, which we will call the \textit{extended} acceptance game for $\mathbb{A}$ with respect to a model $\mathbb{S}$, denoted $\mathcal{E}(\mathbb{A},\mathbb{S})$. Given a modal $\Lambda$-automaton $\mathbb{A} = (A,\Delta,a_I,\Omega)$ and a model $\mathbb{S} = (S,\sigma,V)$ this game has three types of positions: pairs of the form $(\psi,s)$ with $\psi \in \mathtt{1ML}_\Lambda(A)$, pairs of the form $(\psi,s)$ with $\psi \in \mathtt{Latt}(A)$ (lattice formulas over $A$, as defined in Definition \ref{onestepdef}), and maps $f : \mathtt{Latt}(A) \to \psf(S)$. Admissible moves are given in Table \ref{below}:

\begin{table}[h]
    \centering
\begin{tabular}{|l|c|l|l|}
\hline
Position  & Player  &  Admissible moves & Priority
\\ \hline
     $(\psi_1\vee \psi_2,s)$
   & $\exists$  
   & $\{(\psi_1,s),(\psi_2,s) \}$ & $0$
\\ 
     $(\psi_1\wedge \psi_2,s)$  
   & $\forall$  
   & $\{(\psi_1,s),(\psi_2,s) \}$ & $0$
\\ 
     $(a,s) \in A \times S$  
   & $\exists$  
   & $\{(\Delta(a,V^\dagger(s)),s) \}$ & $\Omega(a)$
\\ 
     $(\lambda(\psi_1,...,\psi_n),s)$  
   & $\exists$  
   &
     $\{ f: \mathtt{Latt}(A)\rightarrow{\mathcal{P}}S \mid \sigma(s) \in \lambda(f(\psi_1),...,f(\psi_n))\}$ &  $0$
\\
     $(\top,s)$  
   & $\forall$  
   &
     $\emptyset $ & $0$
\\ 
     $(\bot,s)$  
   & $\exists$  
   &
     $\emptyset$ & $0$
\\ 
     $ f: \mathtt{Latt}(A)\rightarrow{\mathcal{P}}(S) $  
   & $\forall$  
   &
     $\{ (\psi,s) \mid s \in f(\psi)\}$ & $0$ \\
\hline
    \end{tabular}
\caption{ \label{below}
The extended acceptance game 
} 
\end{table}
Here, $\lambda \in \Lambda \cup \{\gbox\}$.
  It is easy to see that the position $(\top,s)$ is always winning for $\exists$, and $(\bot,s)$ is always winning for $\forall$.
The following result is entirely routine to prove:
\begin{thm}
\label{extensive}
Let $\mathbb{A} =  (A,\Delta,a_I,\Omega)$ be a modal $\Lambda$-automaton and $(\mathbb{S},s)$ a pointed $\fun$-model. Then $\mathbb{A}$ accepts $(\mathbb{S},s)$ if, and only if, $(a_I,s)$ is a winning position in the extended acceptance game. 
\end{thm}
We also have the following observation:
\begin{prop}
\label{invariance-extended}
Let $\psi$ be any one-step formula in $\mathtt{1ML}_\Lambda(A)$ for a set of liftings $\Lambda$ for some given set functor $\fun$, and let $\mathsf{P} \in \{\forall , \exists\}$. If $(\mathbb{S},s) \gsim (\mathbb{S}',s')$ then the position $(\psi,s)$ is winning for $\mathsf{P}$ in $\mathcal{E}(\mathbb{A},\mathbb{S})$  if and only if $(\psi,s') $ is a winning position for $\mathsf{P}$ in $\mathcal{E}(\mathbb{A},\mathbb{S}')$.
\end{prop}

This means that for a pair of $\tilde{\fun}$-models $\mathbb{S}$ and $\mathbb{S}'$ such that $(\mathbb{S},s) \gsim (\mathbb{S}',s')$, and a formula $\psi \in \mathtt{ML}_{\Lambda \cup \{\gbox\}}^1(A)$, the position $(\psi,s)$ is winning for $\mathsf{P}$ in $\mathcal{E}(\mathbb{A},\mathbb{S})$  if and only if $(\psi,s') $ is a winning position for $\mathsf{P}$ in $\mathcal{E}(\mathbb{A},\mathbb{S}')$.
\begin{defi}
Let $\bbA = (A,\Delta,a_I,\Om)$ be a $\Lambda \cup \{\gbox\}$-automaton and let $\psi$ be a formula in the range of $\Delta$. We say that $\psi$ is $\bbA$-\emph{valid} if for all (finite) pointed $\fun$-models $(\mathbb{S},s)$, the position $(\psi,s)$ is  winning in $\mathcal{E}(\mathbb{A},\mathbb{S})$. We say that $\psi$ is $\bbA$-\emph{satisfiable} if there is \emph{some} (finite) pointed $\fun$-model $(\mathbb{S},s)$, the position $(\psi,s)$ is  winning in $\mathcal{E}(\mathbb{A},\mathbb{S})$.
\end{defi}

Let $\Lambda$ be any set of predicate liftings for $\fun$ and let $\mathbb{A}$ be any modal $\Lambda \cup \{\gbox\}$-automaton. Making use of our fixed $\tilde{\fun}$-model $\mathbb{U}$ we shall define a translation $t_\mathbb{A} : \mathtt{1ML}_{\Lambda \cup \{\gbox\}}(A) \to \mathtt{1ML}_{\Lambda}(A)$ by induction as follows:
\begin{itemize}
\item For $\lambda \in \Lambda$, we set $t_\mathbb{A}(\lambda( \psi_1,...,\psi_k)) = \lambda(\psi_1,...,\psi_k)$, and similarly for the dual $\lambda^d$ for every lifting $\lambda \in \Lambda$.
\item $t_\mathbb{A}(\top) = \top$ and $t_\mathbb{A}(\bot) = \bot$.
\item  $t_\mathbb{A}(\psi_1 \vee \psi_2) = t_\mathbb{A}(\psi_1) \vee t_\mathbb{A}(\psi_2)$ and $t_\mathbb{A}(\psi_1 \wedge \psi_2) = t_\mathbb{A}(\psi_1) \wedge t_\mathbb{A}(\psi_2)$.
\item $t_\mathbb{A}(\gbox \psi) = \left \{ \begin{array}{ll} \top & \text{ if  $\psi$ is $\bbA$-valid }  \\
\bot & \text{ otherwise } \end{array} \right \} $.
\item $t_\mathbb{A}(\gdi \psi) = \left \{ \begin{array}{ll} \top & \text{ if $\psi$ is $\bbA$-satisfiable}  \\
\bot & \text{ otherwise } \end{array} \right \}$.
\end{itemize}
Note that this translation depends on the whole automaton $\mathbb{A}$, not just the set of variables $A$. So for any given set of variables $A$, we have one translation $t_\mathbb{A} : \mathtt{1ML}_{\Lambda \cup \{\gbox\}}(A) \to \mathtt{1ML}_{\Lambda}(A)$ for each automaton $\mathbb{A}$ with states $A$, and the translations will generally be different for different choices of $\mathbb{A}$. However, we will drop the index $\mathbb{A}$ from now on to simplify notation. Given a modal $\Lambda\cup \{\gbox\}$-automaton $\mathbb{A} = (A,\Delta,a_I,\Omega)$, we will write $t(\mathbb{A})$ for the automaton $(A,\Delta^t,a_I,\Omega)$ where $\Delta^t$ is defined by $\Delta^t(a,c) = t(\Delta(a,c))$, where it is understood that $t$ denotes the translation $t_\mathbb{A}$. Clearly $t(\mathbb{A})$ is a modal $\Lambda$-automaton. 

We shall  view  the translation $t$ as a map defined on the entire domain $\mathtt{1ML}_{\Lambda \cup \{\gbox\}}(A) \cup \mathtt{Latt}(A)$ by  setting $t(\psi) = \psi$ for $\psi \in  \mathtt{Latt}(A)$.

\begin{lem}
\label{mainlemma}
For every finite pointed model $(\mathbb{S},s)$, and for any modal $\Lambda\cup \{\gbox\}$-automaton $\mathbb{A}$, we have
$$(\mathbb{S},s) \Vdash t(\mathbb{A}) \text{ iff } (\mathbb{S} \oplus \mathbb{U},s) \Vdash \mathbb{A}$$
\end{lem}
\begin{proof}

For left to right, suppose that $(\mathbb{S},s) \Vdash t(\mathbb{A})$.  By Theorem \ref{extensive} there is a  strategy $\chi$ for $\exists$ in the extended acceptance game for $t(\mathbb{A})$ and $\mathbb{S}$ which is winning at $(a_I,s)$. Without loss of generality we may assume that $\chi$ is positional, and  a winning strategy at \emph{every} winning position in $\mathcal{E}(t(\mathbb{A}),\mathbb{S})$. Our goal is to construct a positional strategy $\chi'$ for $\exists$ in the extended acceptance game for $\mathbb{A}$ and $\mathbb{S} \oplus \mathbb{U}$, which prescribes a move for $\exists$ at every position $(\psi,v)$ belonging to $\exists$ and with $v\in S$, such that:
\begin{enumerate}
\item $\chi'$ assigns a legitimate move to every position belonging to $\exists$ of the form $(\psi,v)$ with $v \in S$ such that $(t(\psi),v)$ is a winning position in $\mathcal{E}(t(\mathbb{A}),\mathbb{S})$,
\item every  $\chi'$-guided partial match $\rho$ starting at $(a_I,s)$ and ending with a position $(\psi,v)$ satisfies one of the following two criteria: 
\begin{description}
\item[a] $v \in S$ and $t[\rho]$ is a $\chi$-guided match in $\mathcal{E}(t(\mathbb{A}),\mathbb{S})$ (hence consists only of winning positions for $\exists$).
\item[b]  $(\psi,v)$ is a winning position in $\mathcal{E}(\mathbb{A},\mathbb{S} \oplus \mathbb{U})$. 
\end{description}
\end{enumerate}
Here, given that $\rho = \pi_1....\pi_n$, we define $t[\rho] = t(\pi_1)....t(\pi_n)$ where $t(\psi,v) = (t(\psi),v)$ if $\psi \in \mathtt{1ML}_\La(A) \cup \mathtt{Latt}(A)$, and $t(f) = f$ for a position $f : \mathtt{Latt}(A) \to \psf(S)$. Clearly, we can build a winning strategy in $\mathcal{E}(\mathbb{A},\mathbb{S} \oplus \mathbb{U})$ from such a strategy $\chi'$.

We define the strategy $\chi'$ by a case distinction, given a position $(\psi,v)$ belonging to $\exists$ with $v \in S$.
If $\psi = \alpha_1 \vee \alpha_2$ then $t(\psi) = t(\alpha_1) \vee t(\alpha_2)$, so we set $\chi'(\psi,v) = (\alpha_i,v)$ where $i$ is chosen so that $\chi(t(\psi),v) = (t(\alpha_i),v)$. If $\psi = \lambda(\varphi_1,...,\varphi_n)$ then $t(\psi) = \lambda(\varphi_1,...,\varphi_n)$ too, so we set $\chi'(\psi,v) = \chi(t(\psi),v)$. A simple naturality argument shows that this is still a legitimate move in $\mathcal{E}(\mathbb{A},\mathbb{S} \oplus \mathbb{U})$. 

Finally, the interesting case is the one involving the support modality: at a position $(\gdi \varphi,v)$, if $(t(\gdi \varphi),v)$ is winning for $\exists$ in $\mathcal{E}(\mathbb{A},\mathbb{S})$ then we must have $t(\gdi \varphi) = \top$, hence $\varphi$ is $\mathbb{A}$-satisfiable. Hence, by the construction of $\mathbb{U}$,  and by Observation \ref{invariance-extended}  there is some $u \in U$ such that $(\varphi,u)$ is winning for $\exists$ in $\mathcal{E}(\mathbb{A},\mathbb{S} \oplus \mathbb{U})$. So we let the strategy $\chi'$ pick the position $(\varphi,u)$. For the case involving the dual $\gbox$, if $t(\gbox \varphi) = \top$ then $\varphi$ is valid, and we can let $\exists$ pick the map sending $\varphi$ to $S + U$.

We also need to check that every $\chi'$-guided partial match satisfies one of the conditions (a) or (b), and we prove this by an induction on the length of a partial match. The only interesting case is for the extension of a partial match $\rho$ ending with a position $(\gbox \varphi,v)$. By the induction hypothesis on $\rho$, we have $v\in S$ and $(t(\gbox \varphi),v)$ is winning for $\exists$, hence we must have $t(\gbox \varphi) = \top$. This means that $\varphi$ is $\bbA$-valid, so \emph{any} move $(\varphi,w)$ by $\forall$ answering the move $S + U$ by $\exists$ will satisfy the condition (b), i.e. $(\varphi,w)$ is winning for $\exists$ in $\mathcal{E}(\bbA,\mathbb{S}\oplus \mathbb{U})$.
\medskip

For right to left, suppose that  $\exists$ has a winning strategy $\chi$ at the position $(a_I,s)$ in the extended acceptance game for $\mathbb{A}$ with respect to $\mathbb{S} \oplus \mathbb{U}$. We shall give $\exists$ a winning strategy $\chi'$ at the same position in the game for $t(\mathbb{A})$ with respect to $\mathbb{S}$. We shall inductively associate with every $\chi'$-guided partial match $\pi$ of length $k$ a $\chi$-guided ``shadow match'' $(\psi_1,v_1),...,(\psi_k,v_k)$ such that $\pi$ is of the form
$$(t(\psi_1),v_1),....,(t(\psi_k),v_k).$$
We shall also make sure that whenever $\pi'$ is an extension of $\pi$, the shadow match associated with $\pi'$ is an extension of the shadow match associated with $\pi$ as well. It will clearly follow that $\exists$ wins every infinite $\chi'$-guided match.

For the singleton match consisting of $(a_I,s)$ we let $(a_I,s)$ itself be the shadow match. (This is acceptable because, by convention, we have set $t(a_I) = a_I$.) For a match $\pi$ of length $k$ we define the move $\chi'(\pi)$ depending on the shape of the last position on the associated shadow match. Again we treat only the interesting cases. 

If the last position on the shadow match is $(\lambda(\psi_1,...,\psi_m),v)$ then $\chi$ provides a map $f : \mathtt{Latt}(A) \to \cvp(S + U)$ such that $$\sigma_+(v) \in \lambda_{S+ U}(f(\psi_1),...,f(\psi_m)).$$ We set $\chi'(\pi) = f'$, where $f' : \mathtt{Latt}(A) \to \cvp (S)$ is defined by $\theta \mapsto f(\theta) \cap S$ for each $\theta \in \mathtt{Latt}(A)$. This move is legal since $v \in S$ and by naturality of $\lambda$. It is easy to see how to extend the shadow match for each response by $\forall$.

If the last position on the shadow match is $(\gbox\psi,v)$ then since this position is winning for $\exists$, it must be the case that every position $(\psi,w)$ for $w \in S + U$ is winning for $\exists$, hence this holds for every $w \in U$. This can only be true if $\psi$ is $\mathbb{A}$-valid and so we have $t(\gbox \psi) = \top$. This means we are done since $(\top,v)$ is a winning position for $\exists$. Similarly, if the last position on the shadow match is $(\gdi \psi,v)$, then there is some $w \in S + U$ such that $(\psi,w)$ is winning for $\exists$, so $\psi$ is $\mathbb{A}$-satisfiable. Hence $t(\gdi \psi) = \top$, and the conclusion follows as in the previous case.
\end{proof}
We can now prove our characterization theorem:

\begin{thm}
\label{removeboxes}
Let $\fun$ be any set functor and let $\Lambda$ be any set of predicate liftings for $\fun$. Then:
$$ \mu \mathtt{ML}_{\Lambda \cup \{\gbox\}}{/}{\sim} \equiv \mu \mathtt{ML}_\Lambda.$$
\end{thm}

\begin{proof}
Suppose a formula $\varphi$ of $\mu\mathtt{ML}_{\Lambda \cup \{\gbox\}}$ is invariant for $\fun$-bisimilation. By the finite model property for $\mu\mathtt{ML}_{\Lambda \cup \{\gbox\}}$ it suffices to show that $\varphi$ is equivalent to a $\mu\mathtt{ML}_{\Lambda}$-formula over finite models. Let $\mathbb{A}$ be a modal $\Lambda \cup \{\gbox\}$-automaton equivalent to $\varphi$, and let $\psi$ be a formula of $\mu \mathtt{ML}_\Lambda$ equivalent to the automaton $t(\mathbb{A})$.  
Consider an arbitrary  finite  pointed $\tilde{\fun}$-model $(\mathbb{S},s)$. We have:
\begin{displaymath}
\begin{array}{lcll}
 (\mathbb{S},s)\Vdash \varphi & \Leftrightarrow & (\mathbb{S} \oplus \mathbb{U},s) \Vdash \varphi & \text{(Lemma \ref{applemma} + assumption on $\varphi$)}\\
& \Leftrightarrow &  (\mathbb{S} \oplus \mathbb{U},s) \Vdash \mathbb{A} & \\
& \Leftrightarrow & (\mathbb{S},s)\Vdash t(\mathbb{A}) & \text{(Lemma \ref{mainlemma})} \\
& \Leftrightarrow & (\mathbb{S},s) \Vdash \psi & 
\end{array}
\end{displaymath}
as required.
\end{proof}

\subsection{An adequate uniform construction for $\underline{\mon}$}
\label{adguidedonc}
In this section we shall apply Corollary \ref{exp-complete-cor} to derive a characterization result for the monotone $\mu$-calculus. The first fact that we need is the following:
\begin{prop}
The set of liftings $\{\Box,\gbox\}$ is expressively complete with respect to $\underline{\mon}$.
\end{prop}
\begin{proof}
Since $\underline{\mon}$ preserves weak pullbacks, it suffices by Proposition \ref{p:expr-comp-nablas} to show that the one-step language $\mathtt{1ML}_{\{\Box,\gbox\}}$ has the same expressive power as the language based on the  Moss liftings for $\underline{\mon}$. Given a finite set $A$, a subset $\Gamma \subseteq \mathtt{Latt}(A)$ and some $N \in \mon \mathtt{Latt}(A)$ that is supported by $\Gamma$, we want to define a formula $\gamma$ in $\mathtt{1ML}_{\{\Box,\gbox\}}$ such that 
$$X,(X',\alpha),V \Vdash_1 \gamma \text{ iff } ((X',\alpha),(\Gamma,N)) \in \overline{(\underline{\mon})} (\Vdash^0_V)$$
This formula can be obtained as a conjunction $\theta_\Gamma \wedge \theta_N$, where $\theta_\Gamma$ is like the standard nabla-formula for the powerset functor corresponding to $\Gamma$, but expressed in terms of the support modality, and $\theta_N$ is defined as the nabla formulas for the monotone neighborhood functor introduced by Santocanale and Venema in \cite{sant:unif10}. We omit the details.
\end{proof}
The remainder of this section is devoted to defining a uniform construction for the supported companion $\underline{\mon}$ of $\mon$ w.r.t. a given finite set of formulas, and prove its adequacy. We proceed as follows: throughout the section we fix an arbitrary natural number $k$ and a finite set of variables $A$, and define a uniform construction $(-)_*$ which will be shown to be adequate for every formula in $\mathtt{1SO}_{\{\Box,\gbox\}}(A)$ of quantifier depth $\leq k$. It clearly follows  that there is an adequate uniform construction for every finite set  $\Gamma$ of  formulas in $\mathtt{1SO}_{\{\Box,\gbox\}}$, since we can apply the result to the maximum quantifier depth of any formula in $\Gamma$. 

To simplify notation a bit, we note that it suffices to define our uniform construction for what we may call \emph{normal} one-step $\underline{\mon}$-frames $(X,(S,\alpha))$, satisfying the condition that $S = X$, and prove that the adequacy condition holds for all models based on such frames. It is then a straightforward matter to extend this to an adequate uniform construction defined on all $\underline{\mon}$-frames. We shall denote a normal one-step frame $(X,(X,\alpha))$ more simply by just $(X,\alpha)$. Then for any two normal $\underline{\mon}$-models $(X,\alpha,V)$ and $(X',\alpha',V')$ over $A$, i.e. models that are based on normal one-step frames, a relation $R \subseteq X \times X'$ is a one-step $\overline{(\underline{\mon})}$-bisimulation between these models if, and only if:  
\begin{itemize}
\item $u R u'$ implies $u \in V(a)$ iff $u' \in V'(a)$, for all $a\in A$;
\item for all $Z$ in $\alpha$ there is $Z' $ in $\alpha'$ such 
   that for all $t' \in Z'$ there is $t \in Z$ with $t R t'$;
\item for all $Z'$ in $\alpha'$ there is $Z $ in $\alpha$ such
   that for all $t \in Z$ there is $t' \in Z'$ with $t R t'$;
\item for all $u \in X$ there is $u' \in X'$ with $u R u'$;
\item for all $u' \in X'$ there is $u \in X$ with $u R u'$.
\end{itemize}
The last two conditions just say that the relation $R$ is full on $X$ and $X'$. We call such a relation a \emph{global neighborhood bisimulation} between $(X,\alpha,V)$ and $(X',\alpha',V')$.
For any normal one-step frame $(X,\alpha)$, we shall construct a cover $h_\alpha : (X_*,\alpha_*) \to (X,\alpha)$ where $(X_*,\alpha_*)$ is also a normal one-step frame. Since these are abbreviated notations for the one-step $\underline{\mon}$-frames $(X,(X,\alpha))$ and $(X_*,(X_*,\alpha_*))$, we should have $\underline{\mon} h_\alpha (X_*,\alpha_*) = (X,\alpha)$, which means that $h_\alpha$ should be a surjective map from $X_*$ to $X$, such that $\mon h_\alpha (\alpha_*) = \alpha$. The intuition behind the construction is that we want to create infinitely many \emph{disjoint} copies of each neighborhood, and furthermore we want to create sufficiently many copies of each state \emph{within} each copy of a neighborhood.
\begin{defi}
\label{d:F-monstar}
Given a set $X$, and an object $\alpha  \in \mon X$, put
$$
X_* \isdef 
 X \times 2^{k} \times \psf(X) \times \omega
$$
and let $\pi_X$ be the projection map from $X_*$ to $X$.
Define $\alpha_* \in \mon(X_*)$ 
by setting  $Z \in \alpha_*$ for $Z 
\subseteq X_*$ iff $\left\lceil Y,j\right\rceil \subseteq Z$ for some
$Y \in \alpha$ and some $j < \omega$, where
$$\left\lceil Y,j \right\rceil := \{(u,i,Y,j) \mid u \in Y, i < 2^k 
\}.
$$
The sets of the form $\left\lceil Y,j\right\rceil$ for $Y \in \alpha$ will be called the 
\textit{basic members} of $\alpha_*$. The set of all elements in $X_*$ that do not belong to any basic member will be called the \textit{residue} of the frame $(X_*,\alpha_*)$. Note that $X_*$ is partitioned into each of the basic members along with the residue, as an extra partition cell.
\end{defi}
\begin{prop}
For every given normal one-step frame $(X,\alpha)$, the projection map $\pi_X : (X_*,\alpha_*) \to (X,\alpha)$ is a cover.
\end{prop}

\begin{proof}
Clearly the projection map $\pi_X$ is always surjective since $(u,0,\emptyset,0) \mapsto u$ for all $u \in X$. (Note that $0 < 2^n$ for all $n \in \omega$, so we always have $(u,0,\emptyset,0)  \in X_*$.) 

We need to check that $\mon \pi_X(\alpha_*) = \alpha$. In other words, we have to check that for all $Z \subseteq X$, we have $Z \in \alpha$ iff $\pi_X^{-1}(Z) \in \alpha_*$. For left to right, if $Z \in \alpha$ then $\left\lceil Z,0\right\rceil $ is a basic member of $\alpha_*$, and clearly $\left\lceil Z,0\right\rceil \subseteq \pi_X^{-1}(Z)$. Conversely, suppose $\pi_X^{-1}(Z) \in \alpha_*$. Then there is some basic member $\left\lceil Y, j \right \rceil \in \alpha_*$ with $\left\lceil Y, j \right \rceil \subseteq \pi_X^{-1}(Z)$. But then $Y \in \alpha$, and furthermore $Y \subseteq Z$: if $u \in Y$ then $(u,0,Y,j) \in \left\lceil Y, j \right \rceil$, so $(u,0,Y,j)\in \pi_X^{-1}(Z)$ meaning that $\pi_X(u,0,Y,j) = u \in Z$. So $Z \in \alpha$ as required.
\end{proof}

Our goal is to show that, for any pair of normal and globally neighborhood bisimilar models $(X,\alpha,V)$ and $(Y,\beta,U)$, the corresponding models $(X_*,\alpha_*,V_{[\pi_X]})$ and $(Y_*,\beta_*,U_{[\pi_Y]})$ satisfy the same formulas in $\mathtt{1SO}_{\{\Box,\gbox\}}(A)$ of quantifier depth $\leq k$.  It should not be too surprising that we can prove this, but the actual proof is not entirely trivial.  It suffices to check it for formulas in which the operator $\gbox$ and $\gdi$ does not occur, since over normal one-step models this operator is easily definable using the second-order quantifiers. From now on we keep the models $(X,\alpha,V)$ and $(Y,\beta,U)$ fixed, as well as a global neighborhood bisimulation $R$ relating these models. Throughout this section, we use the notation $(X,\alpha,V) \equiv_k (Y,\beta,U)$ to say that two one-step models satisfy the same formulas of $\mathtt{1SO}_{\{\Box,\gbox\}}(A)$ with at most $k$ nested quantifiers. 

\begin{defi}
A propositional $A$-\textit{type} $\tau$ is a subset of $A$. Given a set $X$ and a valuation $V : A \rightarrow \cvp (X)$, the propositional $A$-type of $v \in X$ is defined to be $V^\dagger(v) = \{a \in A \mid v \in V(a)\}$.

Given a subset $Z$  of $X_*$ or $Y_*$, a valuation $W$ such that 
$W : B \rightarrow \cvp(X_*)$ or $W : B \rightarrow \cvp(Y_*)$, and a natural number $m$, the \textit{$m$-signature} of $Z $ over $B$ relative to the valuation $W$ is the mapping $\sigma_Z : \psf(B) \rightarrow \{0,...,m\} $ defined by:
$$\sigma_Z(t) := \mathtt{min}(\vert \{x \in Z \mid W^\dagger(x) = t\}\vert,m)$$
\end{defi}
 
\begin{defi}
 Let $B$ be any set of variables containing $A$, and let $V_1 : B \rightarrow \cvp(X_*)$ and $V_2 : B \rightarrow \cvp(Y_*)$. Then for any natural number $n$ we write  
$$(X_*,\alpha_*,V_1) \approx_{n} (Y_*,\beta_*,V_2) $$
and say that these one-step models \textit{match up to depth $n$}, if:
\begin{enumerate}
\item For every $n$-signature $\sigma$ over variables $B$, the number of basic elements of signature $\sigma$ in $\alpha_*$ and $\beta_*$ respectively are either both finite and equal, or both infinite.
\item The residues of the two one-step models have the same $n$-signature.
\end{enumerate}
\end{defi}

Using the assumption that the models $(X,\alpha,V)$ and $(Y,\beta,U)$ are  globally neighborhood bisimilar, we get the following lemma. 

\begin{lem}
\label{kstep}
 $(X_*,\alpha_*,V_{[ \pi_X]}) \approx_{2^k} (Y_*,\beta_*,U_{ [\pi_Y]})$.
\end{lem}
\begin{proof}
To see that the residues of the two models have the same $2^k$-signature, for one direction just note that if the residue of $(X_*,\alpha_*,V_{[\pi_X]})$ contains an element $(u,i,Z,j)$ then it contains infinitely many elements of the same propositional type, namely one member $(u,i,Z,p)$ for every $p \in \omega$. But then so will the residue of $(Y_*,\beta_*,U_{[\pi_Y]})$: just pick some $v$ with $u R v$ (again using the fact that $R$ is a global neighborhood bisimulation). Since $v \notin \emptyset$, for every $p \in \omega$ the element $(v,0,\emptyset,p)$ will be  a member of the residue of  $(Y_*,\beta_*,U_{[\pi_Y]})$ of the same propositional type as $(u,i,Z,j)$.

Now for the basic elements. First note that, for any $2^k$-signature $\sigma$, $\alpha_*$ either contains no basic elements of signature $\sigma$, or infinitely many: if there is some basic element $\left\lceil Z,j \right\rceil$ of signature $\sigma$, then for any $i \neq j$, the basic element $\left\lceil Z,i \right\rceil$ has the same $2^k$-signature as $\left\lceil Z,j \right\rceil$ with respect to the valuation $V_{[\pi_X]}$. The same holds for $\beta_*$ with respect to the valuation $U_{[\pi_Y]}$.
Hence, it suffices to show that $\alpha_*$ contains a basic element of signature $\sigma$ w.r.t. $V_{[\pi_X]}$ iff $\beta_*$ contains a basic element of signature $\sigma$ w.r.t $U_{[\pi_Y]}$. 

We consider only one direction: suppose that $\alpha_*$ contains a basic element $\left\lceil Z,j\right\rceil$ of signature $\sigma$, where $Z \in \alpha$. Then there must be $Z' \in \beta$ such that, for all $v \in Z'$, there is $u \in Z$ with $u R v$, since $R$ was a global neighborhood bisimulation. Furthermore, since $R$ is full on $X$, for every $w \in Z$ we can pick some $w' \in Y$ with $w R w'$, and put
$$Z'' = Z'\cup \{ w' \mid w \in Z \}$$
By monotonicity we have $Z'' \in \beta$. Furthermore, $Z$ and $Z''$ are clearly related so that the following back-and-forth conditions hold: for all $u \in Z$ there is $v \in Z''$ with $u R v$, and for all $v \in Z''$ there is $u \in Z$ with $u R v$. Since any two states related by $R$ have the same propositional type, it follows that the same propositional types appear in  $\left\lceil Z,j\right\rceil$ and $\left\lceil Z'',j\right\rceil$. But since both these sets contain at least $2^k$ copies of every propositional type that appears in them, it follows that  $\left\lceil Z,j\right\rceil$ and $\left\lceil Z'',j\right\rceil$ have the same $2^k$-signature, as required.
\end{proof}

We are going to show, by induction on a natural number $m \leq k$, that if two one-step models  of the form $(X_*,\alpha_*,V_1)$ and $(Y_*,\beta_*,V_2)$ match up to depth $2^m$, then they satisfy the same formulas of quantifier depth $m$. Together with the previous lemma, it then follows that the one-step models $(X_*,\alpha_*,V_{[\pi_X]})$  and $(Y_*,\beta_*,U_{[\pi_Y]})$ satisfy the same formulas of quantifier depth $\leq k$. For the basis case of $2^0 = 1$, we need the following result (we recall that the case of atomic formulas $\gbox a$ can be safely ignored):
\begin{lem}
\label{atomicstep}
Let $B$ be a set of variables containing $A$, and let $V_1 : B \rightarrow \cvp(X_*)$ and $V_2 : B \rightarrow \cvp(Y_*)$ be valuations such that
$$(X_*,\alpha_*,V_1) \approx_{1} (Y_*,\beta_*,V_2) $$
Then these two one-step models satisfy the same atomic formulas of the one-step
language $\mathtt{1SO}_{\{\Box\}}(B)$.
\end{lem}
\begin{proof}
We only prove one direction for each case. 

Suppose first that
$$(X_*,\alpha_*,V_1)\Vdash_1 p \subseteq q$$
where $p,q \in B$. Suppose for a contradiction that $V_2(p) \nsubseteq V_2(q)$. Then there is some $(u,i,Z,j) \in Y_*$ such that $(u,i,Z,j) \in V_2(p)\setminus V_2(q)$. If $(u,i,Z,j)$ comes from the residue of $(Y_*,\beta_*)$ then since the residue of $(X_*,V_*)$ has the same $1$-signature, it must contain some element $(u',i',Z',j')$ of the same $1$-type, and so we cannot have $V_1(p) \subseteq V_1(p)$. The case where $(u,i,Z,j)$ comes from a basic member is similar. 

Now, suppose that 
$$ (X_*,\alpha_*,V_1)\Vdash_1 \Box p$$
Then $V_1(p) \in \alpha_*$, so there is some basic element $\left\lceil Z,j\right\rceil \in \alpha_*$ with $\left\lceil Z,j\right\rceil \subseteq V_1(p)$. There must be some basic $\left\lceil Z',j'\right\rceil \in \beta_*$ of the same $1$-signature over $B$ as $\left\lceil Z,j\right\rceil$, and clearly it follows that $\left\lceil Z',j'\right\rceil \subseteq V_2(p)$ and so $V_2(p) \in \beta_*$ as required. 
\end{proof}

We now only need the following lemma:
\begin{lem}
\label{mainlemmamonstar}
Let $B$ be a finite set of variables containing $A$, let $0 < m \leq k$ and let $V_1 : B \rightarrow Q(X_*)$ and $V_2 : B \rightarrow \cvp (Y_*)$ be valuations such that
$$(X_*,\alpha_*,V_1) \approx_{2^m} (Y_*,\beta_*,V_2)$$
Let $q$ be any fresh variable. Then for any valuation
$V_1'$ over $B \cup \{q\}$
extending $V_1$ with some value for $q$, there exists a valuation
$V_2' $ over $B \cup \{q\}$
extending $V_2$, such that
$$(X_*,\alpha_*,V_1') \approx_{2^{(m-1)}} (Y_*,\beta_*,V_2')$$
and vice versa. 
\end{lem}

\begin{proof}
We only prove one direction since the other direction can be proved by a symmetric argument. 

Let $V_1'$ be given. By the hypothesis, for any $2^m$-signature $\sigma$ over the variables $B$,  the number of basic elements of signature $\sigma$ in $\alpha_*$ and $\beta_*$ relative to $V_1$ and $V_2$ are either both finite and the same, or both infinite. Let $\sigma_1,...,\sigma_k$ be a list of all the distinct $2^m$-signatures over $B$ such that the set of basic elements of $\alpha_*$ and $\beta_*$ of signature $\sigma_i$, with $1 \leq i \leq k$, is non-empty but finite, and let $\sigma_{k+1},...,\sigma_{l}$ be a list of all the $2^m$-signatures such that, for $k + 1 \leq i \leq l$, there are infinitely many basic elements of $\alpha_*$ and of $\beta_*$ of signature $\sigma_i$. Then, for each $i \in \{1,...,l\}$, let $\alpha_*[\sigma_i]$ denote the set of basic elements in $\alpha_*$ of signature $\sigma_i$, and similarly let $\beta_*[\sigma_i]$ denote the set of basic elements of $\beta_*$ of signature $\sigma_i$. Then $\alpha_*[\sigma_1],...,\alpha_*[\sigma_l]$ is a partition of the set of basic elements of $\alpha_*$ into non-empty cells, and similarly $\beta_*[\sigma_1],...,\beta_*[\sigma_l]$ is a partition of the set of basic elements of $\beta_*$.

Given the extended valuation $V_1'$ on $X_*$ defined on variables $B \cup \{q\}$, we similarly let $\tau_1,...,\tau_{k'}$ be a list of all the $2^{m - 1}$-signatures over $B \cup \{q\}$ such that, for $1 \leq i \leq k'$, the set of basic elements of $\alpha_*$ of $2^{m - 1}$-signature $\tau_i$ is non-empty but finite. We let $\tau_{k' + 1},...,\tau_{l'}$ be a list of all the $2^{m - 1}$-signatures over $B \cup \{q\}$ such that, for each $i$ with $k' + 1 \leq i \leq l'$, the set of basic elements of $\alpha_*$ of $2^{m - 1}$-signature $\tau_i$ is infinite. Let $\alpha_*[\tau_i]$ denote the set of basic elements of $\alpha_*$ of $2^{m - 1}$-signature $\tau_i$, so that the collection $\alpha_*[\tau_1],...,\alpha_*[\tau_{l'}]$ constitutes a second partition of the set of basic elements of $\alpha_*$. It will be useful to introduce the abbreviation $D_1$ for the finite set $\alpha_*[\sigma_1] \cup...\cup \alpha_*[\sigma_k]$, and the abbreviation $D_2$ for the finite set $\alpha_*[\tau_1] \cup...\cup \alpha_*[\tau_{k'}]$. 

For each $i$ with $1 \leq i \leq k$, there is a bijection between the set $\alpha_*[\sigma_i]$ and $\beta_*[\sigma_i]$, and we can paste all these bijections together into a bijective map 
$$f : \alpha_*[\sigma_1] \cup...\cup \alpha_*[\sigma_k] \rightarrow \beta_*[\sigma_1] \cup ... \cup \beta_*[\sigma_k]$$ 
Since every basic element of $\alpha_*$ not in $D_1$ belongs to a $2^m$-signature of which there are infinitely many basic elements in $\beta_*$, and since $D_1 \cup D_2$ is finite, it is easy to see that we can extend the map $f$ to a map $g$ which is an injection from the set $D_1 \cup D_2$
into the set of basic elements of $\beta_*$, such that for each basic element $\left\lceil Z,j\right\rceil$ in $D_1 \cup D_2$, $\left\lceil Z,j\right\rceil$ and $g(\left\lceil Z,j\right\rceil)$ have the same $2^m$-signature over $B$, and such that $g\!\upharpoonright\!D_1 = f$. Each basic element of $\beta_*$ not in the image of $g$ must then be of one of the $2^m$-signatures $\sigma_{k + 1},...,\sigma_{l}$, and so we can partition the set of basic elements of $\beta_*$ outside the image of $g$ into the cells
$\beta_*[\sigma_{k + 1}]\setminus g[D_2],...,\beta_*[\sigma_l]\setminus g[D_2]$.
For each $i$ with $k + 1 \leq i \leq l$, let $\gamma^i_{1},...,\gamma^i_{r}$ list all infinite sets of the form 
$\alpha_*[\sigma_i] \cap \alpha_*[\tau_j]$ for $k' + 1 \leq j \leq l'$.
The list $\gamma^i_{1},...,\gamma^i_{r}$ must be non-empty, and so since the set $\beta_*[\sigma_{i}]\setminus g[D_2]$ is also infinite, we may 
partition it into $r$ many infinite cells and list these as $\delta^i_{1},...,\delta^i_{r}$.
Now, for each basic element $\left\lceil Z,j\right\rceil$ of $\beta_*$, we 
define a map 
$W_{\left\lceil Z,j\right\rceil}$ from  $B \cup \{q\}$ to 
  $\psf(\left\lceil Z,j\right\rceil)$
by a case distinction as follows:

\emph{Case 1:} $\left\lceil Z,j\right\rceil = g(\left\lceil Z',j'\right\rceil)$ 
for some $\left\lceil Z',j'\right\rceil \in D_1 \cup D_2$. 
Then $\left\lceil Z,j\right\rceil$ and $\left\lceil Z',j'\right\rceil$ have the
same $2^m$-signature over $B$. Using this fact we define the valuation 
$W_{\left\lceil Z,j\right\rceil}$ so that, for each $p \in B$, we have 
$W_{\left\lceil Z,j\right\rceil}(p) = V_2(p) \cap \left\lceil Z,j\right\rceil$,
and so that $\left\lceil Z',j'\right\rceil$ and $\left\lceil Z,j\right\rceil$ 
have the same $2^{m - 1}$-signature over $B \cup \{q\}$ with respect to the
valuations $V_1'$ and $W_{\left\lceil Z,j\right\rceil}$. 
 We show how to assign the value of the variable $q$: for each propositional type $t$ over $B \cup \{q\}$, there are three different possible cases to consider. If $\left\lceil Z',j'\right\rceil$ has $m' < 2^{m-1}$ elements of type $t \cup \{q\}$ over $B \cup \{q\}$, then pick $m'$ many elements of $\left\lceil Z,j\right\rceil$ of type $t$ and mark them by $q$. This is possible since $m' < 2^{m-1} \leq 2^m$ and $\left\lceil Z',j'\right\rceil$ and $\left\lceil Z,j\right\rceil$ have the same $2^m$-signature. If there are $m' < 2^{m-1}$ elements of $\left\lceil Z',j'\right\rceil$ of type $t$ over $B \cup \{q\}$, then pick $m'$ elements of $\left\lceil Z,j\right\rceil$ of type $t$ over $B$, and mark all the other elements of $\left\lceil Z,j\right\rceil$ of type $t$ by $q$. Finally, if there are at least $2^{m-1}$ elements of $\left\lceil Z',j'\right\rceil$ of type $t \cup \{q\}$ over $B \cup \{q\}$ and at least $2^{m - 1}$ elements of $\left\lceil Z',j'\right\rceil$ of type $t$ over $B \cup \{q\}$, then all in all there must be at least $2^m$ elements of $\left\lceil Z',j'\right\rceil$ of type $t$ over $B$, and so there must be at least $2^m$ elements of $\left\lceil Z,j\right\rceil$ of type $t$ over $B$. Pick $2^{m-1}$ of these and mark them by $q$. Finally, let $W_{\left\lceil Z,j\right\rceil}(q)$ be the set of elements of $\left\lceil Z,j\right\rceil$ marked by $q$.

\emph{Case 2:} $\left\lceil Z,j\right\rceil$ is not in the image of $g$. Then there must be some $i \in \{k' + 1,...,l'\}$ such that $\left\lceil Z,j\right\rceil \in \beta_*[\sigma_{k + 1}]\setminus g[D_2]$, and this set is partitioned into $\delta^i_1,...,\delta^i_r$. Let $\left\lceil Z,j\right\rceil \in \delta^i_j$, and pick some arbitary element $\left\lceil Z',j'\right\rceil$ of the set $\gamma^i_j$. Then $\left\lceil Z',j'\right\rceil$ and $\left\lceil Z,j\right\rceil$ have the same $2^m$-signature over $B$ and we can proceed as in Case 1.

We define the valuation $V_2'$ so that the intersection of $V_2'(q)$ with the union of all the basic members of $\beta_*$ equals the union of the sets $W_{\left\lceil Z,j\right\rceil}(q)$ for $\left\lceil Z,j\right\rceil$ a basic element in $\beta_*$, and so that the residue of $(Y_*,\beta_*)$ has the same $2^{(m - 1)}$-signature as the residue of $(X_*,\alpha_*)$ with respect to the valuations $V_1'$ and $V_2'$. This can be done using the same reasoning as in the two previous cases. We now need to check that
$$(X_*,\alpha_*,V_1') \approx^{2^{(m-1)}} (Y_*,\beta_*,V_2')$$  
First, suppose there are infinitely many basic elements of $\alpha_*$ of some $2^{m-1}$ signature $\tau_j$, meaning that $k' \leq j \leq l'$. Then since the set $\alpha_*[\tau_j]$ is infinite, $D_1$ is finite and 
$$\alpha_*[\tau_j] = (D_1 \cap \alpha_*[\tau_j]) \cup (\alpha_*[\sigma_{k +1}] \cap \alpha_*[\tau_{j}]) \cup ... \cup (\alpha_*[\sigma_l] \cap \alpha_*[\tau_{j}])$$
 there must be some $i \in \{k+1,...,l\}$ such that the set $\alpha_*[\sigma_i] \cap \alpha_*[\tau_j]$ is infinite. This means that $\alpha_*[\sigma_i] \cap \alpha_*[\tau_j]$ appears in the list $\gamma^i_1,...,\gamma^i_r$, and so we see that all elements of some member of the list $\delta^i_1,...,\delta^i_r$ will have the $2^{m-1}$-signature $\tau_j$. Since each member of this list is infinite, we see that there must be infinitely many basic elements of $\beta_*$ of signature $\tau_j$. 

Conversely, suppose there are infinitely many basic elements of $\beta_*$ of $2^{m-1}$-signature $\tau_j$ over $B \cup \{q\}$. Then since the image of $g$ is finite, some of these elements must be outside the image of $g$, which means that for some $i \in \{k + 1,...,l\}$, some member of the list $\delta^i_1,...,\delta^i_r$ will consist of elements of signature $\tau_j$. This means that some member of the list $\gamma_1^i,...,\gamma_r^i$ will consist of elements of signature $\tau_j$, and since each member of this list is infinite we see that $\alpha_*$ has infinitely many basic elements of $2^{m-1}$-signature $\tau_j$ over $B \cup \{q\}$.

Finally, suppose that there are finitely many basic elements of $\alpha_*$ and $\beta_*$ of $2^{m-1}$-signature $\tau_j$. We check that the mapping $g$ restricts to a bijection between the basic elements of $\alpha_*$ and $\beta_*$ of this signature. First, $g$ is injective and maps basic elements of $\alpha_*$ of signature $\tau_j$ to basic elements of $\beta_*$ of signature $\tau_j$. It only remains to show that (the restriction of) $g$ is surjective, i.e. each basic element $\left\lceil Z,r\right\rceil$ of signature $\tau_j$ is equal to $g(\left\lceil Z',r'\right\rceil)$ for some $\left\lceil Z',r'\right\rceil$. But suppose $\left\lceil Z,r\right\rceil$ is not in the image of $g$; then it is in one of the members of the list $\delta_1^i,...,\delta_r^i$ for some $i$, and since each of these members is an infinite set of basic elements of the same signature, we see that there are infinitely many basic elements of $\beta_*$ of signature $\tau_j$, contrary to our assumption. Hence, the proof is done. 
\end{proof}

\begin{thm}
The monotone $\mu$-calculus is the neighborhood bisimulation invariant fragment of monotone monadic second-order logic.  In a formula:
$$\mathtt{MMSO}/{\sim} \equiv \mu \mathtt{MML}$$
\end{thm} 

\begin{proof}
It suffices to prove that $\mathtt{MSO}_{\{\Box,\gbox\}}{/}{\sim}$ is equal to $\mu \mathtt{ML}_{\{\Box,\gbox\}}$, and then apply Theorem \ref{removeboxes}. For this, by Theorem \ref{mainone} it suffices in turn to prove that the construction $(-)_*$ is adequate for all formulas in  $\mathtt{1SO}_{\{\Box,\gbox\}}(A)$ of quantifier depth $\leq k$. We can prove this by combining the last three lemmas, using Ehrenfeucht-Fraïssé games for the one-step language. Lemmas \ref{kstep} and \ref{mainlemmamonstar} provide a recipe for how ``Duplicator'' can survive $k$ steps of the game comparing the models $(X_*,\alpha_*,V_{[ \pi_X]})$ and  $(Y_*,\beta_*,U_{ [\pi_Y]})$, and  \ref{atomicstep} guarantees that the valuations constructed at the end of the game will satisfy the same atomic formulas. Working out the full argument is entirely standard, so we leave the details to the reader.  
\end{proof}

\section{Future work}

We conclude by mentioning some questions for future research:
\begin{enumerate}
\item
Is there a good \emph{categorical} and ``logic free'' characterization of those set functors $\fun$ that
admit an adequate (or weakly adequate) uniform construction, for instance, in terms of $\fun$ 
preserving certain limits or colimits? Trying to answer this question would involve a deeper study of how the model theory of one-step languages is related to categorical properties of the type functor involved. Related to this, an anonymous referee pointed out to us that the machinery of one-step frames, covers and uniform constructions are all taking place in the \emph{category of elements} associated with the set functor $\fun$, and properties of the functor $\fun$ are closely related to properties of the corresponding category of elements \cite{barr:cate90}. This could very well be a fruitful direction to investigate further. 
\item
Can we improve our work on the supported companion functor, to the effect that every
set functor $\fun$ has a companion $\fun '$ that admits an adequate uniform 
construction?
Relating this to the previous question, we would like to understand \emph{why} the supported companion to
$\mon$ admits an adequate uniform construction, but not $\mon$ itself. The construction achieves two things, in general: first, it ensures that every $\alpha \in \underline{\fun} X$ has a unique smallest support (even when $X$ is infinite), often called its \emph{base}. Second, and in our view more importantly, we get that the map $\mathsf{base}_X : \underline{\fun} X \to \psf X$ defined by sending each $\alpha$ to its base is a natural transformation. Conditions for a functor under which this holds has been isolated by Gumm in \cite{gumm:tcoa05}. Are these conditions related to the existence of adequate uniform constructions?
\item
It would be interesting to further explore the relation between $\MSOT$ and the first-order 
logic of Litak \& alii~\cite{lita:coal12} for $\fun$-coalgebras.
For instance, an interesting question would be whether (on $\fun$-tree models)
$\MSOT$ is equivalent to some extension of this first-order language with 
 fixpoint operators.
\item Finally, there is the question of finding sufficient \emph{and necessary} conditions for a Janin-Walukiewicz theorem to hold. Related to this question, we have not been able to produce an example of a functor for which the Janin-Walukiewicz theorem does \emph{not} hold (the question of whether such an example can be found was raised to us by an anonymous referee). Given all that we can say for sure, it could be the case that a Janin-Walukiewicz theorem simply holds for \emph{every} set functor, but we conjecture that at least some conditions on the functor are required.   
\end{enumerate}

 \subsection*{Acknowledgement} We wish to thank the anonymous referees for their very thorough and constructive feedback, which led to many substantial improvements on the original draft of this paper.

\bibliographystyle{IEEEtran}


\end{document}